\documentclass[11pt]{article}
\usepackage[letterpaper, margin=1.0in]{geometry}
\usepackage[utf8]{inputenc}

\date{}
\usepackage{hyperref}
\usepackage{amsmath,amssymb,amsfonts,amsthm}
\usepackage{mathtools}
\usepackage{enumerate}
\usepackage{color}
\usepackage{xspace}
\usepackage{soul}
\usepackage{tikz}
\usetikzlibrary{shapes,shapes.callouts,shadows,arrows,backgrounds,%
matrix,patterns,arrows,decorations.pathmorphing,decorations.pathreplacing,%
positioning,fit,calc,decorations.text%
}

\usepackage{cases}
\usepackage{multirow}
\usepackage{todonotes}
\usepackage{caption,subcaption}

\theoremstyle{plain}
  \newtheorem{theorem}{Theorem}
  \newtheorem{corollary}[theorem]{Corollary}
  \newtheorem{observation}[theorem]{Observation}
  \newtheorem{lemma}[theorem]{Lemma}
  \newtheorem{proposition}[theorem]{Proposition}
  \newtheorem{claim}{Claim}[theorem]
  \newcounter{theoremspecial}
  \newtheorem{claimspecial}{Claim}[theoremspecial]
  \newtheorem*{theorem*}{Theorem}
  \newtheorem*{corollary*}{Corollary}
  \newtheorem*{lemma*}{Lemma}
  \newtheorem*{proposition*}{Proposition}
  \newtheorem*{claim*}{Claim}
  
\theoremstyle{definition}
  \newtheorem{definition}[theorem]{Definition}
  \newtheorem{example}[theorem]{Example}

  \newtheorem{property}[theorem]{Property}
  \newtheorem*{definition*}{Definition}
  \newtheorem*{example*}{Example}
  \newtheorem*{question*}{Question}
  \newtheorem*{conjecture*}{Conjecture}
  \newtheorem*{remark*}{Remark}

\newenvironment{claimproof}
{\noindent {\em Proof of claim:} }
{\hfill $\diamond$ \smallskip}

\DeclarePairedDelimiter\abs{\lvert}{\rvert}
\DeclarePairedDelimiter\norm{\lVert}{\rVert}

\mathchardef\mhyphen="2D

\newcommand*{\integers}{\mathbb{Z}}
\newcommand*{\rationals}{\mathbb{Q}}
\newcommand*{\naturals}{\mathbb{N}}
\newcommand*{\range}[2]{\ensuremath{\{#1,\dots,#2\}}}

\newcommand{\bigoh}{\mathcal{O}}
\newcommand{\cc}[1]{{\mbox{\textnormal{\textsf{#1}}}}\xspace}

\renewcommand{\P}{\cc{P}}
\newcommand{\NP}{\cc{NP}}
\newcommand{\FPT}{\cc{FPT}}

\newcommand{\Weft}{{\cc{W}}}
\newcommand{\W}[1]{{\Weft}{{[#1]}}}

\DeclareMathOperator\supp{supp}
\DeclareMathOperator\lcm{lcm}

\usepackage{boxedminipage}
\newcommand{\pbDef}[3]{
  \noindent
  \begin{center}
  \begin{boxedminipage}{0.98 \columnwidth}
  {\sc #1}\\[5pt]
  \begin{tabular}{l p{0.70 \columnwidth}}
  {\sc Instance}: & #2\\
  {\sc Question}: & #3
  \end{tabular}
  \end{boxedminipage}
  \end{center}
}

\newcommand{\pbDefP}[4]{%
  \noindent
  \begin{center}
  \begin{boxedminipage}{0.98 \columnwidth}
  {\sc #1}\\[5pt]
  \begin{tabular}{l p{0.70 \columnwidth}}
  {\sc Instance}: & #2\\
  {\sc Parameter:} & #3\\
  {\sc Question}: & #4
  \end{tabular}
  \end{boxedminipage}
  \end{center}
}

\newcommand*{\csp}[1]{\ensuremath{\textsc{CSP}(#1)}}
\newcommand*{\mincsp}[1]{\ensuremath{\textsc{MinCSP}(#1)}}
\newcommand*{\lin}[2]{\textsc{\ensuremath{#1}-Lin(\ensuremath{#2})}}
\newcommand*{\minlin}[2]{\textsc{Min-\ensuremath{#1}-Lin(\ensuremath{#2})}}
\newcommand*{\mintwolin}{\textsc{Min-2-Lin}}

\newcommand*{\twolin}{\textsc{2-Lin}}

\newcommand*{\F}{\mathbb{F}\xspace}
\newcommand*{\G}{\mathbb{G}\xspace}
\newcommand*{\D}{\mathbb{D}\xspace}
\newcommand*{\K}{\mathbb{K}\xspace}

\newcommand*{\B}{\mathcal{B}\xspace}
\newcommand*{\cG}{\mathcal{G}\xspace}

\newcommand*{\cH}{\mathcal{H}\xspace}
\newcommand*{\cF}{\mathcal{F}\xspace}

\newcommand*{\cP}{\mathcal{P}\xspace}

\newcommand*{\STS}{\textnormal{star}}

\newcommand*{\DML}{\textnormal{DML}\xspace}
\newcommand*{\BGC}{\textnormal{BGC}\xspace}

\newcommand*{\RBGC}{\textnormal{RBGC}\xspace}
\newcommand*{\RBGCE}{\textnormal{RBGCE}\xspace}
\newcommand*{\PPC}{\textsc{Pair Partition Cut}\xspace}

\title{Almost Consistent Systems of Linear Equations}
\author{
Konrad K. Dabrowski\thanks{Newcastle University, UK, \texttt{konrad.dabrowski@newcastle.ac.uk}} \and
Peter Jonsson\thanks{Link{\"o}ping University, Sweden, \texttt{peter.jonsson@liu.se}} \and
Sebastian Ordyniak\thanks{University of Leeds, UK, \texttt{sordyniak@gmail.com}} \and
George Osipov\thanks{Link{\"o}ping University, Sweden, \texttt{george.osipov@pm.me}} \and
Magnus Wahlstr{\"o}m\thanks{Royal Holloway, University of London, UK, \texttt{Magnus.Wahlstrom@rhul.ac.uk}}
}
\date{\today}

\begin{document}

\maketitle

\begin{abstract}
Checking whether a system of linear equations is consistent is a basic computational problem with ubiquitous applications.
When dealing with inconsistent systems, one may seek an assignment that minimizes the number of unsatisfied equations.
This problem is NP-hard and UGC-hard to approximate within any constant even for two-variable equations over the two-element field.
We study this problem from the point of view of parameterized complexity, with the parameter being the number of unsatisfied equations. 
We consider equations defined over Euclidean domains---a family of commutative rings that generalize finite and infinite fields including the rationals, the ring of integers, and many other structures.
We show that if every equation contains at most two variables, the problem is fixed-parameter tractable.
This generalizes many eminent graph separation problems such as Bipartization, Multiway Cut and Multicut parameterized by the size of the cutset. 
To complement this, we show that the problem is W[1]-hard when three or more variables are allowed in an equation, as well as for many commutative rings that are not Euclidean domains.
On the technical side, we introduce the notion of important balanced subgraphs, 
generalizing important separators of Marx [Theor. Comput. Sci. 2006] to the setting of biased graphs.
Furthermore, we use recent results on parameterized MinCSP [Kim et al.,
SODA 2021] to efficiently solve a generalization of Multicut with
disjunctive cut requests.
\end{abstract}

\newpage

\tableofcontents

\newpage

\section{Introduction}

Algorithms for systems of linear equations 
have been studied since ancient times~\cite{grcar2011ordinary}.
As H{\aa}stad~\cite{haastad2001some} aptly remarks, for computer science
``[t]his problem is in many regards as basic as satisfiability''.
Well-known methods like Gaussian elimination can recognize
and solve consistent systems of equations.
However, these methods are not well suited for dealing with inconsistent systems.
In the optimization version of the problem called 
\textsc{MaxLin} one seeks an assignment 
maximizing the number of satisfied equations.
In its dual, called \textsc{MinLin}, the objective is to minimize
the number of unsatisfied equations.
Both \textsc{MaxLin} and \textsc{MinLin} remain \NP-hard in severely
restricted settings, which has motivated an extensive study
of approximation algorithms for these problems.
However, the problems resist approximation within reasonable bounds:
in particular, \textsc{MinLin} over the two-element field
restricted to equations with at most two variables
is not approximable within any constant factor under the 
Unique Games Conjecture
(UGC)---in fact, it has been suggested that
constant-factor inappoximability of this version of \textsc{MinLin}
may be equivalent to UGC~(see Definition~3~in~\cite{khot2016candidate} and the following discussion).
This motivates exploring other approaches to resolving inconsistent systems.

Crowston et al.~\cite{crowston2013parameterized} initiated the study of the parameterized complexity of 
\textsc{MinLin} with the parameter $k$ being the number of unsatisfied equations.
They focus on systems over the two-element field and prove that 
when every equation has at most two variables, the problem admits a 
$\bigoh^*(2^k)$~\footnote{$\bigoh^*$ hides polynomial factors in the bit-size of the instance.}  
algorithm, while allowing three or more variables in an equation leads to \W{1}-hardness.
This rules out the existence of an algorithm for this problem running in $\bigoh^*(f(k))$ time
for any computable function $f$ under the standard assumption $\FPT \neq \W{1}$.
In this paper we substantially extend the study of the parameterized complexity of \textsc{MinLin}
by considering equations over commutative rings.
Thus, we study \emph{Euclidean domains}, 
which include the finite fields $\F_q$, infinite fields such as the rationals $\rationals$, 
 the ring of integers $\integers$, the Gaussian integers $\integers[i]$,
 the ring of polynomials~$\F[x]$ over a field~$\F$, and many more structures.
Perhaps unsurprisingly, we show that with 
three or more variables per equation,
the problem over Euclidean domains is \W{1}-hard (in fact, the hardness proof only uses coefficients $0$, $1$ and $-1$,
so the result holds for equations over any abelian group).
On the other hand, $\mintwolin$, where each equation has two variables 
turns out to be much more interesting: the problem is fixed-parameter
tractable for every effective Euclidean domain, i.e. those that admit representations such that 
the basic operations are
polynomial-time computable and multiplication is well behaved
(see Property~\ref{property:edom-product} for details).
Note that asking about the parameterized complexity of $\mintwolin$ over a domain
only makes sense if the problem of checking consistency of a system is not \NP-hard
(otherwise, the problem is intractable even for $k=0$).
This is where the effectiveness of Euclidean domains becomes important.
To the best of our knowledge, there are no published algorithms for solving systems of equations
over Euclidean domains in the literature even for the special case with at most two variables per equation.
Thus, we develop methods for checking consistency of such systems in Section~\ref{sec:edom-algorithm}.
These methods form the underpinning of our fpt algorithm for \mintwolin{} over Euclidean domains.

\begin{table}
    \centering
    \begin{tabular}{| l | l | l | l |}
    \hline
    Problem & Solution & Method & Reduces to \\
    \hline 
    \hline
    \textsc{Bipartization} & \cite{reed2004finding} in 2004 & Iterative compression & $\minlin{2}{\F_2}$ \\
    \hline
    \textsc{$q$-Multiway Cut} & \cite{marx2006parameterized} in 2006 & Important separators (IS) & \minlin{2}{\F_q} \\
    \hline
    \textsc{Multiway Cut} & \cite{marx2006parameterized} in 2006 & Important separators & \minlin{2}{\rationals} \\
    \hline
    \textsc{Multicut} & \cite{bousquet2018multicut}, \cite{marx2014fixed}$^\dag$ in 2011 & $^\dag$Random sampling of IS & \minlin{2}{\integers} \\
    \hline
    \end{tabular}
    \caption{Graph separation problems related to $\mintwolin$.}
    \label{tab:problems}
\end{table}

\medskip \noindent \textbf{Background.}
We start with a few basic definitions.
Let $\D=(D;+,\cdot)$ denote a commutative ring.
An expression $c_1 \cdot x_1+\dots+c_r \cdot x_r=c$ is a {\em (linear) equation over} $\D$ if $c_1,\dots,c_r,c \in D$ and $x_1,\dots,x_r$ are variables with domain $D$.
Let $S$ denote a set (or equivalently a system) of equations over $\D$.
We let $V(S)$ denote the variables appearing in $S$, and we
say that $S$ is {\em consistent} if there is an assignment
$\varphi : V(S) \rightarrow D$
that satisfies all equations in $S$. 
An instance of the computational problem $\lin{r}{\D}$
is a system $S$ of equations in $r$ variables
over $\D$, and the question is whether $S$ is consistent.
To assign positive integer weights to the elements of any set $Y$, 
we use a weight function $w : Y \rightarrow \naturals^+$
and write $w(X)$ for any subset $X \subseteq Y$
as a shorthand for $\sum_{e \in X} w(e)$.
The following is the main computational problem in this paper.

\pbDefP{$\minlin{r}{\D}$}
{An instance $S$ of $\lin{r}{\D}$, a weight function 
$w : S \rightarrow \naturals^+$, and an integer $k$.}
{$k$.}
{Is there a set $Z \subseteq S$ such that $S - Z$ is consistent and 
$w(Z) \leq k$?}

Crowston et al.~\cite{crowston2013parameterized} studied
the problem $\minlin{r}{\F_2}$ and prove that
$\minlin{2}{\F_2}$ is in \FPT.
However, their methods do not seem sufficient to solve $\mintwolin$ over structures larger than $\F_2$.
As a possible explanation and additional motivation,
we note that $\mintwolin$ over $\F_2$, $\F_q$, $\rationals$ and $\integers$
generalize well-known graph separation problems that have served as milestones
for the development of parameterized algorithms:
these are \textsc{Bipartization},
\textsc{$q$-Terminal Multiway Cut}, 
\textsc{(General) Multiway Cut} and 
\textsc{Multicut}, respectively
(see Table~\ref{tab:problems} for a short summary of the progress).

In \textsc{Bipartization}, given a graph $G$ and an integer $k$, 
the goal is to remove at most $k$ edges from the graph to make it bipartite.
To reduce to $\minlin{2}{\F_2}$, create a variable for every vertex and
add an equation $x - y = 1$ for every edge $\{x,y\}$ in $G$.
The parameterized complexity status of this problem was resolved
by Reed~et~al.~\cite{reed2004finding} using the newly introduced method of
\emph{iterative compression}, which has since become a common opening
of fpt algorithms including those presented in this paper 
(see Chapter~4~in~\cite{cygan2015book} for many more examples).

In \textsc{$q$-Terminal Multiway Cut}, 
given a graph $G$, a set of $q$ vertices $t_1,\dots,t_q$ called terminals, 
and an integer $k$,
the goal is to remove at most $k$ edges from $G$ 
to separate the terminals into distinct connected components.
The problem is in \P for $q = 2$ and \NP-hard for $q \geq 3$.
The reduction to $\minlin{2}{\F_q}$ works 
by introducing an equality $x = y$ for every edge $\{x,y\}$ in $G$,
and assigning a distinct field element $\alpha_i$ to every terminal $t_i$
by adding equation $t_i = \alpha_i$
with weight $k + 1$ (prohibiting its deletion).
Note that the construction above does not work if there are more than $q$ terminals.
This limitation does not arise over infinite fields, 
so $\minlin{2}{\rationals}$ generalizes \textsc{Multiway Cut} with 
arbitrarily many terminals.
Marx~\cite{marx2006parameterized} presented the first fpt algorithm 
for \textsc{Multiway Cut} based on \emph{important separators}. 
This work was followed by a string of 
improvements~\cite{chen2009improved,xiao2010simple} 
including the approach based on
linear programming~\cite{cygan2013multiway,guillemot2011fpt} 
that is especially relevant to our work.

In \textsc{Multicut}, given an graph $G$,
a set of $m$ cut requests $(s_1,t_1),\dots,(s_m,t_m)$,
and an integer $k$, the goal is to remove at most $k$
edges from $G$ to separate $s_i$ from $t_i$ for all $i$.
This problem clearly generalizes \textsc{Multiway Cut}:
a reduction may introduce a request for every pair of terminals.
In turn, $\minlin{2}{\integers}$ generalizes it as follows:
add an equation $x = y$ for every edge $\{x,y\}$ in $G$;
then, for every pair of terminals $(s_i, t_i)$, 
introduce two new variables $s'_i$ and $t'_i$, and 
add two equations $s_i = p_i s'_i$ and 
$t_i = p_i t'_i + 1$, where 
$p_i$ is the $i$th prime number. 
Clearly, no path connecting $s_i$ and $t_i$ 
may exist in a consistent subset of equations, 
since this would imply a contradiction (different 
remainders modulo $p_i$). 
Moreover, if all cut requests are fulfilled, 
a satisfying assignment can be obtained 
by applying the Chinese Remainder Theorem in each component.
The parameterized complexity status of \textsc{Multicut} was 
resolved simultaneously by Bousquet et al.~\cite{bousquet2018multicut} 
and Marx~and~Razgon~\cite{marx2014fixed}. 
The latter introduced the method of 
\emph{random sampling of important separators},
also known as \emph{shadow removal}.

Another problem related to $\mintwolin$ is 
\textsc{Unique Label Cover} 
(ULC)~\cite{ChitnisCHPP16contract,iwata2016half,khot2002power}.
In ULC over an alphabet $\Sigma$, 
the input is a set of constraints of the form $\pi(x) = y$, 
where $x$ and $y$ are variables and $\pi$ is a permutation of $\Sigma$. 
Constraints are consistent if there is an assignment of values from $\Sigma$ 
to the variables that satisfies all constraints. 
The question is whether the input set can be made consistent 
by removing at most $k$ constraints. 
ULC lies at the heart of the Unique Games Conjecture. 
In the realm of parameterized complexity, it is known 
that ULC is fixed-parameter tractable when parameterized by 
$k + \abs{\Sigma}$, but \W{1}-hard when parameterized by $k$ alone.
To connect this problem with $\mintwolin$, consider for example
a field $\F$ and an equation $ax + by = c$ with $a, b, c \in \F$. 
For every value of $y$ there is exactly one value of $x$ 
that satisfies this equation. 
Thus, any equation is equivalent to a permutation constraint over $\F$, 
and ULC generalizes $\minlin{2}{\F}$. 
As an immediate consequence, observe that $\minlin{2}{\F_q}$ is in \FPT 
for all finite fields $\F_q$, and can be solved in
$\bigoh^*(q^{2k})$ time using the best known algorithm for ULC~\cite{iwata2016half,iwata201801all}. 
On the other hand, ULC is strictly more general than $\minlin{2}{\F}$: 
Consider $\Sigma = \{0,1,2,3,4\}$ and a permutation 
$\pi$ mapping $(0,1,2,3,4)$ to $(1,0,3,2,4)$ 
(see Figure~\ref{fig:permutation}). 
It is easy to see that no linear equation over~$\F_5$ defines this permutation. 

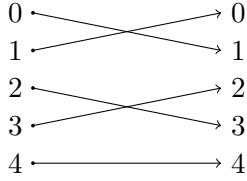
\begin{figure}[tb]
\centering
\begin{tikzpicture}[scale=0.5]
\filldraw[black] (0,0) circle (1pt) node[anchor=east]{4};
\filldraw[black] (0,1) circle (1pt) node[anchor=east]{3};
\filldraw[black] (0,2) circle (1pt) node[anchor=east]{2};
\filldraw[black] (0,3) circle (1pt) node[anchor=east]{1};
\filldraw[black] (0,4) circle (1pt) node[anchor=east]{0};

\filldraw[black] (5,0) circle (0pt) node[anchor=west]{4};
\filldraw[black] (5,1) circle (0pt) node[anchor=west]{3};
\filldraw[black] (5,2) circle (0pt) node[anchor=west]{2};
\filldraw[black] (5,3) circle (0pt) node[anchor=west]{1};
\filldraw[black] (5,4) circle (0pt) node[anchor=west]{0};

\draw[->] (0,4) -- (5,3);
\draw[->] (0,3) -- (5,4);
\draw[->] (0,2) -- (5,1);
\draw[->] (0,1) -- (5,2);
\draw[->] (0,0) -- (5,0);
\end{tikzpicture}
\caption{A permutation of $\{0,1,2,3,4\}$ not expressible as a linear equation.}
\label{fig:permutation}
\end{figure}

\medskip \noindent \textbf{Results.}
We prove that $\minlin{2}{\D}$, where $\D$ 
is an efficient Euclidean domain,
is fixed-parameter tractable.
For the special case when $\D$ is a field,
we provide a faster $\bigoh^*(2^{\bigoh(k \log k)})$ algorithm.
Furthermore, if $\D$ is a finite, $q$-element field,
we provide a $\bigoh^*((2q)^k)$ algorithm
improving upon the $\bigoh^*(q^{2k})$ upper bound
obtained by reduction to ULC.
To complement the results, we show that
$\minlin{3}{\D}$ is \W{1}-hard
(ruling out the existence of fpt algorithms for
$\minlin{r}{\D}$ when $r \geq 3$)
and we show that $\minlin{2}{\K}$ is \W{1}-hard for certain commutative rings $\K$ that are not Euclidean domains.
For example, the hardness result holds if $\K$ is
isomorphic to the direct product of nontrivial rings (such as the ring $\integers/6 \integers$ of integers modulo 6).

\medskip \noindent \textbf{Important balanced subgraphs.}
Our main technical contribution is the notion of
\emph{important balanced subgraphs},
which is a substantial generalisation of the \emph{important separators}
of Marx~\cite{marx2006parameterized}.
We believe that they can be applied to other problems as well,
so we give a general explanation.
Consider a parameterized deletion problem
where the input consists of an edge-weighted graph~$G$,
an integer $k$, and a polynomial-time
membership oracle to a family $\cF$ 
of minimal forbidden subgraphs of $G$ 
that we call \emph{obstructions}.
A (sub)graph of $G$ is \emph{balanced} if it does not 
contain any obstructions.
The goal is to find a set of edges of total weight at most $k$
that intersects all obstructions in $\cF$.
This objective is dual to finding 
a maximum-weight balanced subgraph of $G$.
For example, in \textsc{Bipartization} 
a graph is balanced if it is bipartite, and
the set of obstructions consists of all odd cycles.
Wahlstr\"{o}m~\cite{wahlstrom2017lp} presented 
a general method based on LP-branching for 
solving this problem in fpt time
when the obstructions $\cF$ are a family of cycles with the 
\emph{theta property}. This property can
roughly be defined as follows:
if a chordal path $P$ is added to a cycle $C$ from $\cF$,
then at least one of the smaller cycles formed by $P$ and $C$ is also in $\cF$.
For illustration, consider the theta graph in Figure~\ref{fig:theta}:
here $C = x_1 x_2 x_3 x_4$ is a cycle and $x_2 x_4$
is a chordal path that cuts it into two smaller cycles 
$C_1 = x_2 x_1 x_4$ and $C_2 = x_2 x_3 x_4$. 
If $C$ is in the family, then either $C_1$ or $C_2$ is in the family.
For instance, the set of all odd cycles in 
a graph has the theta property since 
any chordal path added to an odd cycle 
forms an odd and an even cycle.
Alternatively, the problem can be defined in terms of biased graphs.
A \emph{biased graph} is a pair $(G,\B)$ where $G$ is a graph
and $\B$ is a set of simple cycles in $G$ such that the complement
of $\B$ has the theta property; cycles outside $\B$ are referred to as the \emph{unbalanced}
cycles in $(G,\B)$. 
Biased graphs are encountered, for instance, in matroid theory~\cite{Zaslavsky:jctb89}.
The problem is then to find a set of $k$ edges that intersects every unbalanced cycle in $(G,\B)$.
In the case of \textsc{Graph Bipartization}, the set $\B$ contains all even cycles.

\begin{figure}[bt]
\centering
\begin{tikzpicture}
  \draw (0,0) ellipse (1cm and 2cm);
  \filldraw[black] (0,2) circle (1pt) node[anchor=south]{$x_1$};
  \filldraw[black] (-1,0) circle (1pt) node[anchor=east]{$x_2$};
  \filldraw[black] (0,-2) circle (1pt) node[anchor=north]{$x_3$};
  \filldraw[black] (1,0) circle (1pt) node[anchor=west]{$x_4$};

  \draw (-1,0) sin (-0.5,0.5) cos (0,0) sin (0.5,-0.5) cos (1,0);
\end{tikzpicture}
\caption{An example of a theta graph.}
\label{fig:theta}
\end{figure}
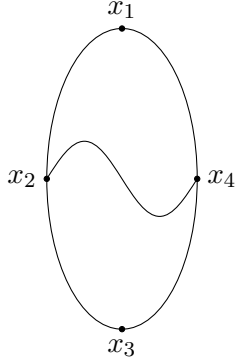

It is instructive to approach this global problem
by instead considering a local version
where a single root vertex $x$ is distinguished,
and the goal is to remove edges of total weight at most~$k$
to make the connected component of $x$ balanced.
For a balanced subgraph $H$ of $G$, 
define $c(H)$ as the cost of carving $H$ out of $G$ 
i.e. the sum of weights on all edges between 
$V(H)$ and $V(G) \setminus V(H)$ plus
the weights of edges in the subgraph of $G$ induced by $V(H)$ 
that are not in $H$. 
To solve the global problem, we can
choose a root, enumerate
solutions to the local problem
i.e. balanced subgraphs of cost at most $k$ that include $x$, 
and solve the remaining part recursively
(possibly with some branching 
to guess the intersection of the local and global solutions).
The caveat is that the number of balanced subgraphs of cost at most $k$
does not have to be bounded by any function of $k$.
To overcome this obstacle, we need another observation:
if there are two balanced subgraphs $H_1$ and $H_2$ 
such that $c(H_2) \geq c(H_1)$ and $V(H_2) \subseteq V(H_1)$, 
then $H_1$ is clearly a better choice than $H_2$
both in terms of cost and in the amount of ``work'' left in the remaining graph.
If the conditions above hold for $H_1$ and $H_2$,
we say that $H_1$ \emph{dominates} $H_2$.
See Figure~\ref{fig:ibs} for an illustration.
Formally, we want to compute a set $\cH$ of 
\emph{important balanced subgraphs} defined
analogously to the important separators: 
for any balanced subgraph $H'$ including $x$ of cost at most~$k$,
there is a subgraph $H \in \cH$ such that
$V(H') \subseteq V(H)$ and $c(H') \geq c(H)$.
In other words, the balanced subgraphs in~$\cH$
are Pareto efficient balanced subgraphs
in terms of maximizing the set of covered vertices
and minimizing the weight. 
However, note that the 
number of incomparable solutions is not bounded in $k$.
For example, if the input consists of a single, 
unbalanced cycle $C_n$ on $n$ vertices, then we may
output a single important balanced subgraph $H$,
with $V(H)=V(C_n)$ and $c(H)=1$, but there are $n$
incomparable solutions with these parameters, produced
by deleting any one edge of the cycle.
We handle this by proving that there is a 
\emph{dominating family} $\cH$ of
important balanced subgraphs such that
$\abs{\cH} \leq 4^k$ and every balanced subgraph of $G$
is dominated by some member of~$\cH$.
Moreover, there is an fpt algorithm that computes $\cH$
by branching based on the optimum
of the half-integral LP-relaxation of the local problem.

We note that important balanced subgraphs strictly generalize
important separators:
given a graph and two subsets of vertices,
one can recover important separators
by computing a dominating family of important balanced subgraphs---see 
Example~\ref{ex:important-separators} for further details.
In fact, the bounds achieved are identical:
using the important balanced subgraph framework to enumerate important
separators yields at most $4^k$ important separators of cost at most $k$,
and they can be enumerated in $\bigoh^*(4^k)$ time,
both of which match the bounds for important separators~\cite{chen2009improved,marx2006parameterized}.
Moreover, the increased generality is crucial for our algorithms, since
the cost of carving out a subgraph includes 
both the cost of a separator and a transversal of
unbalanced cycles reachable from the root after separation.
Thus, important balanced subgraphs can be used for graph separation
problems in a more general sense than simply enumerating graph cuts,
e.g. for computing transversals of obstruction families
with the theta property. 
As we show in what follows, removing a family of obstructions
is a key step in our $\mintwolin$ algorithms.

Let us illustrate important balanced subgraphs using a few examples. 
\begin{itemize}
\item Consider the biased graph $(G,\B)$ defined above, where $\B$ contains the
set of even cycles. Then a balanced subgraph of $(G,\B)$ is
precisely a bipartite graph, and our result outputs
connected bipartite subgraphs containing the root vertex $x$.
\item More generally, a \emph{group-labelled graph} is a graph $G$ in which
every oriented edge of $G$ is assigned a \emph{group label} $\gamma$ from a group $\Gamma$
so that for any edge $\{u,v\} \in E(G)$, the labels observe $\gamma(uv)=\gamma(vu)^{-1}$. 
Let a cycle $C=(v_1,\ldots,v_n)$ be \emph{balanced} if the product $\gamma(v_1v_2)\gamma(v_2v_3) \dots \gamma(v_nv_1)$
of the group labels of the edges of $C$ is the identity element of $\Gamma$.
This forms a biased graph $(G,\B)$ where $\B$ is the class of balanced cycles in $G$.
Thus the important balanced subgraph theorem can be used to, for example,
output connected subgraphs where every cycle has a length which is a multiple of $b$ 
for some $b \in \naturals$. This holds even if $\Gamma$ is non-abelian.
\item As a special case of the previous example, consider the \textsc{Subset Feedback Edge Set} problem.
In this problem, the input is a graph $G$ with a set of special edges $F \subseteq E(G)$, 
and the goal is to find a set of edges $X \subseteq E(G)$ such that no edge of $F$
is contained in a cycle in $G-X$. It is easy to observe that the class of cycles intersecting $F$
has the theta property (and in fact, can be captured as the unbalanced
cycles in a group-labelled graph). Then a subgraph $H$ of $G$ is balanced if any edge 
of $F \cap E(H)$ is a bridge in $H$. 
\end{itemize}
For more examples, see Wahlstr\"om~\cite{wahlstrom2017lp} and Zaslavsky~\cite{Zaslavsky:jctb89}.

From a technical perspective, the result follows from a refined analysis of the
LP formulation of Wahlstr\"om~\cite{wahlstrom2017lp} for the problem of
computing a minimum (vertex) transversal for the set of unbalanced cycles.
Wahlstr\"om showed that this can be solved in $\bigoh^*(4^k)$ time,
given oracle access to $\B$, where $k$ is the solution size. The result works
in two parts. First, consider the rooted case outlined above,
where the input additionally distinguishes a root vertex $x \in V(G)$
and the task is to find a balanced subgraph of minimum cost, rooted in $x$.
Wahlstr\"om provided a half-integral LP-relaxation for this problem, and showed
that it can be used to guide a branching process for an fpt algorithm
computing a min-cost rooted balanced subgraph. Second, the 
LP-relaxation is shown to have some powerful \emph{persistence properties} (see Section~\ref{sec:biasedgraphs}), that 
allow the solution from the rooted case to be reused in solving
the general problem. We reformulate and simplify these results for the edge deletion case.
We find that the extremal (``furthest'', or \emph{important}) optima to the LP
are described by a rooted, connected, balanced subgraph $H$ of $G$, 
where edges of value~$1$ in the LP are deleted edges within $V(H)$,
and the half-integral edges are the edges leaving $H$, i.e.
with precisely one endpoint in $V(H)$. Furthermore, every balanced
subgraph $H'$ of $G$ with $x \in V(H')$ can be ``improved'' so that
$V(H) \subseteq V(H')$ without increasing $c(H')$ (and such that
the edges of value $1$ in the LP are not contained in $E(H')$). 
The dominating family of important balanced subgraphs of $(G,\B)$ rooted in $x$ 
can then be obtained by branching over the status of the half-integral edges
leaving $V(H)$, in an analysis similar to that of Chen et al.~\cite{chen2009improved}
for the bound $4^k$ on the number of important separators.

\def\SCCF{0.3}
\begin{figure}
  \centering
  \begin{subfigure}{0.3\textwidth}
    \centering
    \includegraphics[scale=\SCCF]{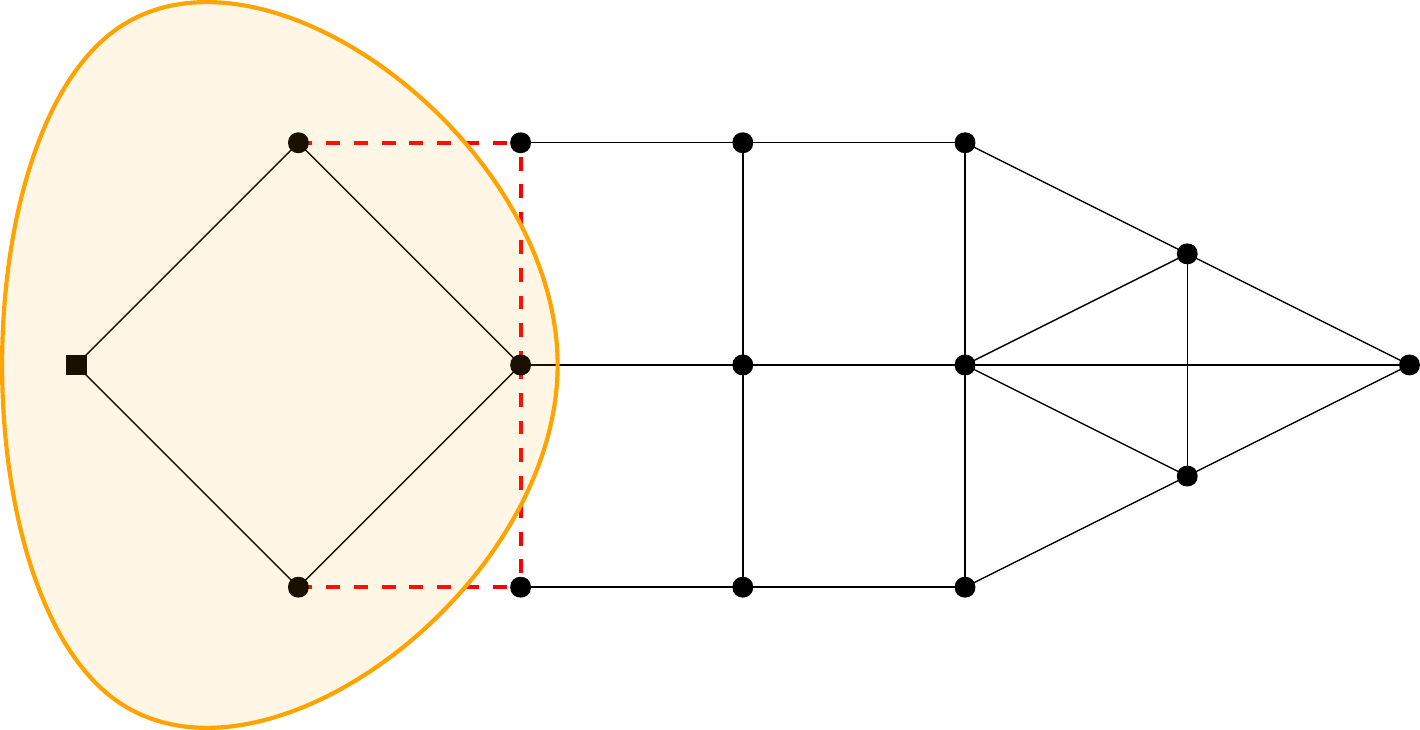}
    \caption{Subgraph $H_1$.}
  \end{subfigure}%
  \begin{subfigure}{0.3\textwidth}
    \centering
    \includegraphics[scale=\SCCF]{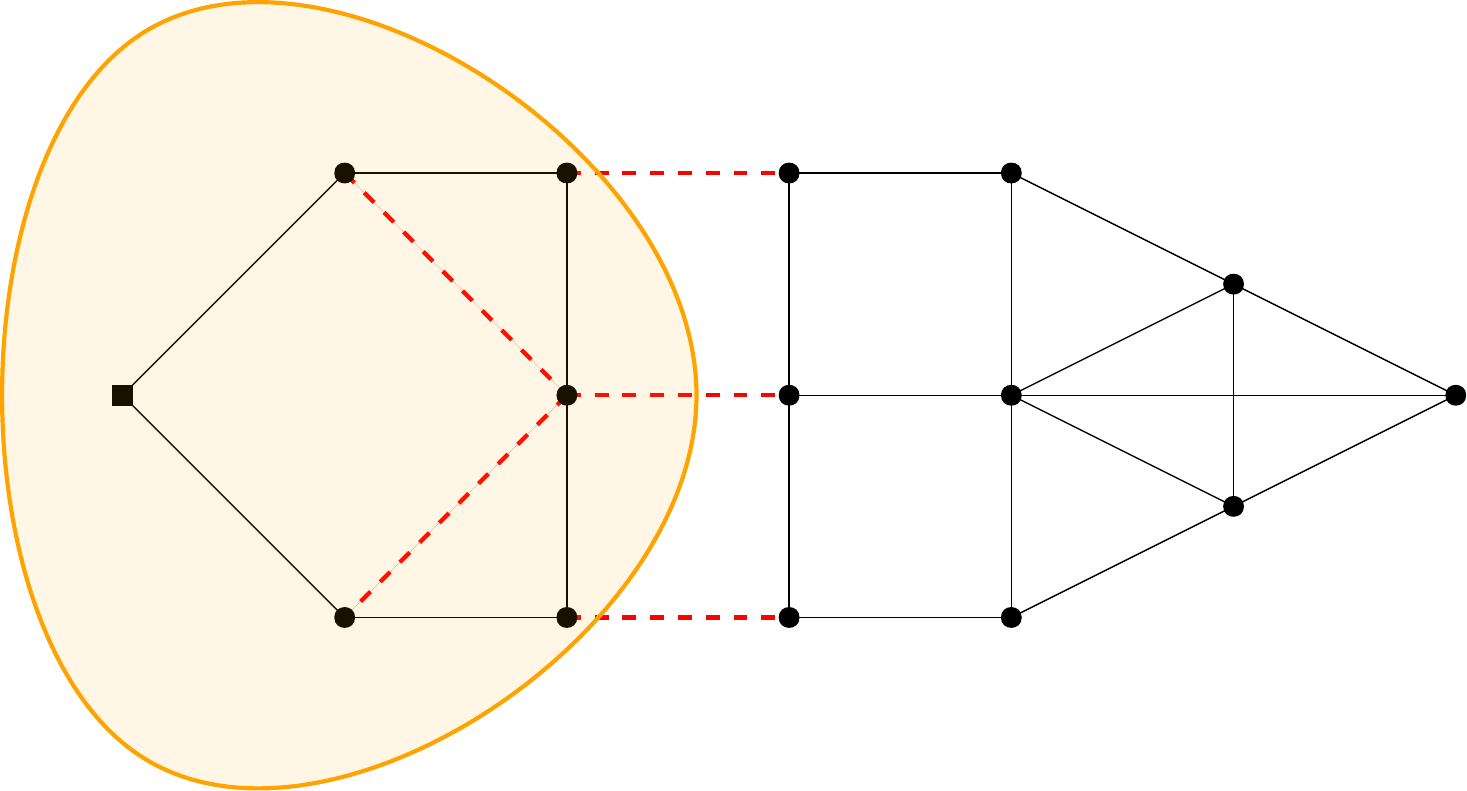}
    \caption{Subgraph $H_2$.}
  \end{subfigure}%
  \begin{subfigure}{0.3\textwidth}
    \centering
    \includegraphics[scale=\SCCF]{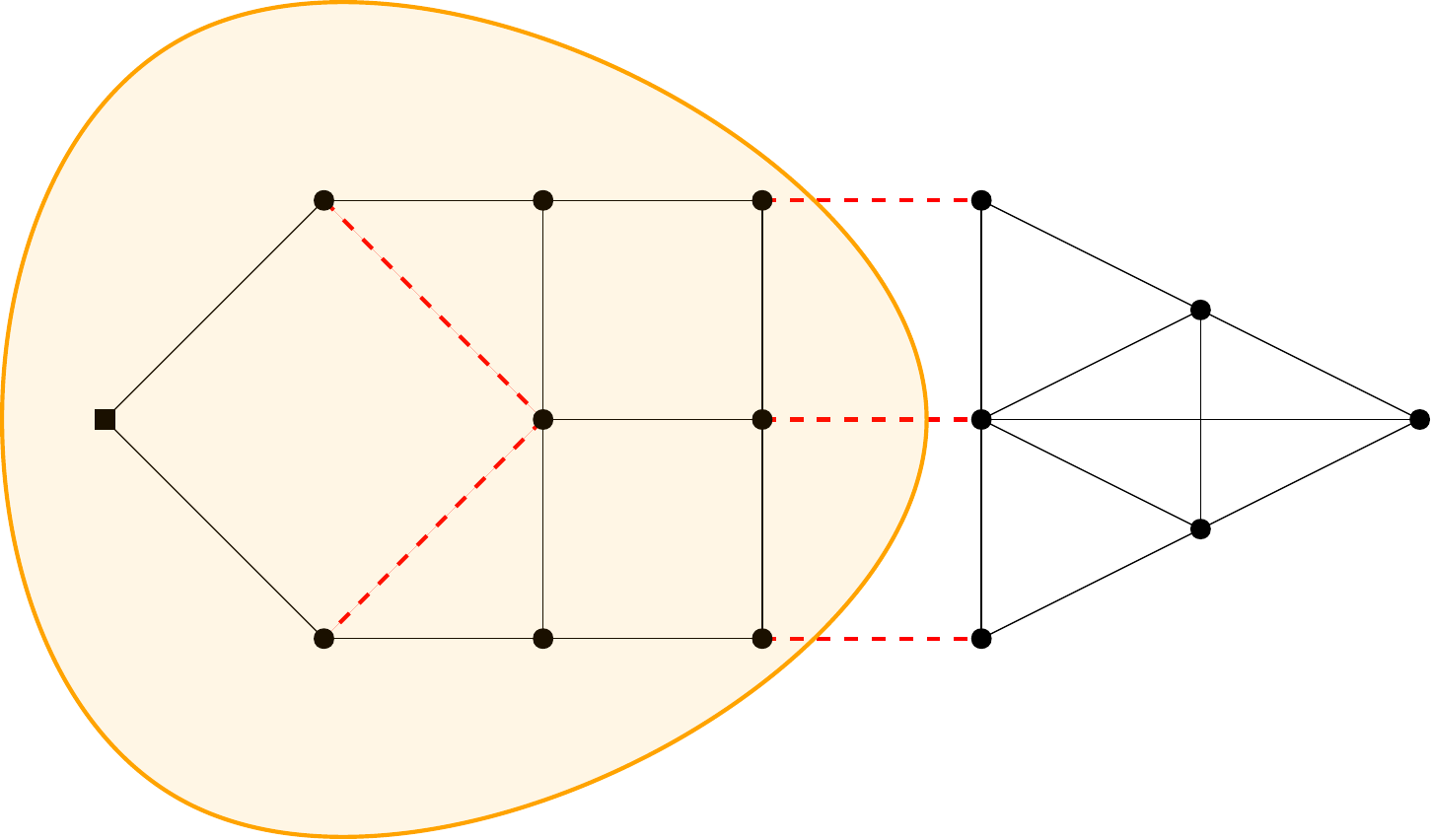}
    \caption{Subgraph $H_3$.}
  \end{subfigure}%
  \caption{Examples of rooted balanced subgraphs of the same graph.
  The root is the leftmost vertex,
  balanced cycles are even cycles, and
  all edges have unit weight.
  Red dashed edges are deleted, and the orange area
  covers all vertices reachable from the root.
  The cost of $H_1$ is $4$, while the cost of 
  $H_2$ and $H_3$ is $5$.
  Subgraph $H_1$ is incomparable with $H_2$ and $H_3$ since it has lower cost
  but $V(H_1)$ is a strict subset of $V(H_2)$ and $V(H_3)$.
  On the other hand, $H_2$ is dominated by $H_3$  since $V(H_2) \subsetneq V(H_3)$ while they have the same cost.}
  \label{fig:ibs}
\end{figure}

\medskip \noindent \textbf{$\mintwolin$ Algorithms for Fields.}
In short, our fpt algorithms are based on three steps:
compression, cleaning, and cutting.
Given an instance $(S, w_S, k)$ of $\mintwolin$,
we first use iterative compression to compute 
a slightly suboptimal ``solution'' $X$.
In the cleaning step we consider the primal graph of $S$
i.e. the graph with vertices for variables of $S$
and edges for equations, and 
produce a dominating family
of important balanced subgraphs
around a subset of vertices in $V(X)$.
Finally, the problem reduces to computing
a cut in the cleaned graph
that fulfills certain requirements.

For a basic example, consider $\minlin{2}{\rationals}$ i.e.
$\mintwolin$ over the field of rationals.
Every such instance can be viewed as a graph where an edge
connecting two variables is labelled by an equation
from $S$ ranging over these two variables.
Observe that any acyclic instance (with respect to the underlying primal graph) of $\lin{2}{\rationals}$
is consistent, since we can pick an arbitrary variable,
assign any value to it, and 
then propagate to the remaining variables
according to the equations labelling the edges. 
Thus, any inconsistency in an instance of $\minlin{2}{\rationals}$ 
is due to cycles.
By standard linear algebra,
a cycle may have zero, one, or infinitely many satisfying 
assignments---we call such cycles \emph{inconsistent}, \emph{non-identity} or \emph{identity}, respectively.
If an instance contains only identity cycles, we call it \emph{flexible},
and observe that, similarly to acyclic instances,
all flexible instances are consistent.
This follows from propagating a value in the same way as above.

By iterative compression, we may assume 
that we have an over-sized solution $X$ at our disposal i.e.
a set of equations of total weight $k + 1$
such that $S - X$ is consistent.
In the special case when $S - X$ is not only consistent but flexible,
a solution to the instance is a minimum cut $Z$ in $S - X$ 
that separates vertices $V(X)$ 
into distinct connected components according to some partition.
Since $\abs{V(X)} \leq 2k + 2$, we can enumerate partitions of $V(X)$
in fpt time, and compute a minimum cut $Z$ 
using the algorithm for \textsc{Multiway Cut}.
In the general case when $S - X$ is not flexible,
we assume that $X$ is minimal and hence
every connected component of $S - X$ contains
a non-identity cycle. This implies that $S - X$ admits 
a unique satisfying assignment $\varphi_X$
(otherwise, there is an edge in~$X$ connecting 
a flexible component with another component of $S - X$,
and the equation labelling that edge 
can be added back to $S - X$ without causing inconsistency).
Let $\varphi_Z$ be the assignment that satisfies $S - Z$.
We guess which variables in $V(X)$ have the same value in $\varphi_X$ 
and $\varphi_Z$ and which do not.
Let $T \subseteq V(X)$ be the subset of variables
that receive different values under these assignments.
The key observation is that the change propagates i.e.
every variable reachable $T$ in $S - Z$ has
a different value under $\varphi_X$ and $\varphi_Z$.
Since non-identity cycles in $S - X$ admit a unique satisfying assignment 
(which is $\varphi_X$),
none of them can remain in $S - Z$ and be reachable from $T$.
We show that the set of non-identity cycles in $S - X$
has the theta property.
This allows us to use the method of important
balanced subgraphs to get rid of non-identity cycles
reachable from the changing terminals $T$.
More specifically, in one of the branches
we obtain a set $F$ of size at most $k$
such that the connected components of $T$ in $S - (X \cup F)$
are free from non-identity cycles and thus flexible.
Moreover, all variables in the remaining components 
have the same value in $\varphi_X$ and $\varphi_Z$.
Thus, we can concentrate on the flexible part of the cleaned instance
and use the partition-guessing and cutting idea outlined above.
We remark that the reduction to the flexible case is analogous
to the \emph{shadow removal} process of Marx and Razgon~\cite{marx2014fixed},
but works in $\bigoh^*(4^k)$ time instead of $\bigoh^*(2^{k^3})$
(later improved to $\bigoh^*(2^{k^2})$ in~\cite{chitnis2015directed}) time
required by random sampling of important separators.
We also note that essentially the same algorithm work for $\minlin{2}{\F}$,
where $\F$ is a field.

\medskip \noindent \textbf{$\mintwolin$ Algorithms for Euclidean Domains.}
Let us now consider the general $\minlin{2}{\D}$ problem 
where $\D$ is a Euclidean domain.
Euclidean domains are less restricted than fields and they consequently
capture a wider and more multifaceted range of problems.
There are many examples of interesting Euclidean domains
that are not fields; the two prime examples are probably the ring
of integers $\integers$ and the rings
of univariate polynomials $\F[x]$ where the coefficients are
members of some field $\F$.
Important differences between fields and Euclidean domains become apparent
even when considering simple structures such as the ring of integers.
While in the case of fields all obstructions to consistency of $\lin{2}{\F}$ 
instances are cycles, in Euclidean domains paths may also be obstructions.
For example, consider the following system of equations over $\integers$:
$\{ y - 2x = 1, y - 2z = 0 \}$.
While both equations have integer satisfying assignments
(e.g. $(y,x) \mapsto (1,0)$ and $(y,z) \mapsto (2,1)$, respectively),
they are not simultaneously satisfiable:
the equation obtained by cancelling out $y$ is $2x - 2z = 1$
and it has no integer solutions.
This complicates the handling of Euclidean domains in algorithms.
Another issue is that $\lin{2}{\D}$ is less studied for Euclidean domains where,
in contrast, Gaussian elimination has been known for centuries and solves $\lin{r}{\F}$ 
for every $r$ in polynomial time.
Polynomial-time algorithms are known for $\lin{r}{\integers}$ and
$\lin{r}{\F[x]}$~\cite{kannan1985solving,Kannan:Bachem:sicomp79} for every $r$,
but to the best of our knowledge, 
there are no general algorithms for arbitrary Euclidean domains 
described in the literature, even for the simpler $\lin{2}{\D}$ problem. 
This forces us to develop methods for checking
consistency of $\lin{2}{\D}$ instances to be used in our \mintwolin\ algorithms. 
In short,  
$\minlin{2}{\D}$ is intrinsically a more complicated problem and this is reflected
in the more complicated fpt algorithm.

We show that 
after taking similar steps to those in the $\minlin{2}{\rationals}$ algorithm
(iterative compression, and cleaning by the method of important balanced subgraphs),
the solution is again a certain cut in the cleaned graph.
However, this time it is not sufficient to partition $T$ into connected components,
but we additionally need to break some paths that have no solutions in $\D$.
The latter requirements can be expressed as 
certain disjunctive cut request for the non-terminals.
The requests are of the form $(\{x,s\}, \{y,t\})$, 
where $s$ and $t$ are terminals, and
the cut is required to either separate $x$ from $s$ or $y$ from $t$.
We refer to Section~\ref{ssec:ppc} for the formal definition of the problem.
These cut requests are used to deal with inconsistent paths.
More concretely, we can check whether the path going 
from $x$ to $s$, then from $s$ to $t$, and finally from $t$ to $y$
is inconsistent in $\D$, and then add the cut request for it.
By iterative compression, the optimal solution is disjoint from $X$,
so the part of this path between $s$ and $t$ cannot be cut. 
Computing all cut requests requires polynomial time:
we consider every pair of terminals $s, t$ and non-terminals $x, y$,
and add a corresponding cut request if necessary.
To compute a separator that satisfies such disjunctive requests in fpt time, 
we reduce the cut problem to a special case of the \textsc{MinCSP}
parameterized by the solution cost.
Kim~et~al.~\cite{KimKPW21flow-arXiv,KimKPW21flow} solve this \textsc{MinCSP}
using the recently introduced technique of flow augmentation.

\medskip \noindent \textbf{Roadmap.}
The remainder of the paper is structured as follows.
In Section~\ref{sec:biasedgraphs} we describe the LP-based approach to
parameterized deletion problems, define important balanced subgraphs
and develop the fpt algorithm producing a dominating family of
important balanced subgraphs.
Section~\ref{sec:graphseparation} contains fpt algorithms for
the graph separation problems used in the $\mintwolin$ algorithms.
In Sections~\ref{sec:edom-algorithm}~and~\ref{sec:fields-algorithm}
we present the general algorithm for Euclidean domains
and faster algorithms for fields, respectively.
Section~\ref{sec:hardness} is devoted to \W{1}-hardness results.
We finish off in Section~\ref{sec:conclusion-and-discussion},
summarizing and discussing the results, open questions and 
possible directions for future work.

\medskip \noindent \textbf{Preliminaries and notation.}
We assume familiarity with the basics of graph theory, linear and abstract algebra, and combinatorial optimization. 
The necessary material can be found in, for instance, the textbooks by Diestel~\cite{Diestel:GT}, Artin~\cite{Artin:A}, and Schrijver~\cite{Schrijver:PE}, respectively.

We use the following graph-theoretic terminology in what follows.
Let $G$ be an undirected graph.
We write $V(G)$ and $E(G)$ to denote the vertices and edges of $G$, respectively.
For every vertex $v \in V(G)$,
let the {\em neighbourhood of $v$ in $G$} denoted by $N_G(v)$ be the set
$\{u \in V(G) \mid \{u, v\} \in E(G)\}$ and the \emph{closed neighbourhood} $N_G[v]=N_G(v) \cup \{v\}$.
We extend this notion to sets of vertices $S \subseteq V(G)$
in the natural way: $N_G(S)= (\bigcup_{v \in S} N_G(v)) \setminus S$.
If $U \subseteq V(G)$, then the {\em subgraph of $G$ induced
by $U$} is the graph $G'$ with
$V(G')=U$ and $E(G')=\{ \{v, w\} \mid v,w \in U \; {\rm and} \; \{v, w\} \in E(G)\}$. 
We denote this graph by $G[U]$.
If $Z$ is a subset of edges in $G$, we write $G-Z$ to denote the graph~$G'$ with $V(G') = V(G)$ and $E(G') = E(G) \setminus Z$.
For $X,Y \subseteq V(G)$, 
an \emph{$(X,Y)$-cut} is a subset of edges $Z$
such that $G - Z$ does not contain a path with 
one endpoint in $X$ and another in $Y$.
When $X,Y$ are singleton sets $X=\{x\}$ and
$Y=\{y\}$, we simplify the notation and write $xy$-cut instead
of $(X,Y)$-cut.

\section{Graph Cleaning}
\label{sec:biasedgraphs}

We will now consider one of the cornerstones in our algorithms
for \mintwolin: {\em graph cleaning}.
The framework we present is intimately connected with {\em  biased graphs}.
These are combinatorial
objects of importance especially to matroid theory~\cite{Zaslavsky:jctb89}.
To introduce biased graphs,
we recall that a \emph{theta graph}
is a collection of three vertex-disjoint paths with shared
endpoints---see Figure~\ref{fig:theta} for an illustration.
A \emph{biased graph} is a pair
$(G,\B)$ where $G$ is an undirected graph and $\B \subseteq
2^{E(G)}$ is a set of cycles in $G$ (referred to as the \emph{balanced
cycles of $G$}) with the property that if two cycles $C,C' \in \B$
form a theta graph, then the third cycle of $C\cup C'$ is
also in $\B$. We refer to a set of cycles~$\B$ with this property
as \emph{linear}. 
An example of two cycles forming a theta graph is given in Figure~\ref{fig:theta}
with $C$ following $x_1 \rightarrow x_2 \rightarrow x_4 \rightarrow x_1$ 
and $C'$ following $x_2 \rightarrow x_3 \rightarrow x_4 \rightarrow x_2$.
Given a biased graph $(G,\B)$, we always assume that $\B$ is defined via a membership oracle that
takes as input a cycle $C$ (provided as an edge set) and tests whether $C \in \B$. 

The most basic biased graph cleaning problem
is the following.

\pbDef{Biased Graph Cleaning (\BGC)}
{A biased graph $(G,\B)$ and an integer $k$.}
{Is there a set $X \subseteq V(G)$ such that $|X| \leq k$ and all cycles
in $G-X$ are balanced, i.e. members of $\B$.}

We will consider an LP-relaxation of \BGC and 
its rooted variant
in Section~\ref{sec:LP-relaxation}.
Results based on this LP-relaxation will then be
used in
Section~\ref{sec:biasedgraphcleaning}
where we present fpt algorithms for 
various biased graph cleaning problems.
These results are directly used in our single-exponential time
algorithm for \mintwolin\ over finite fields (Section~\ref{ssec:finfields-algorithm}).
Inspired by these kinds of problems and their 
solution structure, we introduce the concept of
{\em important balanced subgraphs} in Section~\ref{sec:importantbalancedsubgraphs}.
Our main result shows that we can efficiently compute a small family of important balanced subgraphs such that every other balanced subgraph is dominated by a member in the set. This
forms an important step of our later fpt algorithms for
\mintwolin\ over Euclidean domains.

\subsection{LP-relaxation for Rooted Biased Graph Cleaning}
\label{sec:LP-relaxation}

Previous work by Wahlstr\"om~\cite{wahlstrom2017lp} shows that \BGC has an fpt algorithm
running in $O^*(4^k)$ time. The algorithm is
based around a particular LP-relaxation.
The workhorse of this result is the following \emph{rooted} variant of \BGC.
Note that we have extended the problem with vertex weights.

\pbDef{Rooted Biased Graph Cleaning (\RBGC)}
{A biased graph $(G,\B)$, a vertex-weight function $w : V(G)
  \rightarrow \naturals$, a vertex $v_0 \in V(G)$, and an integer $k$.}
{Is there a set $X \subseteq V(G)$ such that $w(X) \leq k$, $v_0 \not\in X$, and all cycles in the connected component 
of $v_0$
in $G-X$ are balanced.}

We will now review the LP-relaxation that underlies the
fpt algorithms for \BGC and \RBGC.
Note that both \BGC and \RBGC can be defined in terms of seeking a set
of vertices $X$ that intersects a class of obstructions. In \BGC, the
obstructions  are simply all unbalanced cycles of $(G,\B)$. 
In \RBGC, we can define a class of obstructions consisting of the
combination $(P,C)$ of an unbalanced cycle $C$ and a (possibly empty)
path $P$ connecting $C$ to $v_0$. Such a rooted unbalanced
cycle is referred to as a \emph{balloon}~\cite{wahlstrom2017lp}.
Then we see that a set $X \subseteq V(G)$ is a solution to an instance
$I=((G,\B), w, v_0, k)$ of \RBGC if and only if $X$ intersects every
balloon (where the balloons are implicitly rooted in $v_0$).

A balloon $(P,C)$ can obviously be decomposed into two paths:
for any $v \in V(C) \setminus V(P)$, $(P,C)$ is the union of two paths $P_1$, $P_2$ from $v_0$ to $v$.
The LP-relaxation of an \RBGC instance $I=((G,\B), v_0, w, k)$ is then
defined as follows.
Let $x \in [0,1]^{V(G)}$ be an assignment, and for a path~$P$
define $x(P)=\sum_{v \in V(P)} x(v)$. The LP-relaxation of
$I$ has objective
\[
  \min \sum_{v \in V(G)} w(v) x(v)
\]
subject to the constraints $x(v_0)=0$, $x(v) \geq 0$ for every $v \in V(G) \setminus \{v_0\}$  
and
\[
x(P_1) + x(P_2) \geq 1
\]
for every balloon $(P,C)$ decomposed into two paths $P_1$ and $P_2$.
Wahlstr\"om~\cite{wahlstrom2017lp} showed several properties of this LP.
First, an optimal solution can be found in polynomial time, given
access to a membership oracle for $\B$. Second, it is \emph{half-integral}
i.e. the LP
always has an optimum $x^* \in \{0, 1/2, 1\}^{V(G)}$. 
Finally, it is \emph{persistent} in the sense that there is an optimal
solution $X \subseteq E(G)$ such that if $x^*(v)=1$ then $v \in X$
(and for certain vertices $v$ with $x^*(v)=0$ we can conclude $v \notin X$).

More precisely, we have the following. The \emph{support} of an
LP-solution $x$ is the set 
$\supp(x)=\{v \in V(G) \mid x(v)>0\}$.
  Let $I=((G,\B), v_0, w, k)$ be an instance of \RBGC
  and let $x=V_1 + \tfrac 1 2 V_{1/2}$ be a half-integral optimum to
  the LP-relaxation of $I$, i.e. $x(v)=1$ for $v \in V_1$, $x(v)=1/2$ for $v \in V_{1/2}$ and $x(v)=0$ otherwise.
  Let $V_R(x) \subseteq V(G)$ be the set of vertices connected to $v_0$
  in $G-\supp(x)$. Then $x$ is an \emph{extremal LP-optimum}
  if $V_R(x)$ is maximal among all LP-optima $x$.
If~$x$ is a half-integral extremal LP-optimum
for $I$, then there is an optimal solution $X \subseteq V(G)$ to $I$
such that $V_1(x) \subseteq X$ and $V_R(x) \cap X = \emptyset$. 
Via these properties, we can design an fpt-algorithm for \RBGC 
by a branch-and-bound approach over the LP~\cite{wahlstrom2017lp}.

In fact, an even stronger, more technical property holds, which we will
need in what follows. 
Let~$G$ be an undirected graph with vertex weights $w : V(G) \rightarrow {\mathbb N}$. 
For a set $U \subseteq V(G)$, we let $w(U)=\sum_{v \in U} w(v)$.

\begin{lemma}[Wahlstr\"om~{\cite[Lemma~6]{wahlstrom2017lp}}] \label{lem:v-persistence}
  Let $x=V_1+\tfrac 1 2 V_{1/2}$ be a half-integral extremal
  LP-optimum for a \RBGC instance $I=((G, \B), w, v_0, k)$
  with vertex weights $w : V(G) \rightarrow {\mathbb N}$  and let $V_R(x)$ be defined as above.
  Let $S \subseteq V(G)$ be a vertex set with $v_0 \in S$ such that $G[S]$ is balanced and connected.
  Then we can find a replacement solution that grows the closed neighbourhood  $N_G[S]$ to $N_G[S \cup V_R(x)]$ without
  paying a larger cost for deleting vertices. 
  More formally, there is a set of vertices $S^+$ and a set $S' \subseteq S^+$ 
  such that $G[S']$ is balanced and the following hold.
  \begin{enumerate}
  \item $S^+=N_G[S \cup V_R(x)]$;
  \item $N_G[S'] \subseteq S^+$;
  \item $V_R(x) \subseteq S'$;
  \item $V_1(x) \subseteq (S^+ \setminus S')$;
  \item $w(S^+ \setminus S') \leq w(N_G(S))$.
  \end{enumerate}
\end{lemma}

\subsection{Biased Graph Cleaning}
\label{sec:biasedgraphcleaning}

In this section we consider further variants of biased graph cleaning problems.
Our goal is to show that the vertex-weighted extension of \textsc{Biased Graph Cleaning}
and the edge version of \textsc{Rooted Biased Graph Cleaning} (\RBGCE)
both admit fpt-algorithms. These results are standard,
but they are needed for later results in the paper.

We begin by noting that \RBGC is in \FPT, using the same LP-branching
algorithm as in the unweighted case~\cite{wahlstrom2017lp}.

\begin{proposition}[Wahlstr\"{o}m~{\cite[Theorem 1]{wahlstrom2017lp}}]\label{pro:rbgc}
  \RBGC{} admits an fpt-algorithm with run-time $\bigoh^*(2^k)$ assuming
  that $\B$ is given by a polynomial-time membership oracle.
\end{proposition}
\begin{proof}[Proof sketch]
  Let $I=((G,\B), w, v_0, k)$ be an instance of \RBGC{},
  where $(G,\B)$ is a biased graph, $w$ is a set of vertex weights on
  $G$, $v_0 \in V(G)$ is the root vertex and $k \in \naturals$ is the deletion budget. 
  We note that the algorithm for the unweighted case~\cite[Lemma~8]{wahlstrom2017lp}
  also applies in the presence of integer weights. We use the slightly
  more careful version of the extended preprint~\cite{wahlstrom17preprint}.
  Let $W=\sum_{v \in V(G)} w(v)$, and assume $k<W$, as otherwise the
  instance is trivial. For a vertex $v \in V(G)$, we say
  \emph{fix $v=0$} to refer to setting a weight $w(v)=2W$, and
  \emph{fix $v=1$} to refer to setting $w(v)=0$. 
  Let $x^*$ be a half-integral extremal LP-optimum as in
  Lemma~\ref{lem:v-persistence}; such an optimum can be computed
  efficiently in a greedy manner~\cite{wahlstrom17preprint}.
  Let $\lambda$ be the cost of $x^*$. If $\lambda > k$, then we reject the instance.
  Otherwise $\lambda < W$ and as $x^*$ is half-integral, if we have fixed $v=0$, then
  we must have $x^*(v)=0$. Furthermore, if $\lambda<k/2$,
  then $\supp(x^*)$ is an integral solution of cost at most $k$.
  Now, as in~\cite{wahlstrom17preprint}, we either find an integral
  LP-optimum, or we find a half-integral vertex $v \in \supp(x^*)$
  such that $x^*(v)=1/2$, fixing $v=0$ increases the LP-optimum cost, and $w(v)>0$.
  Then we can recursively branch on fixing $v=0$, and on fixing $v=1$
  and decreasing $k$ by $w(v)$. In the former case $\lambda$
  increases by at least $1/2$ since the LP is half-integral, and $k$
  is unchanged.  In the latter case $\lambda$ decreases by $w(v)/2$,
  but $k$ decreases by $w(v)$. Hence the value of $k-\lambda$
  decreases by at least $1/2$ in both branches. It follows
  that an exhaustive branching takes $\bigoh^*(2^{2(k-\lambda)})$ time,
  and since $\lambda \geq k/2$ initially, this is $\bigoh^*(2^k)$. 
\end{proof}

Next, we consider the \emph{edge deletion} version of the problem,
\textsc{Rooted Biased Graph Cleaning Edge} (\RBGCE{}),
which is defined similarly to \RBGC{} except that the solution is an edge
set, not a vertex set, and where the input comes with edge weights $w$ 
instead of vertex weights. By a standard reduction, this problem is also in \FPT.

\begin{proposition}\label{pro:rbgce}
  \RBGCE{} admits an fpt-algorithm with run-time $\bigoh^*(2^k)$ assuming
  that $\B$ is given by a polynomial-time membership oracle.
\end{proposition}
\begin{proof}
  Let $I=((G,\B),w,v_0,k)$ be an instance of \RBGCE{}.  We provide a
  polynomial-time and parameter preserving reduction to \RBGC{}, which together with
  Proposition~\ref{pro:rbgc} shows the proposition. The instance
  $I'=((G',\B'),w',v_0',k')$ of \RBGC{} is obtained from $I$ as
  follows. We set $k'=k$. The graph $G'$ is obtained from $G$ by
  subdividing every edge of $G$ exactly once. We set the
  weight $w'$ of every original vertex to $k+1$ and the weight of
  every new (subdividing) vertex $x_e$, subdividing an edge $e$, to $w'(x_e)=w(e)$. Finally, we let $\B'$ be the set of
  all cycles $C$ in $G'$ such that the cycle obtained from $C$ after
  reversing the subdivision is in $\B$.
  This completes the construction of $I'$, which can clearly be
  achieved in polynomial-time and is parameter preserving. It remains to show that $I$ is a
  yes-instance if and only if $I'$ is a yes-instance.

  Towards showing the forward direction, let $X \subseteq E(G)$ be a
  solution for $I$ and let $X'\subseteq V(G')$ be the set of vertices
  used for subdividing the edges in
  $X$. We claim that $X'$ is a solution for $I'$. Suppose for
  contradiction that this is not the case. Then, there is a cycle $C'$
  in $G'-X'$ reachable from $v_0$ with $C' \notin \B'$. But then, the cycle $C$ obtained from $C'$ after reversing the
  subdivision of all edges is also in $G-X$ and reachable from $v_0$. Finally, because
  $C'\notin \B'$, we obtain that $C \notin \B$, contradicting our
  assumption that $X$ is a solution of $I$.

  Towards showing the reverse direction, let $X' \subseteq V(G')$ be a
  solution for $I'$. Then, because the weight of every original vertex
  is $k+1$, $X'$ only contains vertices used for the subdivision of
  edges of $G$.
  We claim that the set $X$ containing all edges of $G$
  whose corresponding vertex in $G'$ is in~$X'$ is a solution for
  $I$. Suppose for contradiction that this is not the case and there
  is a cycle $C$ in $G-X$ reachable from $v_0$ with $C \notin \B$. Let $C'$ be the
  corresponding cycle in $G'$, i.e. $C'$ is obtained from $C$ after
  subdividing each edge of $C$. Then, $C'$ is also in $G'-X'$ because
  $X'$ does not contain any (original) vertex of $C$. Moreover, $C'$
  is reachable from $v_0$ and $C' \notin \B'$, a contradiction to our
  assumption that $X'$ is a solution for $I'$.
\end{proof}

On a side note, we observe that the standard, non-rooted versions of the problems 
\RBGC{} and
\RBGCE{} are in \FPT.
The result follows from same procedure as in the original paper~\cite{wahlstrom2017lp},
building on Prop.~\ref{pro:rbgc} and~\ref{pro:rbgce}.

\begin{proposition}\label{pro:bgcw-fpt}
  The non-rooted variants of \RBGC{} and \RBGCE{} both admit an fpt-algorithm with run-time $\bigoh^*(4^k)$ assuming
  that $\B$ is given by a polynomial-time membership oracle.
\end{proposition}

\subsection{Important Balanced Subgraphs}
\label{sec:importantbalancedsubgraphs}

\emph{Important separators} are a central concept in fpt algorithms for graph separation problems that was
originally defined by Marx~\cite{marx2006parameterized}. Let $G$ be an undirected graph,
and let $X, Y \subset V(G)$ be disjoint sets of vertices. For 
an $(X,Y)$-cut $C \subseteq V(G)$, let $R(X,C)$ be the set of vertices reachable
from $X$ in $G-C$. Then $C$ is an \emph{important $(X,Y)$-separator}
if, for every $(X,Y)$-cut $C'$ such that $|C'| \leq |C|$
and $R(X,C) \subseteq R(X,C')$, we have $C'=C$.
In other words, for any $(X,Y)$-cut $C'$ such that $R(X,C) \subsetneq R(X,C')$
we must have $|C'| > |C|$.  Marx showed that for any graph $G$
and sets $X, Y$, there are at most $f(k)$ distinct important separators
$C$ with $|C| \leq k$~\cite{marx2006parameterized}, and this
bound was later improved to $f(k)=4^k$ (see~\cite{cygan2015book}).
The same bound applies to both undirected and directed graphs,
and by using standard reductions it also applies to edge cuts.
Important separators are a key component in many fpt algorithms,
including the algorithms for \textsc{Multiway Cut}~\cite{marx2006parameterized} and
\textsc{Multicut}~\cite{marx2014fixed} (see Cygan et al.~\cite{cygan2015book} for more applications).

We show a new result on the solution structure of \RBGCE that
generalizes important separators for undirected edge cuts.
We first note that the number of (in some sense) incomparable
solutions to \RBGCE is not bounded in $k$.
Indeed, if $C$ is an unbalanced cycle on $n$ vertices, then deleting
any one edge of $C$ is a minimal solution, and there is no clear order
of preference between these solutions. 
On the other hand, it turns out that a result in the style of
important separators does hold in terms of vertex sets of balanced
connected subgraphs of a biased graph.

Let us introduce some terminology. 
Let $(G, \B)$ be a biased
graph and let $w \colon E(G) \to \mathbb{Z}$ be a set of edge weights.
Let $H$ be a subgraph of $G$.
Let $\delta_G(X)$ for $X \subseteq V(G)$ denote the set of edges in
$G$ with precisely one endpoint in $X$, and for $F \subseteq E(G)$
we denote $w(F)=\sum_{e \in F} w(e)$.
We then define the \emph{cost of $H$} as
the cost of the edges in $G$ incident with $V(H)$ but not present in $H$ i.e.
\[
  c_G(H) = w(E(G[V(H)]) \setminus E(H)) + w(\delta_G(V(H))).
\]
We refer to $(E(G[V(H)]) \setminus E(H)) \cup \delta_G(V(H))$ as the
\emph{deleted edges} of $H$.

\bigskip

Let $H$ and $H'$ be balanced subgraphs
(with respect to $\B$) of $G$. We say that
$H'$ \emph{dominates} $H$ if $V(H) \subseteq V(H')$ and $c_G(H) \geq c_G(H')$,
and that $H'$ \emph{strictly dominates} $H$ if at least one of these two
inequalities is strict. Analogously to important separators, we refer to $H$
as an \emph{important balanced subgraph} in $(G,\B)$ if $H$ is a connected,
balanced subgraph of $G$ and no balanced subgraph $H'$ 
of~$G$ strictly dominates~$H$. 
We refer the reader to Figure~\ref{fig:ibs} for an illustration. 
Note by Lemma~\ref{lemma:dom-is-connected} that we may assume here that $H'$ is also connected.
Importantly, observe that if $H$ is dominated by~$H'$, then $H$ might not be a subgraph of $H'$. In the example of a single unbalanced cycle $C$, the subgraphs $C-\{e\}$ for $e \in E(C)$ all mutually dominate each other, although not strictly. 

Let $\cG := \cG(G, \B, k, v_0)$ be the family of connected balanced subgraphs in $(G, \B)$ 
that contain $v_0$ and have cost at most $k$.
A subset $\cH \subseteq \cG$ is a \emph{dominating family for $\cG$}
if for any $H \in \cG$ there is a subgraph $H'\in \cH$ that dominates $H$.
We show the following result.

\begin{theorem}[Dominating family of important balanced subgraph] \label{thm:important-subgraphs}
\setcounter{theoremspecial}{\thetheorem}
  Let $(G, \B)$ be a biased graph with positive integer edge
  weights $w$, let $v_0 \in V(G)$ and let $k$ be an integer.
  Let $\cG := \cG(G, \B, k, v_0)$ be the family of connected balanced subgraphs in $(G, \B)$ 
  that contain $v_0$ and have cost at most $k$.
  Then, in $\bigoh^*(4^k)$ time we can compute a dominating family $\cH$ for $\cG$
  such that $|\cH| \leq 4^k$. Furthermore, every member of $\cH$ is an important balanced subgraph of $(G, \B)$.
\end{theorem}

Before we present our proof of Theorem~\ref{thm:important-subgraphs}, we illustrate that important
biased subgraphs are indeed a generalisation of important separators.

\begin{example} \label{ex:important-separators}
Let $G$ be an undirected graph and $s, t \in V(G)$ be distinguished vertices. Note 
that since we are considering edge cuts, the assumption that $s$ and $t$ are single
vertices (as opposed to disjoint vertex sets $X$ and $Y$) can be
made without loss of generality. Now, add two vertices $z, z'$
and three edges $e_1=\{t,z\}$, $e_2=\{z,z'\}$, $e_3=\{t,z'\}$. Let $G'$ be the resulting graph,
and let $\mathcal{B}$ be the set containing all cycles except $C_t=\{e_1,e_2,e_3\}$.
Note that $\mathcal{B}$ trivially defines a linear class, so $(G', \mathcal{B})$ is a biased graph.
Finally, set edge weights $w(e_i)=k+1$, for $i \in \{1,2,3\}$, and $w(e)=1$ for every other edge $e \in E(G)$,
and use $v_0=s$ as the root vertex. 
Then a connected subgraph $H$ of $G'$ with $s \in V(H)$ and of cost at most $k$ is balanced
if and only if $t \notin V(H)$ (since breaking the unbalanced cycle $C_t$
would exceed the budget). Hence $H$ is a connected, balanced subgraph of $G'$
with $s \in V(H)$ and of cost at most $k$ if and only if the set of deleted edges $C$ of $H$ 
contains an $st$-cut in $G$. Furthermore, in such a case $V(H)=R(\{s\},C)$.
We see that the important $(s,t)$-separators in $G$ of cost at most $k$ directly correspond to the 
deleted edges of important balanced subgraphs~$H$ of $G'$ with $s \in V(H)$ and of cost at most $k$.
\end{example}

We proceed to prove Theorem~\ref{thm:important-subgraphs}; this occupies the rest of the subsection. 
We first note that we may assume that strictly dominating subgraphs are connected. 

\begin{lemma} \label{lemma:dom-is-connected}
  Assume that $G$ is a connected undirected graph that has no zero-weight edges.  Let $H$ be
  a connected balanced subgraph of $G$. If there is a balanced graph
  $H'$ that strictly dominates $H$, then there is also a connected balanced 
  graph $H'$ that strictly dominates $H$. Furthermore, if $H'$ is of minimum cost
  among all balanced graphs that dominate $H$, then $H'$ is connected. 
\end{lemma}
\begin{proof}
  Assume that there exists a balanced subgraph $H'$ of $G$ that strictly
  dominates $H$, and let~$H'$ be chosen to minimize $c_G(H')$ among all
  such subgraphs $H'$. Suppose that $H'$ is not connected, and first
  suppose that $H'$ contains a connected component $C$ such that $C
  \cap V(H)=\emptyset$.  Then $H'-C$ is balanced, $c_G(H'-C) < c_G(H')$
  since $G$ is connected, and $V(H) \subseteq V(H'-C)$. Hence $H'-C$
  also dominates $H$, and would be the preferred choice over $H'$.
  Hence we proceed assuming that every connected component of $H'$
  intersects $V(H)$. Now let $e \in E(H)$ be an edge connecting
  distinct connected components in $H'$. Such an edge clearly exists,
  e.g. follow a path in $H$ whose endpoints lie in distinct
  components of $H'$. Then adding $e$ to $H'$ yields another balanced
  subgraph $H''$, since no cycle passes through $e$ in $H''$.
  Then $c_G(H'')<c_G(H')$ and $H''$ dominates $H$; hence by choice of~$H'$,
  no such edge can exist. We conclude that $H'$ is connected. 
\end{proof}

For the proof of Theorem~\ref{thm:important-subgraphs},
we begin by adapting the LP-relaxation for \RBGC
and Lemma~\ref{lem:v-persistence} to the edge-deletion version \RBGCE.
More precisely, the LP-relaxation for \RBGC and Lemma~\ref{lem:v-persistence} are
both defined in terms of a solution space $x \in \{0,1/2,1\}^{V(G)}$ 
of half-integral relaxed solutions over the vertex set of a graph.
We give a natural reduction from \RBGCE to \RBGC,
and as a consequence construct 
half-integral optimal solutions $x^* \in \{0,1/2,1\}^{E(G)}$
to the edge-deletion version of the problem.
We also observe the following persistence properties
of such a solution~$x^*$, as a simplification
of  Lemma~\ref{lem:v-persistence}.
We refer to these solutions as the LP-relaxation of \RBGCE
(indeed, they are a projection of the solutions to the LP-relaxation
of the \RBGC-instance resulting from the reduction, so they
correspond to an LP over variables $E(G)$).

\begin{lemma} \label{edge-persistence}
  Let $I=((G, \B), v_0, k)$ be an instance of \RBGCE.
  In polynomial time, we can compute a half-integral extremal optimum
  $x^*=X_1 + \tfrac 1 2 X_{1/2}$ of the LP-relaxation of $I$ such
  that the following holds. 
  Let $X = X_1 \cup X_{1/2}$ be the support of $x^*$, $X \subseteq E(G)$. Let $G_R$ be the
  subgraph consisting of edges reachable from $v_0$ in $G-X$.
  Let $H$ be any connected balanced subgraph of $G$ with
  $v_0 \in V(H)$.  Then there is a balanced subgraph $H'$ of $G$ on
  vertex set $V(G_R) \cup V(H)$ such that $G_R$ is a subgraph of $H'$,
  $c_G(H') \leq c_G(H)$, and $X_1 \cap E(H') = \emptyset$. 
  
  In particular, unless $V(G_R) \subseteq V(H)$, there is a graph $H'$ that strictly dominates $H$ and has $V(G_R) \subseteq V(H')$.  \label{lemma:gr-in-dom}
\end{lemma}
\begin{proof}
  Let $(G', \B')$ be the biased graph obtained from $(G,\B)$ by subdividing every 
  edge $e \in E(G)$ by a new vertex $z_e$. Here, $\B'$ contains a cycle $C'$
  if and only if it is a subdivision of a cycle $C \in \B$. 
  Apply Lemma~\ref{lem:v-persistence} to $(G', \B')$ 
  giving every vertex $v \in V(G)$ weight $w(v)=2w(E(G))+1$
  and the subdividing vertices weight $1$. For an edge $e \in E(G)$,
  we say that the vertex $z_e$ which subdivides~$e$ in $G'$ \emph{represents} $e$ in $G'$. 
  For a subgraph $G_0$ of $G$, let $V'(G_0) \subseteq V(G')$ contain the copy
  in $G'$ of every vertex $v \in V(G_0)$ as well as the vertex $z_e$ subdividing $e$
  for every edge $e \in E(G_0)$. Note that $V'$ maps connected, respectively balanced
  subgraphs of $G$ to vertex sets $S$ such that $G'[S]$ is connected, respectively balanced.
  Consider the vertex sets $S=V'(H)$ and $R=V'(G_R)$.

  Since $G'[S]$ is balanced and connected,
  Lemma~\ref{lem:v-persistence} provides sets
  $S', S^+ \subseteq V(G')$, where $R \subseteq S'$,
  $N_{G'}[S'] \subseteq S^+$ and $S^+=N_{G'}[S \cup R]$.
  Let $H'$ be the subgraph of $G$ defined
  by $S'$, i.e. $V(H')=S' \cap V(G)$ and $e \in E(H')$ for
  $e \in E(G)$ if and only if the vertex subdividing $e$ is contained
  in~$S'$. We claim that $V(H')=V(H) \cup V(G_R)$.  In one direction,
  $V(H') \subseteq V(H) \cup V(G_R)$, since $S' \subseteq S^+ = N_{G'}[S \cup R]$
  and every vertex of $N_{G'}(S \cup R)$ 
  represents an edge in $G$. In the other direction, 
  we claim $V(H) \cup V(G_R) \subseteq S'$.
  Indeed, $V(G_R) \cup V(H) \subseteq S^+ = N_{G'}[S \cup R]$, and if there were a vertex
  $v \in V(G) \cap (S^+ \setminus S')$, then the cost of $S^+ \setminus S$ would
  exceed $w(N_{G'}(S))=c_G(H)$. Thus $V(H')=V(G_R) \cup V(H)$.
  Furthermore $H'$ is balanced, since any unbalanced cycle in $H'$
  would correspond to an unbalanced cycle in $G'[S']$.
    
  Next, we note that $G_R$ is a subgraph of $H'$ since $R=V'(G_R) \subseteq S'$.
  Finally, $c_G(H') \leq w(S^+ \setminus S') \leq w(N_{G'}(S))=c_G(H)$,
  and $X_1 \cap E(H')=\emptyset$, since the vertices subdividing $X_1$
  are contained in $S^+ \setminus S'$. 

  For the last part, let $H'$ be the subgraph produced above
  from $G_R$ and $H$. Then $H'$ is balanced, $c_G(H') \leq c_G(H)$,
  and $V(H')=V(G_R) \cup V(H)$ is a strict superset of $V(H)$. 
\end{proof}

We show one more property of the LP-relaxation. 
Recall that the constraints of the LP
are written as
$
x(P_1) + x(P_2) \geq 1
$
for every balloon $B=(P,C)$ decomposed into two paths $P_1$ and~$P_2$.
Equivalently, for every balloon $B=(P,C)$, there is a constraint
where the edges of $P$ have coefficient~$2$, and the edges of $C$ have coefficient $1$. 
The constraint for $B=(P,C)$ is \emph{tight} if equality holds in the constraint.
By so-called \emph{slackness conditions}, it is known that for any LP optimum $x^*$
and any edge $e$ in the support of $x^*$ there is a tight constraint involving $e$,
i.e. a balloon $B=(P,C)$ with $e \in E(B)$ such that the constraint for $B$ is tight.

\begin{lemma} \label{lemma:half-int-form}
  Let $x^* = X_1 + \tfrac 1 2 X_{1/2}$ be a half-integral extremal optimum
  computed in Lemma~\ref{edge-persistence}. Let $X=X_1 \cup X_{1/2}$,
  let $G_R$ be the subgraph corresponding to the connected component of $v_0$
  in $G-X$, and let $V_R=V(G_R)$.
  Then $X_1=E(G[V_R]) \setminus E(G_R)$ and $X_{1/2}= \delta_G(V_R)$.
\end{lemma}
\begin{proof}
  For the first item, let $e \in X_1$. By the slackness conditions, there must be a tight
  constraint in the LP which contains $e$. By inspection of the constraints,
  this implies that there is a balloon $B_e=(P,C)$ rooted in $v_0$ such that $e$
  is contained in $B_e$. 
  Since $x^*(e) = 1$, it may only appear with coefficient~$1$
  in the summation of the constraint, hence $e \in C$.
  Moreover, we have $x^*(e')=0$ for every edge $e' \in P \cup C \setminus \{e\}$.
  Then $B_e-e$ is contained in $G_R$. Similarly, let $e \in X_{1/2}$ and
  assume towards a contradiction that $e$ is spanned by $G_R$, i.e. $e \subseteq V_R$.
  Let $B_e=(P,C)$ be a balloon with $e \in E(B_e)$ such that the corresponding
  constraint is tight. There are two cases. First suppose that~$e$ occurs
  in the path $P$ of $B_e$. Since $e$ occurs with coefficient $2$
  in the constraint corresponding to~$B_e$, for every other edge $e' \in E(B_e)$
  we must have $x^*(e')=0$. But since $e \subseteq V_R$, this implies that every vertex of
  $B_e$ is in $V_R$. In particular, there is a path from $v_0$ to $C$ entirely contained in $G_R$,
  and considering a shortest such path we find a path $P'$ that is
  internally disjoint from $C$. This produces a balloon $B_e'=(P',C)$ disjoint from $X$,
  which is a contradiction. Next, suppose that $e \in E(C)$. Note that by tightness,
  $V(C) \subseteq V_R$. Indeed, by tightness $C$ intersects precisely two edges of $X_{1/2}$
  and none of $X_1$, and since $e \notin \delta_G(V_R)$ by assumption,
  it follows that both edges of $E(C) \cap X$ are spanned by $G_R$, i.e. $V(C) \subseteq V_R$.
  Then $C \setminus X_{1/2}$ splits into two paths $P_1$ and $P_2$, where
  one of them may be edgeless but both consist entirely of vertices of $V_R$. 
  Let $P'$ be a shortest path in $G_R$ from $P_1$ to $P_2$. Then $P'$
  forms a chordal path for the unbalanced cycle $C$, hence results in at least one new
  unbalanced cycle $C'$ of weight $1/2$ in $x^*$. Furthermore, there is a
  path $P''$ contained in $G_R$ forming a balloon $B_e'=(P'',C')$ of weight $1/2$
  in $x^*$, which is a contradiction to $x^*$ being an LP solution. 
  Hence $X_{1/2}=\delta_G(V_R)$. 
\end{proof}

We can now show the main result. 

\begin{proof}[Proof of Theorem~\ref{thm:important-subgraphs}]
  We assume that $G$ is connected, or otherwise restrict our attention to the connected component
  of $G$ containing the vertex $v_0$. Furthermore, by assumption the edge weights of $G$ are positive.
  Hence Lemma~\ref{lemma:dom-is-connected} applies. 
  Now, recall that $\cG$ denotes the family of all connected, balanced subgraphs in $(G,\B)$ that
  contain $v_0$ and have cost at most $k$ and let $\cH' \subseteq \cG$ be all subgraphs $H \in \cG$
  that are not strictly dominated by any member of $\cG$. We observe that every member of $\cH'$ is 
  important. Indeed, let $H \in \cG$ and assume that there is a balanced subgraph $H'$ of $(G,\B)$ that dominates $H$. 
  Choose $H'$ to minimize $c_G(H')$. Then by Lemma~\ref{lemma:dom-is-connected} $H'$ is connected.
  Furthermore $c_G(H') \leq c_G(H) \leq k$ and $v_0 \in V(H) \subseteq V(H')$. Thus $H' \in \cG$.
  Thus any subgraph~$H$ of $(G,\B)$ that is ``domination maximal'' within $\cG$ is important in $(G,\B)$,
  and we can focus on computing a dominating family $\cH \subseteq \cG$.

  For this, we present a branching procedure over the LP-optimum.  Let a
  \emph{branching state} be defined by a tuple $(E_0, E_1)$ where
  $E_0, E_1 \subseteq E(G)$ are disjoint edge sets.  For a branching
  state $B=(E_0, E_1)$, we
  let $LP_e(B)$ denote the LP
  on the graph $G-E_1$, with
  edge weights modified so that $w(e)=2k+1$ for every $e \in E_0$.
  Intuitively, edges in $E_0$ can be thought of as undeletable
  while edges in $E_1$ as deleted.
  We let $B^*$ denote the half-integral solution
  to $LP_e(B)$ and let $G_B$ denote
  the corresponding subgraph of $G$, i.e.
  $G_B$ is the connected component of 
  $G - (E_1 \cup \supp{(B^*)})$ containing $v_0$.
  Let us consider the following branching procedure.
  \begin{enumerate}
      \item Let $B=(E_0,E_1)$ be a branching state that initially is set to $(\emptyset,\emptyset)$. 
      \item Let $B^*=X_1 + \tfrac 1 2 X_{1/2}$ and let $X=X_1 \cup X_{1/2}$ be its support.
            Let $k'$ be the cost of $B^*$.
      \item If $|E_1|+k' > k$, then abort the branch without output.
      \item If $X_{1/2}=\emptyset$, output $G_B$ as a potential solution and abort the branch.
      \item Otherwise, initialize a new branching state $B'=(E_0', E_1')$ with $E_0'=E_0 \cup E(G_B)$ and $E_1'=E_1 \cup X_1$. 
\item Let $e \in X_{1/2}$ be a half-integral edge and branch
  recursively on the two states $B_1=(E_0' \cup \{e\},E_1')$ and $B_2=(E_0',E_1' \cup \{e\})$.
  \end{enumerate} 
  We will show that for any balanced, connected subgraph $H$ of $G$ with $c_G(H) \leq k$,
  at least one of the produced subgraphs $G_B$ dominates $H$.  Towards this,
  we need some support claims  about the branching process. 
  
  \begin{claimspecial} \label{claim:zero-part-bconn}
  In every branching state $B=(E_0,E_1)$ encountered by the algorithm, 
  the edge set $E_0$ forms a balanced connected subgraph of $G$ rooted in $v_0$.
\end{claimspecial}
\begin{claimproof}
  We choose to interpret the initial empty edge set as the subgraph of
  $G$ containing the root $v_0$ and no edges or any further vertices.
  The claim now holds by induction from the root. Note that there are
  two places where the $E_0$-part of a branching state is modified.
  First, let $B=(E_0,E_1)$ be a branching state and let $B^*$ be the
  half-integral optimum of $LP_e(B)$ used by the algorithm. Assume
  that the cost of $B^*$ is at most $k'=k-|E_1|$ as otherwise no
  further branching state is produced. By assumption, $E_0$ forms a
  connected subgraph of $G$, and every edge of $E_0$ has cost $2k+1$
  in $LP_e(B)$. Hence $B^*(e)=0$ for every $e \in E_0$, and
  $E_0 \subseteq E(G_B)$. Thus in the new branching state
  $B'=(E_0',E_1')$ we in fact have $E_0'=E(G_B)$ which is a connected,
  balanced, rooted subgraph of $G$ by construction.
  Otherwise, assume that a new state is formed as $B'=(E_0 \cup \{e\},E_1)$
  for some edge $e$ that is half integral in $LP_e(B)$.
  Then by Lemma~\ref{lemma:half-int-form}, $e$ is an edge leaving $G_B$,
  hence $E_0'=E(G_B) \cup \{e\}$ forms a connected subgraph.
  Finally, we note that $E_0'$ is balanced, since otherwise
  there would exist an unbalanced cycle $C$ in $E_0'$ using the edge
  $e$, but since $e$ is leaving $G_B$, $e$ is a pendant edge in $E_0'$. 
\end{claimproof}

\begin{claimspecial} \label{claim:one-part-touches}
  In every branching state $B=(E_0,E_1)$ encountered by the algorithm,
  every edge of $E_1$ has at least one endpoint in $V(E_0)$. 
\end{claimspecial}
\begin{claimproof}
  Shown by induction. In the initial state $(\emptyset,\emptyset)$,
  it holds vacuously. Thereafter, the $E_1$-part of a branching state
  is modified in two ways. First, let $B=(E_0,E_1)$ be a branching state
  and let $B^*=X_1 + \tfrac 1 2 X_{1/2}$ be the optimum of $LP_e(B)$.
  Let $G_B$ be the corresponding subgraph of $G$ and let $B'=(E_0',E_1')$
  be the new resulting branching state.  Then $E_1'=E_1 \cup X_1$,
  edges of~$E_1$ intersect $V(E_0) \subseteq V(E_0')$ by assumption,
  and edges of $X_1$ are spanned by $E(G_B) \subseteq E_0'$
  by Lemma~\ref{lemma:half-int-form}. Otherwise, we have a
  modification $E_1'=E_1 \cup \{e\}$ for some $e \in X_{1/2}$,
  where $e$ intersects $V(E_0)$
  by Lemma~\ref{lemma:half-int-form}.
\end{claimproof}

We say that $H$ is \emph{compatible with} a branching state $B=(E_0,E_1)$
if $E_0 \subseteq E(H)$ and $E_1 \cap E(H) = \emptyset$.
Note that it follows that every edge of $E_1$ is deleted in $H$; indeed, by
Claim~\ref{claim:one-part-touches} every edge of~$E_1$ intersects
$V(H)$, and every edge intersecting $V(H)$ not present in $H$ is
deleted in $H$. Also say that~$H$ is \emph{domination compatible with}
$B$ if there is a balanced, connected subgraph $H'$ of $G$
rooted in~$v_0$ such that $H'$ dominates $H$ and is compatible with $B$.
Note that if $H$ is domination compatible with 
a leaf state in the branching tree, then the subgraph $G_i$
produced in this state dominates $H$.
Indeed, let $B(E_0,E_1)$ be the leaf branching state, and $G_i=G_B$.
By assumption there is a graph~$H'$ dominating $H$, compatible with $B$.
Then $\delta_G(G_B) \subseteq E_1$, and $E_1 \cap E(H')=\emptyset$.
Furthermore $E(G_B)=E_0 \subseteq E(H')$. Hence $V(H')=V(G_B)$,
and the cost of $G_B$ is optimal among all such graphs by the integrality of the LP solution $L_e(B)$.

We can now prove by induction that for every balanced, connected
subgraph $H$ of $G$ rooted in $v_0$ with $c_G(H) \leq k$, the
branching process will produce at least one balanced subgraph $G_i$
that dominates $H$. We claim by induction that for every level $\ell$
of the branching tree, either such a graph $G_i$ has been produced at
a preceding level or there is a state on level $\ell$ domination
compatible with $H$. 

In the root node, we have the initial branching state
$(\emptyset,\emptyset)$, where we can choose $H'=H$.
Inductively, first assume that $B=(E_0,E_1)$ is a branching state
domination compatible with $H$ via a graph $H'$ dominating $H$,
and let $B^*=X_1 + \tfrac 1 2 X_{1/2}$ be the optimum of $LP_e(B)$. 
Let $B'=(E_0',E_1')$ be the new resulting branching state. We will
show that $H'$ is domination compatible with $B'$, hence the same
holds for $H$.

Recall that $LP_e(B)$ is defined in the subgraph $G':=G-E_1$.
Let $G_B$ be the subgraph of~$G'$ corresponding to the optimum $B^*$.
By Lemma~\ref{edge-persistence} there is a balanced subgraph $H''$ of $G'$ 
that dominates $H'$ in $G'$, such that $G_B$ is a subgraph of $H''$
and $X_1 \cap E(H'')=\emptyset$. We need to show that $E_0' \subseteq E(H'')$,
that $E_1' \cap E(H'')=\emptyset$, that $c_G(H'') \leq c_G(H')$,
and that $V(H'') \supseteq V(H')$. It then follows that
$H''$ dominates $H'$ in $G$ and is compatible with $B'$.

For the first, as before we have $E_0 \subseteq E(G_B)$ by construction so $E_0'=E_0 \cup E(G_B)=E(G_B) \subseteq E(H'')$.
For the second, since $H''$ is a subgraph of $G-E_1$ disjoint from $X_1$,
we have $E_1' \cap E(H'')=\emptyset$. For the cost,
 we have $c_{G'}(H'') \leq c_{G'}(H')$ by Lemma~\ref{edge-persistence}.
As noted above, every edge of $E_1$ is deleted in $H'$;
hence $c_{G'}(H')=c_G(H')-|E_1|$. Similarly, since every
edge of $E_1$ intersects~$V(E_0')$ by Claim~\ref{claim:one-part-touches}
and $E_0' \subseteq E(H'')$, every edge of $E_1'$ is deleted in $H''$
with respect to $G$. Thus $c_G(H'')=c_{G'}(H'')-|E_1| \leq c_{G'}(H')-|E_1|=c_G(H')$.
Finally, $V(H'') \supseteq V(H')$ by Lemma~\ref{edge-persistence}.
Thus~$H'$ is domination compatible with $B'$.

The only remaining step to consider is when a branching state
is modified as $B=(E_0,E_1) \mapsto (E_0 \cup \{e\},E_1)$
or $(E_0,E_1) \mapsto (E_0,E_1 \cup \{e\})$ for some edge $e$
that is half-integral in $LP_e(B)$. However, by assumption there
exists a subgraph $H'$ that dominates $H$ and is compatible with $B$.
Then either $e \in E(H')$ or $e \notin E(H')$, and precisely one
of the two new branching states is compatible with $H'$.
Furthermore, by comparing Lemma~\ref{lemma:half-int-form}
to the definition of the cost function $c_G(H')$, it is clear
that the cost of the resulting state does not exceed $c_G(H') \leq
c_G(H) \leq k$. Hence by induction, there is a leaf in the branching
tree which is domination compatible with $H$.

Finally, we claim that the whole
process produces at most $4^k$ outputs and can consequently be performed in $\bigoh^*(4^k)$ time. 
To see this, we use an approach that is similar to the one used in~\cite{wahlstrom2017lp}. 
Consider the value of the ``LP gap'' $k-(|E_1|+k')$ computed in some node of 
the branching tree corresponding to the above computation. 
Clearly, this value is initially at most $k$, and if it is negative in a node,
then that branch of the computation is aborted. We claim that furthermore,
this gap decreases by at least $1/2$ from a branching state $B$ to both 
of its children $B_1$ and $B_2$.
In the branching state $B_1$,
$B_1^*$ is also a valid solution to the state $B$,
and in the branching state $B_2$, $B_2^*$ becomes
a valid solution to the state $B$ if we modify the value of
$e$ to $x_e=1$. In both cases, we get a valid LP solution to the
state $B$. We claim that these solutions cannot be optima for $LP_e(B)$.
On the one hand, if $E_0 \mapsto E_0 \cup \{e\}$ then the set of reachable
vertices $V(G_B)$ increases strictly. Since the extremal solution $B^*$ is 
chosen so that this set is maximal among all LP-optima, the result cannot
be an LP-optimum. On the other hand, if $E_1 \mapsto E_1 \cup \{e\}$ and the 
resulting branching state produces an optimal solution for $LP_e(B)$, 
then by Lemma~\ref{lemma:half-int-form} the endpoints of $e$ must be spanned
by the resulting set $E_0' \supseteq E_0$, which again contradicts the choice
of $V(G_B)$ as maximal. Thus, the cost of these solutions is
greater than the cost of $B^*$. Since the cost is half-integral
(given integral edge weights), this difference is at least $1/2$.
Hence the entire branching process will finish at depth at most
$2k$, producing at most $2^{2k}$ outputs. 
\end{proof}

\section{Graph Partitioning}
\label{sec:graphseparation}

As discussed in the introduction, the general strategy for our fpt algorithms
aims to reduce $\mintwolin$ over various domains to graph partitioning problems.
In this section we develop algorithms for two problems---\textsc{Partition Cut} and \PPC{}---which arise in the study of $\mintwolin$ over
fields and Euclidean domains, respectively.

\subsection{Partition Cut}

A partition $\cP$ of a finite set $N$ is a family of
pairwise disjoint subsets $B_1, \dots, B_m$ of $N$ such that 
$\bigcup_{i=1}^{m} B_i = N$.
For any $x, y \in N$,
we write $\cP(x) = \cP(y)$ if $x$ and $y$ appear
in the same subset of $\cP$, while
$\cP(x) \neq \cP(y)$ if they appear in distinct subsets.
If $\cP'$ is a partition of $N$ such that
$\cP'(x) = \cP'(y) \implies \cP(x) = \cP(y)$
for all $x,y \in N$,
then we say that $\cP'$ \emph{refines} $\cP$. 
All partitions of a finite set 
can be enumerated in $\bigoh(1)$ 
amortized time per partition~\cite{ichiro1984efficient}.

Let $G$ be an undirected graph, $T$ be a subset of its vertices 
called \emph{terminals},
and $\cP$ be a partition of $T$.
A subset of edges $X$ in $G$ is a \emph{$\cP$-cut} if
no component of $G-X$ contains
terminals from more than one subset of $\cP$.
Consider the following graph separation problem:

\pbDefP{Partition Cut}
{An undirected graph $G$ with positive integer 
edge weights $w_G : E(G) \rightarrow \naturals^+$, 
a set of terminals $T \subseteq V(G)$,
a partition $\cP$ of $T$,
and an integer $k$.}
{$k$.}
{Is there a $\cP$-cut in $G$ of total weight at most $k$?}

We may view this problem in the light of multiway cuts.

\pbDefP{(Edge) Multiway Cut}
{An undirected graph $G$ with positive integer 
edge weights $w_G : E(G) \rightarrow \naturals^+$, 
a set of vertices (terminals) $T \subseteq V(G)$ 
and an integer $k$.}
{$k$.}
{Is there a set of edges $X \subseteq E(G)$ of total weight at most $k$
such that every component of $G-X$ contains at most one vertex from $T$?}

One way to formulate the goal of the solution $X$ in {\sc Partition Cut}
is to ensure the partition of terminals into connected
components of ${G - X}$ refines $\cP$.
Thus, \textsc{Multiway Cut} is a special case of this problem
where every subset of $\cP$ is a singleton i.e. 
$X$ needs to separate all terminals. In fact, we can reduce from \textsc{Partition Cut} to \textsc{Multiway Cut} and thus show that
\textsc{Partition Cut} is in \FPT.
  
\begin{proposition}[Cygan~et~al.~\cite{cygan2013multiway}]
\label{pro:solveMWC}
  \textsc{Multiway Cut} is solvable in $\bigoh^*(2^k)$ time.
  If a solution exists, then the algorithm computes it in this time.
\end{proposition}

\begin{lemma} \label{lem:p-cut-fpt}
  \textsc{Partition Cut} is solvable in $\bigoh^*(2^k)$ time.
  If a solution exists, then the algorithm computes it in this time.
\end{lemma}
\begin{proof}
Let $(G, w_G, T, \cP, k)$ be an instance of \textsc{Partition Cut}, 
where $\cP = \{B_1, \dots, B_m\}$.
For every $i \in [m]$, introduce a superterminal vertex
$s_i$ and connect all terminals in $B_i$ to $s_i$
with edges of weight $k + 1$.
Let the resulting graph be $G'$,
the weight function $w_{G'}$,
and the set of superterminals be $S = \{s_1, \dots, s_m\}$.
Then the instance of \textsc{Multiway Cut} is $(G', w_{G'}, S, k)$.
Correctness of the reduction follows by 
noting that a cut in $G'$ is a solution
only if it partitions superterminals
into distinct connected components.
Since edges connecting any $s \in S$ to any $t \in T$ have weight $k+1$,
they cannot be included in the solution.
Hence, terminals are partitioned according to $\cP$ as well.
The reduction runs in polynomial time and the parameter is unchanged so
we obtain the desired running time via Proposition~\ref{pro:solveMWC}.
\end{proof}

\subsection{Pair Partition Cut}
\label{ssec:ppc}

For $\mintwolin$ over arbitrary Euclidean domains, our reduction 
leads to a more general graph separation problem.
Given a graph $G$ with a set of terminals $T \subseteq V(G)$,
a \emph{(disjunctive) pair cut request} is a tuple $(\{s,u\}, \{t,v\})$
where $s,t \in T$ and $u,v \in V(G)$.
A cut $X \subseteq E(G)$ \emph{fulfills} $(\{s,u\}, \{t,v\})$
if $G - X$ does not contain an $\{s,u\}$-path or a $\{t,v\}$-path.

\pbDefP{Pair Partition Cut}
{An undirected graph $G$ with positive integer 
edge weights $w_G : E(G) \rightarrow \naturals^+$, 
a set of vertices (terminals) $T \subseteq V(G)$,
a partition $\cP$ of $T$, 
a set $\cF$ of pair cut requests,
and an integer $k$.}
{$k$.}
{Is there a $\cP$-cut $X \subseteq E(G)$ of total weight at most $k$
that fulfills every pair cut request in $\cF$?}

We prove that this problem is in \FPT by casting it into the
constraint satisfaction framework.
A constraint satisfaction problem (CSP) is defined by a constraint language $\Gamma$,
which is a set of relation over a domain $D$.
A relation of arity $r$ is a subset of $D^r$.
An instance $I = (V, C)$ of $\csp{\Gamma}$ 
is a set of variables $V$ and a set of constraints $C$
of the form $R(v_1, \dots, v_r)$, where $R \in \Gamma$ is a relation of arity $r$.
The instance $I$ is consistent if it admits an assignment $\varphi : V(C) \rightarrow D$
that satisfied every constraint in $C$ 
i.e. $(\varphi(v_1), \dots, \varphi(v_r)) \in R$ holds for all constraints.
In the parameterized version $\mincsp{\Gamma}$ the input is an instance
$I = (V, C)$ of $\csp{\Gamma}$ together with a weight function $w_C : C \rightarrow \naturals^+$
and the parameter $k \in \naturals^+$,
and the goal is to check whether there is a subset 
$X \subseteq C$ of equations with total weight at most $k$
such that $(V, C \setminus X)$ is consistent.

For the intuition behind the reduction, 
consider an instance of \PPC with a solution $X$.
Assume without loss of generality that $G$ is connected.
Since $X$ contains at most $k$ edges, removing $X$ splits $G$ 
into at most $k+1$ connected components.
Enumerate connected components
of $G - X$ with integers from $0$ to $k$
so that terminals from subset $B_i$ of $\cP$
are in the $i$th connected component.
This is possible since $X$ is a $\cP$-cut.
We define a function 
$\phi : V(G) \rightarrow \range{0}{k}$
such that $\phi(x) = i$ whenever $x$ belongs to $i$th component of $G - X$.
Then for every pair cut request $(\{s,u\}, \{t,v\})$ with $s \in B_i$ and $t \in B_j$,
we either have $\phi(u) \neq i$ or $\phi(v) \neq j$.
This reasoning suggests that all requirements
of \PPC can be encoded using 
the following constraint language $\Gamma_k$
with domain $\range{0}{k}$ 
and relations:
\begin{itemize}
 \item unary relations $(x = i)$ for all $0 \leq i \leq k$, 
 \item binary equality relation $(x = y)$, and
 \item binary relation $(x \neq i) \lor (y \neq j)$ for all $1 \leq i,j \leq k$.
\end{itemize}
To solve CSP$(\Gamma_k)$,
we define another constraint
language $\Gamma'_k$ with domain $\{0,1\}$ and relations:
\begin{itemize}
  \item $(x = 0)$, $(x = 1)$,
  \item $R_k(x_1,y_1,\ldots,x_k,y_k) \equiv \bigwedge_{1 \leq i \leq k} (x_i=y_i) \land \bigwedge_{1 \leq i < j \leq k} (\neg x_i \lor \neg x_j)$,
  \item $(\neg x \lor \neg y)$.  
\end{itemize}

\begin{theorem}[Section 5.7 in \cite{KimKPW21flow-arXiv}]  The problem $\mincsp{\Gamma'_d}$ is fpt parameterized by $\ell=d+c$ where $c$ is total solution cost.
\end{theorem}

The running time of $\mincsp{\Gamma'_k}$ is significant even
though it is fpt with parameter $\ell$.
It is not given explicitly by Kim et al.~\cite{KimKPW21flow},
but works out to $2^{\ell^{b}}$ where
$b \geq 12$~\cite[Lemma~6.14]{KimKPW21flow-arXiv}.
Since this is greater than any other running time contribution in this
paper, we treat this as a function $T(\ell)$ and give our other running
time bounds (where appropriate) in terms of $T(\ell)$.

While there is a directed reduction from
\PPC{} to $\mincsp{\Gamma'_k}$, we 
regard the following two-step reduction clearer
and more readable, and the intermediate problem
being an interesting example of a fixed-parameter
tractable \textsc{MinCSP}.

\begin{theorem} \label{thm:pair-mc-is-fpt}
\PPC is in \FPT.
\end{theorem}
\begin{proof}
First, we spell out the reduction from \PPC
to $\mincsp{\Gamma_k}$.
Given an instance $(G, w_G, T, \cP, \cF, k)$ of \PPC, we
construct an instance $((V,C),w,k)$ of
$\mincsp{\Gamma_k}$.
Let $V = V(G)$ denote the set of variables. 
We define the set of constraints 
$C$ and the weight function $w$ as follows. 
Enumerate subsets in $\cP$ as $B_1, \dots, B_m$
and for every subset $B_i$, 
add the constraints $(t = i)$ for all $t \in B_i$ 
of weight $k + 1$.
For every edge $\{u,v\} \in E(G)$,
add the constraint $(u = v)$ of weight $w_G(\{u,v\})$.
Finally, for every pair cut request
$(\{u,s\}, \{v,t\})$ in $\cF$ with $s \in B_i$
and $t \in B_j$, add the constraint
$(u \neq i) \lor (v \neq j)$ of weight $k + 1$.
Clearly, the reduction can be carried out in polynomial time.
A solution $X$ to $((V,C),w,k)$ may only contain
equality equations because every other constraint
is assigned weight $k+1$,
and $\{ \{u,v\} \in E(G) \mid (u = v) \in X \}$
is a $\cP$-cut in $G$ that fulfills $\cF$.
To obtain a solution to $((V,C),w,k)$
from a solution to the \PPC instance, one may follow the same steps in the opposite direction.

We continue by reducing $\mincsp{\Gamma_k}$
to $\mincsp{\Gamma'_k}$.
Given an instance $I = ((V, C),w,k)$ 
of the former problem, we produce
an equivalent instance $I' = ((V', C'),w',k)$
of the latter, while keeping the parameter unchanged.
To this end, introduce variables $v^{(i)}$
for every $v \in V$ and $i \in \range{1}{k}$.
Intuitively, setting $v^{(i)} = 1$
corresponds to assigning value $i$ to $v$,
while setting $v^{(i)} = 0$ for all $i \in \range{1}{k}$ 
corresponds to assigning $0$ to $v$.
To ensure that $v^{(i)} = 1$ for at most one
value of $i$, add constraints 
$(\neg v^{(i)} \lor \neg v^{(j)})$
of weight $k + 1$ for all $1 \leq i < j \leq k$.
Every constraint $c$ in $C$ is replaced
by constraints in $C'$ of the same weight as follows:

\begin{enumerate}
\item
if $c$ is $t = i$ for $i \in \range{1}{k}$, then
add $t^{(i)} = 1$ to $C'$,

\item
if $c$ is $t = 0$, then add $t^{(i)} = 0$ for all $i \in \range{1}{k}$ to $C'$, each of weight $w(c)$,

\item
if $c$ is $(u = v)$, then add 
${R_k(u^{(1)}, v^{(1)}, \dots, u^{(k)}, v^{(k)})}$ to $C'$, and

\item
if $c$ is $(s \neq i) \lor (t \neq j)$, then add $(\neg s^{(i)} \lor \neg t^{(j)})$ to $C'$.

\end{enumerate}
This concludes the reduction.

Suppose $\phi$ is an assignment to $(V,C)$.
Define $\phi'$ by letting 
$\phi'(v^{(i)}) = 1$ if $\phi(v) = i$ for some $i \geq 1$,
and $\phi(v^{(i)}) = 0$ otherwise.
By construction, $\phi$ and $\phi'$ leave 
constraints of the same total weight unsatisfied.
Hence, if the set of constraints unsatisfied by $\phi$
is a solution to $I$, then the set of constraints
unsatisfied by $\phi'$ is a solution to $I'$.
The same argument works in the opposite direction:
given an assignment $\rho'$ to $(V',C')$,
define assignment $\rho$ to $(V,C)$
by letting $\rho(v) = i$ if 
$\rho'(v^{(i)}) = 1$ for some $i \in \range{1}{k}$,
and $\rho(v) = 0$ otherwise.
Constraints of the type 
$(\neg v^{(i)} \lor \neg v^{(j)})$
ensure that $\rho$ is well-defined.
Moreover, the total weight of constraints
unsatisfied by $\rho$ and $\rho'$ is the same.
Thus, the reduction is correct,
and the theorem follows.
\end{proof}

\section{Algorithm for Euclidean Domains}
\label{sec:edom-algorithm}

We let $\D = (D; +, \cdot)$ denote a Euclidean domain throughout this
section.
Our goal is to present an fpt algorithm for $\minlin{2}{\D}$.
We start by reviewing basic definitions 
and facts about Euclidean domains in Section~\ref{ssec:basic-edom}.
In Section~\ref{ssec:lin2graphs} we develop
a polynomial-time algorithm for $\lin{2}{\D}$
and prove several useful lemmas along the way,
building an understanding of the problem.
We note that polynomial-time algorithms for
$\lin{r}{\D}$ are known for arbitrary $r \in \naturals$
when $\D$ is finite~\cite[Section 6]{Arvind:Vijayaraghavan:cc2010} (in fact, this is true for arbitrary finite rings), the ring of integers~\cite{Kannan:Bachem:sicomp79}
or the ring of univariate polynomials over 
$\rationals$~\cite{kannan1985solving}. However,
we are unaware of such results for general Euclidean domains, even when $r=2$. 
The next three sections follow the common steps of
compression, cleaning, and cutting:
we simplify the problem by applying 
iterative compression in Section~\ref{ssec:edom-compression},
then simplify it even further by applying 
the important balanced subgraph machinery in Section~\ref{ssec:edom-cleaning},
and finally reduce the resulting problem to \PPC in Section~\ref{ssec:edom-alg-def},
giving an overview of the whole algorithm.
Finally, in Section~\ref{ssec:edom-correctness-time}
we prove correctness of the algorithm and analyze its time complexity.

\subsection{Basics of Euclidean Domains}
\label{ssec:basic-edom}

A \emph{Euclidean domain} is an abstract algebraic structure generalizing properties of the integers.
Informally, it is a commutative ring with integer division.
Formally, $\D = (D; +, \cdot)$ is a Euclidean domain if it is 
an {\em integral domain} equipped with a {\em Euclidean function}.
An integral domain is a commutative ring of size at least two where 
the product of any pair of nonzero elements is itself 
nonzero---that is, $\D$ does not contain a zero divisor.
A Euclidean function on $\D$ is a function 
$f : D \rightarrow \naturals_0$ such that $f(0) = 0$
and for any $a,b \in D$ where $b \neq 0$, 
there exist $q,r \in D$ such that $a = bq + r$ and $f(r) < f(b)$. 
One may view $q$ as a {\em quotient} and $r$ as a {\em remainder},
and write $a \equiv r \mod b$ to denote that $r$ is a remainder of 
integer division of $a$ by $b$.
All fields and the ring of integers $\integers$ are Euclidean domains; 
this follows from choosing the Euclidean function to be 
$f(x)=1$ for all $x \neq 0$
and $f(x)=\abs{x}$, respectively.
Further examples include Gaussian integers $\integers[i]$,
Eisenstein integers~$\integers[\omega]$ where $\omega$ is a primitive non-real cubic root of unity, 
the ring of polynomials $\F[x]$ over a field $\F$, and many more.
When working with Euclidean domains, we assume that they are
\emph{effective} i.e. $\D$ admits a reasonable representation of elements
such that basic operations
(addition, subtraction, multiplication, computing quotients and remainders)
requires polynomial time in the bit-size of the operands.
In addition, we require the following property.
Given an element $d \in D$, let $\norm{d}$ denote
the number of bits required to represent $d$.

\begin{property} \label{property:edom-product}
  In an effective ring $\D$, there is a polynomial function $p$ such that
  $\norm{d_1 \cdot ... \cdot d_m} \leq p(\norm{d_1} + ... + \norm{d_n})$ for arbitrary 
  $d_1, \dots, d_n \in \D$.
\end{property}

This is a natural requirement since otherwise we cannot
compute (or even write down) satisfying assignments to simple consistent instances of $\lin{2}{\D}$ like
$\{x_1 = d_1 x_2, x_2 = d_2 x_3, \dots, x_{n-1} = d_n x_n\} \cup \{x_n = 1\}$ in
polynomial time, we cannot perform efficient Gaussian elimination etc.
In many cases (including all examples of Euclidean domains given above), 
$p$ is the identity polynomial
i.e. representing the product of elements requires
at most as many bits as representing them individually.

It is important to note that quotients and remainders
are not unique in $\D$.
For a simple example, consider $\D = \integers$
with Euclidean function $f(x) = \abs{x}$,
let $a = 9$ and $b = 4$,
and note that $9 = 4 \cdot 2 + 1$
and $9 = 4 \cdot 3 + (-3)$.
Since $\abs{1} < \abs{4}$ and $\abs{-3} < \abs{4}$,
both $q = 2$, $r = 1$ and $q = 3$, $r = -3$ are valid quotient-remainder pairs.
However, if we fix the remainder $r$,
then there is at most one value for~$q$
such that $a = b \cdot q + r$.
As a corollary of this observation, 
if $b$ divides $a$ (i.e. $r = 0$), then
the result of dividing $a$ by $b$ is unique.
In such cases we write $a / b$ to denote the unique quotient. 
A \emph{unit} is an element $u$ of $D$
that admits a multiplicative inverse i.e. 
there is an element $v$ of $D$ such that $u v = 1$.
For all $a, b \in D$,
a {\em greatest common divisor} $\gcd(a, b)$
is a maximal (with respect to $f$)
element of $D$ that divides both $a$ and $b$.
If $\gcd(a, b)$ is a unit, then
$a$ and $b$ are \emph{co-prime}.
A {\em least common multiple} $\lcm(a, b)$
is a minimal (with respect to $f$)
element of $D$ that is divisible by~$a$ and~$b$.
Observe that while $\gcd(a, b)$ is not unique,
all greatest common divisors of $a$ and $b$ are
congruent up to multiplication by units:
if $g_1$ and $g_2$ are greatest common divisors
of $a$ and $b$, then there is a unit element $u$ in $D$
such that $a = b \cdot u$.
The same congruence holds for the least common multiples.
As a result, when discussing divisibility we can safely abuse notation by writing $\gcd(a, b)$ and $\lcm(a, b)$ to denote an arbitrary greatest common divisor or least common multiple of $a$ and $b$, respectively.
An analogue of the extended Euclidean algorithm works in effective Euclidean domains.

\begin{proposition}[{\cite[Theorem~4.10]{vonzurGathen:Gerhard:MCA}}] \label{prop:one-equation}
  An equation $ax + by = c$ with $a,b,c \in D$ has a solution in~$\D$
  if and only if $g=\gcd(a,b)$ divides $c$, and all
  satisfying assignments are of the form
  $(x_0 + (b / g) \cdot r,\allowbreak y_0 - (a / g) \cdot r)$ for some fixed $x_0,y_0 \in D$
  and arbitrary $r \in D$. Finally, there is a polynomial time algorithm 
  that checks this condition and computes $g$, $x_0$, and $y_0$.
\end{proposition} 

In light of Proposition~\ref{prop:one-equation}, we can assume that
the instances of $\minlin{2}{\D}$ that we are dealing with do not contain
inconsistent equations (since those can be removed in polynomial time during a preprocessing stage
where the parameter is decreased according to the weight of the equation).
Moreover, we may assume that in every equation $ax + by = c$ the
coefficients $a$ and $b$ are co-prime (since we may divide 
all coefficients by a $\gcd(a, b)$).
By further preprocessing, we may assume that equations of the form $0 \cdot x + 0 \cdot y = 0$
do not appear in the instances: since they are satisfied by any assignment,
they can be removed in advance without affecting the parameter.

We use the following distributive property 
of $\gcd$ and $\lcm$:

\begin{sloppypar}
\begin{proposition} \label{prop:edom-lattice}
  Let $\D$ be a Euclidean domain and 
  let $a_1, \dots, a_n, b \in D$.
  Then every
  $\lcm(\gcd(a_1, b), \dots, \gcd(a_n, b))$ is congruent to every
  $\gcd(\lcm(a_1,\dots,a_n), b)$ up to multiplication by a unit element.
\end{proposition}
\end{sloppypar}

A proof of this statement for $\D = \integers$ and $n = 2$
is a common exercise in number theory and algebra textbooks
(see e.g. Exercise~$23 \varepsilon$~in~\cite{clark1984elements}
or Exercise~III.3~in~\cite{landau2021elementary}).
A proof can be found in~\cite{proofwiki}.
It generalizes in a straightforward way to all Euclidean domains
and all $n \in \naturals^+$ by noting that 
the elements of a Euclidean domain admit unique factorization
up to multiplication by units.

\subsection{$\twolin$ over Euclidean Domains}
\label{ssec:lin2graphs}

The main goal of this section is to present a polynomial-time
algorithm for $\lin{2}{\D}$.
Our approach
exploits
a particular graph (known as the {\em primal} or {\em Gaifman} graph) that describes the structure of 
$\lin{2}{\D}$ instances.
Thus, we begin by presenting algorithms for various graphs
such as paths, stars, and acyclic graphs,
where we use a result that resembles the Chinese Remainder Theorem (Lemma~\ref{lem:edom-acyclic-consistent}). 
Then we extend these results to \emph{flexible} instances, which can be viewed as
a generalization of acyclic instances.
Finally, we use the algorithm for flexible instances as the basis for a
polynomial-time algorithm
that checks consistency of general $\lin{2}{\D}$ instances (Lemma~\ref{lem:edom-algo}).
A useful simplification in our proofs is provided
by \emph{homogenization}---a procedure that transforms the solution
space while preserving the primal graph of the instance (Lemma~\ref{lem:homogenise}).
This technique will be used frequently in this and following sections.

Let $S$ be an instance of $\lin{2}{\D}$.
We associate a {\em primal graph} with $S$:
vertices of this graph correspond to the variables in $V(S)$, and
two vertices $x$ and $y$ are connected by an edge 
if $S$ contains an equation over $x$ and $y$.
We can think of an instance of $\lin{2}{\D}$ as a graph
with edges $\{x,y\}$ labelled by equations over $x$ and $y$. 
We may (without loss of generality) assume that the graph does not have self-loops
by introducing a zero variable $z_0$ and an auxiliary variable $z'_0$, 
adding equations $z'_0 + z_0 = 0$ and $z'_0 - z_0 = 0$,
and replacing single-variable equations $ax = b$ with $ax - z_0 = b$. 
Thus, we assume that the zero variable $z_0$ is available
in every instance of $\lin{2}{\D}$, and in $\minlin{2}{\D}$
equations $z'_0 + z_0 = 0$ and $z'_0 - z_0 = 0$
are given weight $k + 1$.
We use graph-related terminology
(such as connectedness, paths, cycles etc.)
to describe the structure of $S$ while having the primal graph in mind.

One useful trick to simplify consistent instances of $\lin{2}{\D}$
is \emph{homogenization}.
An equation $ax + by = c$ is \emph{homogeneous} if $c = 0$ and
an instance of $\lin{2}{\D}$ is homogeneous if every equation in the instance is homogeneous.
Note that any such system is consistent since it is satisfied by the all-zero assignment.
We show that by applying an invertible affine transformation 
to the solution space, we can turn every consistent instance $S$ of $\lin{2}{\D}$
into a homogeneous system with the same primal graph.
Define a mapping $\Phi$ that acts on every variable $x \in V(S)$ by setting
$x \mapsto a_x x' + b_x$ for some $a_x, b_x \in D$.
We refer to $\Phi$ as a \emph{variable substitution for $S$},
and write $\Phi(S)$ to denote the instance of $\lin{2}{\D}$ 
obtained by substituting every variable $x$ with $a_x x' + b_x$.
A variable substitution is \emph{homogenizing} if $\Phi(S)$ is homogeneous.

\begin{lemma} \label{lem:homogenise}
  Every consistent instance of $\lin{2}{\D}$ admits 
  a homogenizing variable substitution.
\end{lemma}
\begin{proof}
Let $S$ be an instance of $\lin{2}{\D}$
satisfied by assignment $\varphi$.
Define $\Phi$ as $x \mapsto x' + \varphi(x)$ for all $x \in V(S)$.
Note that $\Phi$ is reversible by subtracting $\varphi(x)$.
Consider an equation $ax + by = c$ in $S$.
Note that $a\varphi(x) + b\varphi(y) = c$ since $\varphi$ satisfies the equation.
Its counterpart in $\Phi(S)$ is
$a(x' + \varphi(x)) + b(y' + \varphi(y)) = c$, which simplifies to $ax' + by' = 0$.
The right hand side in the obtained equation is $0$.
Thus, the variable substitution $\Phi$ is homogenizing.
\end{proof}

Now consider a path $P$ of length $\ell-1$ connecting variables $x$ and $y$ in $S$
i.e. a system of $\ell-1$ equations over
$\ell$ distinct variables $p_1, \dots, p_\ell$, where $x = p_1$ and $y = p_\ell$, 
with one equation relating $p_i$ and $p_{i+1}$ for all $i \in \{1,\ldots,\ell-1\}$.
If $\ell = 2$, then $P$ contains a single equation.
Otherwise, we may eliminate intermediate variables
to obtain an equation over $x$ and $y$ by recursively picking
the first two equations in $P$, say, 
$ap_1 + bp_2 = c$ and $a'p_2 + b'p_3 = c'$,
and taking their linear combination 
$a'(ap_1 + bp_2) - b(a'p_2 + b'p_3) = a'c - bc'$,
which simplifies to $(a'a)p_1 - (b'b)p_3 = a'c - bc'$.
We say that $P$ \emph{implies} the final equation over $x$ and $y$
obtained after eliminating all intermediate variables.
This equation is denoted by $e_P$.
For example, let $P$ be a path with 
two equations over $\integers$: $x - 2z = 2$ and $z - y = 1$. 
The equation implied by $P$ is obtained by eliminating $z$
so it is $x - 2y = 4$.
The following observation implies that variable elimination is safe.

\begin{observation} \label{obs:var-elim-safe}
  Every assignment that satisfies $P$ also satisfies $e_P$.
\end{observation} 

We say that an instance $S$ of $\lin{2}{\D}$ is \emph{flexible}
if for every pair of variables $x,y \in V(S)$,
every $\{x,y\}$-path in $S$ implies \emph{equivalent} equations on $x$ and $y$,
i.e. equations with the same set of satisfying assignments.
Otherwise, we say that $S$ is \emph{rigid}.
If $S$ is flexible, we write $e_{xy}(S)$ to denote the equation
implied by the $\{x,y\}$-paths in $S$.
A simple example of flexible instances are acyclic instances.
In the following lemma, we present 
a simple criterion for checking the consistency of such instances.
We start with a useful observation.
For a flexible instance $S$ and any $x \in V(S)$, we define the instance
$\STS(S, x) = \{ e_{xy}(S) \mid y \in V(S) \setminus \{x\} \}$ of $\lin{2}{\D}$.

\begin{lemma} \label{lem:acyclic=star}
  Let $S$ be a connected, acyclic instance of $\lin{2}{\D}$. For any $x \in V(S)$, $S$ and $\STS(S, x)$ have the same set of satisfying assignments.
\end{lemma}
\begin{proof}
By Observation~\ref{obs:var-elim-safe}, an assignment that
satisfies $S$ also satisfies the equations implied by 
the paths in $S$. Consequently, it satisfies $\STS(S, x)$ for any $x$.
Now, suppose $\varphi$ is a satisfying assignment to $\STS(S, x)$
and consider an equation $e \in S$ over variables $y$ and $z$.
It suffices to show that $\varphi$ satisfies $e$.
Clearly, this holds if $y = x$ since then $e \in \STS(S, x)$.
Otherwise, by the construction of equations implied by the paths,
the equation $e$ can be written as
a linear combination of $e_{xy}(S)$ and $e_{xz}(S)$.
Since these two equations are present in $\STS(S, x)$,
$\varphi$ satisfies them, and hence also satisfies $e$.
\end{proof}

We can now present the consistency criterion.  
  
\begin{lemma} \label{lem:edom-acyclic-consistent}
  An acyclic instance of $\lin{2}{\D}$ is consistent if
  and only if it does not contain an inconsistent path.
\end{lemma}
\begin{proof}
One direction of the proof follows by Observation~\ref{obs:var-elim-safe}:
existence of an inconsistent path implies inconsistency of the whole system.
For the opposite direction we proceed by induction on the number of equations.
If a system contains only one equation, 
then the claims follow by Proposition~\ref{prop:one-equation}.
Now consider a system $S$ with $n + 1$ equations
where every path is consistent.
If $S$ contains more than one component, 
then the lemma follows 
by induction in each component.
If $S$ is connected (i.e. $S$ is a tree), 
then pick a leaf $z$ of $S$ and assume $x$ is the neighbour of $z$ in $S$.
By induction, the subtree $S' := S[V(S) \setminus \{z\}]$ without $z$ is consistent, 
so Lemma~\ref{lem:homogenise} applies and
there is a homogenizing substitution $\Phi$ to $S'$.
Note that when $\Phi$ is applied to $S$,
all equations except (possibly) 
the one involving $x$ and $z$ are homogenized.

Assume the variables in $V(S) \setminus \{x,z\}$ are $y_1,\dots,y_n$.
Then the equation $e_{xy_i}(\Phi(S'))$ can be written as 
$a_i x = b_i y_i$ for some co-prime $a_i, b_i \in D$.
Let $B = \lcm(b_1, \dots, b_n)$.
An assignment satisfying $\Phi(S')$ 
can be obtained 
by setting $x \mapsto B \cdot r$ and
$y_i \mapsto a_i (B / b_i) \cdot r$
for any $r \in D$.
The assignment clearly satisfies $\STS(\Phi(S'), x)$ so
Lemma~\ref{lem:acyclic=star} implies that it also satisfies $\Phi(S')$.

We now obtain an assignment for all equations in $S$.
Let the equation over $x$ and $z$ in $\Phi(S)$ be 
\begin{equation} \label{eq:xzequation}
  a \cdot x + b \cdot z = c.
\end{equation}
Since it is consistent, we may assume that $a$ and $b$ are co-prime by Proposition~\ref{prop:one-equation}.
We claim that the following holds by consistency of all paths in $\Phi(S)$.

\begin{claim} \label{claim:gcd}
  $\gcd(b, b_i)$ divides $c$ for all $i \in \{1,\ldots,n\}$.
\end{claim}
\begin{claimproof}
Consider $e_{zy_i}(\Phi(S))$.
Since the $\{z,y_i\}$-path in $\Phi(S)$ goes through $x$, 
the equation can be obtained by cancelling out $x$ 
from $ax + bz = c$ and $a_i x = b_i y_i$. 
Thus, we get 
\begin{equation} \label{eq:edom-acyclic-consistent:1}
a b_i \cdot y_i + a_i b \cdot z = a_i c.
\end{equation}
Assume without loss of generality that $a$ and $a_i$
are co-prime (otherwise divide all coefficients by $\gcd(a, a_i)$).
By assumption, Equation~\eqref{eq:edom-acyclic-consistent:1} is consistent and, by Proposition~\ref{prop:one-equation}, 
$\gcd(a b_i, a_i b)$ divides $a_i c$.
Note further that the pairs $(a,b)$, $(a_i,b_i)$ and $(a, a_i)$ are co-prime.
Hence, $\gcd(a b_i, a_i b) = \gcd(b, b_i)$.
Since $a_i, b_i$ are co-prime,
$\gcd(b, b_i)$ does not divide $a_i$ and it only divides $c$.
\end{claimproof}

To find a solution to the system $S$,
substitute in $B \cdot r$ instead of $x$ into
Equation~\ref{eq:xzequation}
obtaining $aB \cdot r + b \cdot z = c$, 
where $r$ is a fresh variable.
We claim that this equation is consistent.
By Proposition~\ref{prop:one-equation},
it suffices to show that $\gcd(aB, b)$ divides $c$.
First, note that since $a,b$ are co-prime, 
an element is a $\gcd(aB, b)$ if and only if 
it is a $\gcd(B, b)$.
Now, let $\gcd(b, b_i)$ be $g_i$.
By the definition of $B$ and
Proposition~\ref{prop:edom-lattice}, 
every $\gcd(B, b)$ and $\lcm(g_1, \dots, g_n)$
are congruent up to multiplication by units.
Claim~\ref{claim:gcd} implies that $g_i$ divides $c$ for all $i \in \{1,\ldots,n\}$,
so by the definition of lcm,
$\lcm(g_1, \dots, g_n)$ also divides $c$.
Thus, by choosing an appropriate value for $r$,
we obtain an assignment that satisfies Equation~\ref{eq:xzequation}
and $\Phi(S')$ simultaneously.
Hence, the assignment satisfies $\Phi(S)$, and
the lemma follows by applying the
inverse variable substitution $\Phi^{-1}$.
\end{proof}

As an aside, a corollary of this lemma is that $\minlin{2}{\D}$ 
on acyclic instances is in \FPT,
since the problem can be reduced to solving
(the edge-deletion variant of) \textsc{Multicut} on trees~\cite{guo2005fixed}, 
where endpoints of every inconsistent path form a cut request.
We also have the following algorithmic corollary.

\begin{corollary} \label{cor:edom-solve-acyclic}
  There is a polynomial-time algorithm that checks
  whether an acyclic instance of $\lin{2}{\D}$ is consistent,
  and if so, computes a satisfying assignment.
\end{corollary}
\begin{proof}
Let $S$ be an acyclic instance of
$\lin{2}{\D}$ and
let $d_1,\dots,d_n$ be the 
coefficients appearing in~$S$.
The proof of Lemma~\ref{lem:edom-acyclic-consistent} is constructive --
it produces a concrete satisfying assignment to $S$.
The running time for constructing this solution is polynomial in $C n$,
where $n$ is the number of equations in the system
and $C$ is the maximum bit-size of an element from $\D$
computed during the course of the algorithm.
By Property~\ref{property:edom-product}, we have
$C \leq \norm{d_1 \cdot ... \cdot d_n} \leq 
p(\norm{d_1} + ... + \norm{d_n}) \leq p(\norm{S})$, 
where~$\norm{S}$ is the bit-size of $S$.
We conclude that the algorithm runs in polynomial time.
\end{proof}

The results for acyclic instances extend to flexible instances
by considering spanning forests. 

\begin{lemma} \label{lem:edom-flexible}
  A flexible instance $S$ of $\lin{2}{\D}$ is consistent if and only if it
  contains no inconsistent path.
  Moreover, if $S$ is connected, then $S$ and $\STS(S, x)$ 
  have the same set of satisfying assignments for every $x \in V(S)$.
\end{lemma}
\begin{proof}
Let $S$ be a flexible instance of $\lin{2}{\D}$ and let $T$ be a
spanning forest of $S$. 
We claim that~$S$ and $T$ have the same set of satisfying assignments,
and then the lemma holds by
Lemma~\ref{lem:acyclic=star} and
Lemma~\ref{lem:edom-acyclic-consistent}.
First, note that any assignment satisfying $S$ also satisfies $T$
since it is a subinstance of~$S$.
On the other hand, let $\varphi$ be a satisfying assignment for $T$,
and consider an equation $e$ in $S \setminus T$ with variables $x$, $y$.
It suffices to show that $\varphi$ satisfies $e$.
Since $S$ is flexible, $e$ is equivalent to~$e_{xy}(S)$.
Furthermore, $T$ is a spanning forest, so it contains a path
connecting $x$ and $y$, and the path implies an equation
equivalent to $e_{xy}(S)$ and $e$.
Hence, $\varphi$ satisfies $e$ and the lemma holds.
\end{proof}

Analogously to Corollary~\ref{cor:edom-solve-acyclic},
we have the following algorithmic result.

\begin{corollary} \label{cor:edom-solve-flexible}
  There is a polynomial-time algorithm that checks
  whether a flexible instance of $\lin{2}{\D}$ is consistent,
  and if so, computes a satisfying assignment.
\end{corollary}

Lemma~\ref{lem:edom-flexible} suggests an fpt algorithm for solving
$\minlin{2}{\D}$ for flexible instances:
find a spanning forest and reduce to \textsc{Multicut}
by adding endpoints of every inconsistent path 
in the forest as a cut request.
This is also an important stepping stone towards
checking consistency of any instance of $\lin{2}{\D}$.
We provide a full algorithm below.

\begin{lemma} \label{lem:edom-algo}
  There is a polynomial-time algorithm that checks
  whether an instance of $\lin{2}{\D}$ is consistent,
  and if so, computes a satisfying assignment.
\end{lemma}
\begin{proof}
Let $S$ be an instance of $\lin{2}{\D}$.
Without loss of generality, assume that it is connected,
and compute a spanning tree $T$.
We proceed by checking whether $S$ is flexible or not.
To do so, we consider the equations $e \in S - T$.
Assume that $e$ equals $a_1 x + b_1 y = c_1$ and
let $e_{xy}(T)$ equal $a_2 x + b_2 y = c_2$.
We want to check whether these two equations are equivalent.
To do so, we multiply the first one with $a_2$, the second one with $a_1$,
compute their difference and obtain $(b_1 a_2 - a_1 b_2) y = c_1 a_2 - a_1 c_2$.
For conciseness, let $A = b_1 a_2 - a_1 b_2$,
$B = c_1 a_2 - a_1 c_2$, and consider four cases:
\begin{itemize}
  \item If $A = 0$ and $B \neq 0$, then $Ay = B$ is inconsistent.
  \item If $A = 0$ and $B = 0$, then
  $Ay = B$ is satisfied by assigning any value to $y$.
  \item If $A \neq 0$ and $A$ does not divide $B$,
    then $Ay = B$ is inconsistent.
  \item If $A \neq 0$ and $A$ divides $B$,
  then $Ay = B$ is only satisfied by setting $y$ to $B / A$.
\end{itemize}
If $Ay = B$ is inconsistent, then no assignment can satisfy both $e$ and $e_{xy}(T)$,
hence $S$ is inconsistent.
If $A = B = 0$, then $e$ and $e_{xy}(T)$ are equivalent, and we proceed to the next equation in $S-T$.
Finally, if $Ay = B$ has only one satisfying assignment (namely $y \mapsto B / A$),
we assign this value to $y$ and propagate to the rest of the instance,
and then check whether the obtained assignment satisfies every equation.
Thus, the only case we need to consider further is when $S$ is flexible.
This case is handled by Corollary~\ref{cor:edom-solve-flexible} and this concludes the proof.
\end{proof}

\subsection{Iterative Compression}
\label{ssec:edom-compression}

We reduce $\minlin{2}{\D}$ to a simpler problem
by combining homogenization with {\em iterative compression}.
The latter method uses a compression routine that takes
a problem instance together with a solution as an input,
and either calculates a smaller solution 
or verifies that the provided one has minimum size. 
An optimal solution is then computed by 
iteratively building up the instance while 
improving the solution at each step.
If the compression routine runs in fpt time, 
then the whole algorithm also runs in fpt time.
A more comprehensive treatment of the method 
can be found in~\cite[Chapter~4]{cygan2015book}.
We use it to provide a reduction from $\minlin{2}{\D}$ to the following problem:

\pbDefP{Disjoint $\minlin{2}{\D}$ ($\DML(\D)$)}
{An instance $S$ of $\lin{2}{\D}$ with 
positive integer equation weights $w_S : S \rightarrow \naturals^+$, 
an inclusion-wise minimal set $X \subseteq S$ such that $S - X$ is homogeneous,
and an integer $k$ such that $w_S(X) \leq k+1$.}
{$k$.}
{Is there a set $Z \subseteq S - X$ of weight at most $k$ such that 
$S - Z$ is consistent?}

\begin{lemma} \label{lem:dml}
  If $\DML(\D)$ is solvable in $\bigoh^*(f(k))$ time, 
  then $\minlin{2}{\D}$ is solvable in $\bigoh^*(2^k f(k))$ time.
\end{lemma}
\begin{proof}
Let $I = (S, w_S, k)$ be an instance of $\minlin{2}{\D}$.
In this context it is simpler to view equations as a multiset
$S'$ where every equation $e \in S$ is present with multiplicity $w_S(e)$.
Then by iterative compression, we may assume that apart from the input $I$,
we also have access to a multiset~$X$ such that $\abs{X} = k + 1$
and $S' - X$ is consistent.

Suppose $Z$ is an optimal solution to $S'$.
To reduce to $\DML(\D)$, we branch on the possible intersections
$Y = X \cap Z$ of the incoming solution with the optimal solution. 
Since there are $2^{\abs{X}} = 2^{k+1}$ options,
the branching step requires fpt time.
For every guess $Y$, consider the multisets $S' - Y$ and $X - Y$,
and convert them into sets $S_Y$ and $X_Y$, respectively,
defining the weight function $w_Y$ so that 
$w_Y(e)$ for all equations $e$ is the multiplicity of $e$ in $S_Y$.
Note that by definition $S_Y - X_Y$ is consistent,
so we may apply a homogenizing variable substitution to it
by Lemma~\ref{lem:homogenise}.
Finally, set the parameter to $k_Y = k - \abs{Y}$.
We obtain an instance $(S_Y, w_Y, k_Y, X_Y)$ of $\DML(\D)$.
If this instance has a solution, then combining that solution with $Y$
yields a solution to the instance $I$ of $\minlin{2}{\D}$.
On the other hand, if there is no solution for any option $Y$,
then by exhaustion $I$ is a no-instance.
Since we branch in $2^{k+1}$ directions,
and in each branch we solve an instance of $\DML(\D)$
with parameter bounded from above by $k$, 
we obtain the total running time of $\bigoh^*(2^k f(k))$.
\end{proof}

\subsection{Graph Cleaning for Euclidean Domains}
\label{ssec:edom-cleaning}

Lemma~\ref{lem:edom-flexible} provides us with
a good idea of how to solve $\minlin{2}{\D}$ restricted to acyclic and flexible instances.
To approach the general solution, we now need to consider cycles that make instances rigid.
If a consistent cycle is flexible, we say that it is an \emph{identity} cycle.
If a consistent cycle is consistent, but not flexible, then it is a \emph{non-identity} cycle.
There is an alternative characterization of
identity and non-identity cycles in terms of the number
of satisfying assignments.

\begin{lemma} \label{lem:cycles-and-assignments}
  A consistent cycle in $\lin{2}{\D}$ is identity
  if and only if it admits more than one satisfying assignment,
  while a consistent cycle is non-identity
  if and only if it admits a unique satisfying assignment.
\end{lemma}
\begin{proof}
Let $C$ be a consistent instance of $\lin{2}{\D}$ that is a cycle.
By Lemma~\ref{lem:homogenise}, we may assume 
without loss of generality that $C$ is homogeneous.
Then, $C$ is satisfied by the all-zero assignment.
To prove the lemma it suffices to show that
$C$ admits a non-zero assignment if and only if it is identity.

On the one hand, suppose that $C$ is an identity cycle.
Pick an arbitrary equation $e \in C$. 
Note that $P := C \setminus \{e\}$ is a path.
Let $\abs{P} = m$ and assume that the equations on $P$
are $a_i x_i = b_i x_{i+1}$ for $i \in \{1,\ldots,m\}$,
where $a_i, b_i \in D \setminus \{0\}$ and $x_i$ is a variable.
We define the assignment $\varphi$ using a particular \emph{product construction}:
set $\varphi(x_1) = b_1 \cdot b_2 \cdot ... \cdot b_m$, and
$\varphi(x_{i+1}) = (\varphi(x_i) / b_i) \cdot a_i$ for all $i \in \{1,\ldots,m\}$.
In other words, the value $\varphi(x_{i+1})$ is obtained from $\varphi(x_i)$ 
by replacing the factor~$b_i$ in the product by $a_i$.
Consequently, $\varphi(x_{m+1}) = a_1 \cdot a_2 \cdot ... \cdot a_m$.
Since all coefficients $a_i, b_i$ are nonzero, $\varphi$ is a nonzero assignment and
it is easy to verify that $\varphi$ satisfies $P$.
By Observation~\ref{obs:var-elim-safe}, 
it also satisfies the implied equation $e_P$.
Since equations $e$ and $e_P$ are equivalent,
assignment $\varphi$ satisfies $e$ and, therefore, 
it satisfies $P \cup \{e\} = C$.

On the other hand, suppose that $C$ is a non-identity cycle.
By definition, there are variables $x,y \in V(C)$
such that the $\{x,y\}$-paths $P_1$ and $P_2$ forming $C$ imply 
two non-equivalent equations $a_1 x = b_1 y$ and $a_2 x = b_2 y$, respectively.
Multiplying the first equation by $a_2$ and the second by $a_1$,
we obtain the same coefficient in front of $x$.
Since the equations are not equivalent,
the coefficients in front of $y$ must differ
i.e. $b_1 a_2 \neq a_1 b_2$.
Hence, any assignment satisfying $C$
also satisfies $(b_1 a_2 - a_1 b_2) \cdot y = 0$,
which can only be satisfied by setting $y$ to $0$.
The zero value propagates to all remaining variables,
so $C$ is only satisfied by the all-zero assignment.
\end{proof}

The following results allows us to use the graph cleaning machinery
to remove non-identity cycles.

\begin{lemma} \label{lem:non-id-has-theta}
  Let $G_S$ be the primal graph of a consistent instance $S$ of $\lin{2}{\D}$
  and $\B_S$ be the set of identity cycles in $S$.
  Then $(G_S, \B_S)$ is a biased graph.
\end{lemma}
\begin{proof}
By Lemma~\ref{lem:homogenise}, it suffices to consider
a homogeneous instance $S$.
We want to verify that theta property holds for the family of unbalanced cycles in $(G_S, \B_S)$.
To this end, let $P$, $Q$, $R$ be three internally vertex-disjoint $\{x,y\}$-paths in $G_S$,
and assume $P \cup R$ is a non-identity cycle.
We claim that equations $e_P$ and $e_R$ are inequivalent.
Then equation $e_Q$ cannot be equivalent to both $e_P$ and $e_R$.
This implies that either $P \cup Q$ or $Q \cup R$ is non-identity, and the lemma follows.

To prove the claim, assume towards contradiction
that equations $e_P$ and $e_R$ are equivalent,
and $e_P$ is $ax = by$.
Using the product construction from the proof of 
Lemma~\ref{lem:cycles-and-assignments},
define nonzero assignments $\varphi_P$ and $\varphi_R$
satisfying all equations in $P$ and $R$, respectively.
Note that by Observation~\ref{obs:var-elim-safe},
they also satisfy $ax = by$ i.e.
\begin{align}
  \label{eq:phiP} a \cdot \varphi_P(x) &= b \cdot \varphi_P(y), \\
  \label{eq:phiR} a \cdot \varphi_R(x) &= b \cdot \varphi_R(y).
\end{align}
Therefore, after multiplying~\eqref{eq:phiP} by $\varphi_R(x)$ and~\eqref{eq:phiR} by $\varphi_P(x)$, we obtain: 
\begin{align}
  \label{eq:phiPmult} a \cdot \varphi_P(x) \cdot \varphi_R(x) &= b \cdot \varphi_P(y) \cdot \varphi_R(x), \\
  \label{eq:phiRmult} a \cdot \varphi_R(x) \cdot \varphi_P(x) &= b \cdot \varphi_R(y) \cdot \varphi_P(x).
\end{align}
Since the right hand sides of \eqref{eq:phiPmult} and \eqref{eq:phiRmult} 
are equal, we may equate the left hand sides and thus obtain
\[ \varphi_P(y) \cdot \varphi_R(x) = \varphi_R(y) \cdot \varphi_P(x). \]
This equation allows us to define a nonzero assignment 
$\varphi_{PR}$ that satisfies $P \cup R$
by scaling $\varphi_P$ and $\varphi_R$ so that
they agree on the values of $x$ and $y$, namely let
\[
  \varphi_{PR}(z) = 
  \begin{cases}
     \varphi_P(z) \cdot \varphi_R(x) & \text{if } z \in P, \\
     \varphi_R(z) \cdot \varphi_P(x) & \text{if } z \in R.
  \end{cases}
\]
Since $P \cup R$ admits the nonzero satisfying assignment $\varphi_{PR}$,
it is identity by Lemma~\ref{lem:cycles-and-assignments} 
and we arrive at a contradiction.
\end{proof}

By Lemma~\ref{lem:dml}, $\minlin{2}{\D}$ reduces to $\DML(\D)$ in fpt time.
Let $I=(S, w_S, X, k)$ be an instance of the latter problem.
Note that $S - X$ is consistent, so all cycles in it are either
identity or non-identity.
To apply graph cleaning, we construct
a \emph{rooted graph for $I$} as follows.

\begin{definition}\label{def:GI}
Let $I=(S, w_S, X, k)$ be an instance of $\DML(\D)$.
The \emph{rooted graph} for $I$ is a biased graph $(G_I,\B_I)$ defined as follows. The vertex set of $G_I$ is the set of
  the variables of $S-X$ extended with a
  fresh {\em root vertex} $s$.
  The edge set contains all edges in the primal graph of $S-X$ (with the corresponding weights given by $w_S$)
  together with an edge of weight $1$ from $s$ to every vertex in $V(X)$.
  Moreover, let $\B_I \subseteq 2^{E(G_I)}$ be the set of identity cycles in $S - X$.
\end{definition}

Observe that the family of cycles $\B_I$ above admits a polynomial-time oracle
e.g. by checking for every pair of vertices whether two paths connecting them
on the cycle imply the same equation.

\begin{lemma} \label{lem:biased}
  $(G_I, \B_I)$ is a biased graph. 
\end{lemma}
\begin{proof}
Consider a cycle in $G_I$ that is outside of $\B_I$.
Such a cycle either contains the root vertex $s$ or is non-identity in $S - X$.
Adding a chordal path to a cycle of the first kind
creates two cycles at least one of which also contains $s$.
For the cycles of the second kind, invoke Lemma~\ref{lem:non-id-has-theta}.
\end{proof}

The following is an immediate algorithmic consequence of
Theorem~\ref{thm:important-subgraphs} and Lemma~\ref{lem:biased}.

\begin{observation} \label{obs:edom-rooted-enum}
  Let $q$ be a positive integer and let $\cG:= \cG(G_I, \B_I, q, s)$
  be the family of connected balanced subgraphs in $(G_I, \B_I)$
  rooted in $s$ with cost at most $q$. Then, in time $\bigoh^*(4^q)$
  we can compute a dominating family $\cH$ for $\cG$
  of size at most $4^q$. 
\end{observation}

Now we characterize yes-instances of $\minlin{2}{\D}$.
To this end, let $Z$ be an optimal solution to $I$ disjoint from $X$,
and $\varphi$ be an assignment satisfying $S - Z$.
The variables in $V(S)$ are partitioned into two sets by $\varphi$:
$V_0 = \{ v \in V(S) \mid \varphi(v) = 0 \}$ and
$V_\emptyset = \{ v \in V(S) \mid \varphi(v) \neq 0 \}$, i.e.
those assigned zero and non-zero values, respectively.

\begin{lemma} \label{lem:zero-nzero-comps}
  Let $K \subseteq V(S)$ be a connected component of $S - (X \cup Z)$.
  \begin{enumerate}[1.]
    \item \label{lem:zero-nzero-comps:zero-or-nonzero}
    Either $K \subseteq V_0$ or $K \subseteq V_\emptyset$.
    \item \label{lem:zero-nzero-comps:nonzero-flexible}
    If $K \subseteq V_\emptyset$, then $(S - (X \cup Z))[K]$ is flexible.
    \item \label{lem:zero-nzero-comps:shadow-zero}
    If $K \cap V(X) = \emptyset$, we may assume without loss of generality that $K \subseteq V_0$.
  \end{enumerate}
\end{lemma}
\begin{proof}
First, note that $S - X$ is homogeneous, and so is 
the subset of equations in $S - (X \cup Z)$ induced by $K$. 
Statement~1 follows by observing that if one variable is assigned zero in a two-variable homogeneous system, 
then every connected variable must be assigned zero as well.
For statement~2, note that if $K$ is rigid, it can only be satisfied by the all-zero assignment.
Finally, for statement~3, if $K \cap V(X) = \emptyset$, then 
$K$ also induces a homogeneous connected component in $S - Z$,
which can be satisfies by the all-zero assignment
independently of all other variables.
\end{proof}

We now introduce the {\em zero-free subgraph} $H_\emptyset$ of the rooted graph $(G_I,\B_I)$.
\begin{definition}\label{def:ZFS}
  Let $I=(S, w_S, X, k)$ be an instance of $\DML(\D)$, $Z$ be an optimal solution of $I$, $\varphi$ be a satisfying assignment of $S - Z$, and $(G_I, \B_I)$ be the rooted graph for $I$.
  Then, the \emph{zero-free subgraph} $H_\emptyset := H_\emptyset(I, Z, \varphi)$ of $G_I$ 
  (with distinguished vertex $s$) is defined as follows.
  Let $V(H_\emptyset) = V_\emptyset \cup \{s\}$.
  Add every edge from $(G_I - Z)[V_\emptyset]$ to $E(H_\emptyset)$.
  Finally, for each zero-free component $K$ of $S-(X\cup Z)$, pick one vertex $x \in K \cap V(X)$ 
  (which exists by Lemma~\ref{lem:zero-nzero-comps})
  and add the edge $\{s,x\}$ to $E(H_\emptyset)$.
\end{definition}

\begin{lemma} \label{lem:balanced}
  $H_\emptyset$ is a connected balanced subgraph of $(G_I, \B_I)$,
  and $c_{G_I}(H_\emptyset) \leq 3k + 1$.
\end{lemma}
\begin{proof}
Note that by construction, $H_\emptyset$ contains edges from $G_I - Z$
and edges connecting $s$ to $V(X)$ which are also present in $G_I$,
hence it is a subgraph of $G_I$.
$H_\emptyset$ is clearly connected through the vertex $s$.
To see that all cycles in $H_\emptyset$ are balanced, 
consider a zero-free component $K$ in $S - (X \cup Z)$.
By Lemma~\ref{lem:zero-nzero-comps}, $(S - (X \cup Z))[K]$ is flexible so
$H_\emptyset[K]$ is a balanced subgraph of $(G_I, \B_I)$ whenever $K$ is zero-free.
Finally, the vertex $s$ has exactly one neighbour in each component $K$ with $V(X)\cap K=\emptyset$,
so $H_\emptyset$ does not contain any new cycle going through $s$.

The cost of $H_\emptyset$ in $G_I$ is 
$c_{G_I}(H_\emptyset) = \abs{Z} + \abs{V(X)} - k_\emptyset$,
where $k_\emptyset$ is the number of zero-free components in $S - (X \cup Z)$.
Since $1 \leq \abs{Z} \leq k$, $\abs{V(X)} = 2k + 2$, and $k_\emptyset \geq 1$,
we have that $c_{G_I}(H_\emptyset) \leq 3k + 1$.
\end{proof}

\subsection{Algorithm for $\mintwolin$ over Euclidean Domains}
\label{ssec:edom-alg-def}

In the end of this section we will present our fpt algorithm for $\minlin{2}{\D}$.
By Lemma~\ref{lem:dml}, it suffices to prove that $\DML(\D)$ is in \FPT.
To this end, let $(S, w_S, X, k)$ be an instance of $\DML(\D)$, 
$Z$ be a minimum solution, and $\varphi_Z$ be a satisfying assignment to $S - Z$.
Further, assume $F \subseteq S - (X \cup Z)$ is a set of equations
such that every rigid component of 
$S' := S - (X \cup F)$ is zero under $\varphi_Z$.
We call vertices in $V(X \cup F)$ \emph{terminals},
and refer to $F$ as a \emph{cleaning set with respect to $\varphi_Z$}.
We will later show how to obtain a cleaning set using Observation~\ref{obs:edom-rooted-enum}.

\begin{figure}[bt]
  \centering
  \begin{tikzpicture}
    \tikzstyle{gn}=[draw, fill=black, circle, line width=1pt, inner sep=1pt]
    \tikzstyle{tn}=[draw, rectangle, line width=1pt, minimum width=8.25cm, minimum height=2cm, inner sep=0pt]
    \tikzstyle{cn}=[draw, rectangle, line width=1pt, lightgray]
    \tikzstyle{sl}=[draw, line width=1pt, dotted]
    \tikzstyle{ee}=[draw, line width=1pt]
    \tikzstyle{xe}=[draw, lightgray, line width=1pt]
    \tikzstyle{ln}=[]

    \draw
    node[tn] (T) {}
    (T.north) +(0cm,0cm) node[anchor=south, ln] (Tlab) {$V(X\cup F)$}
    
    (T.south west) +(1.05cm, 1.9cm) node[anchor=north, cn, minimum width=1.8cm, minimum height=3.5cm] (C1) {}
    (C1.south) +(0cm,0cm) node[anchor=south,ln] (C1lab) {$C_1$}
    (T.north west) +(2.1cm,0.1cm) node[] (P1lu) {}
    (P1lu) +(0cm,-2.25cm) node[] (P1ll) {}
    (P1lu) edge[sl] (P1ll)

    (P1lu) +(.1cm, -0.2cm) node[anchor=north west, cn, minimum width=1.8cm, minimum height=3.5cm] (C2) {}
    (C2.south) +(0cm,0cm) node[anchor=south,ln] (C2lab) {$C_2$}
    (P1lu) +(2.05cm,0cm) node[] (P2lu) {}
    (P2lu) +(0cm,-2.25cm) node[] (P2ll) {}
    (P2lu) edge[sl] (P2ll)

    (P2lu) +(.1cm, -0.2cm) node[anchor=north west, cn, minimum width=1.8cm, minimum height=3.5cm] (C3) {}
    (C3.south) +(0cm,0cm) node[anchor=south,ln] (C3lab) {$C_3$}
    (P2lu) +(2.05cm,0cm) node[] (P3lu) {}
    (P3lu) +(0cm,-2.25cm) node[] (P3ll) {}
    (P3lu) edge[sl] (P3ll)

    (P3lu) +(.1cm, -0.2cm) node[anchor=north west, cn, minimum width=1.8cm, minimum height=3.5cm] (C4) {}
    (C4.south) +(0cm,0cm) node[anchor=south,ln] (C4lab) {$C_4$}

    (P1lu) +(1.025cm,0cm) node[] (P5lu) {}
    (P5lu) +(0cm,-2.25cm) node[] (P5ll) {}
    (P5lu) edge[sl] (P5ll)

    ;
    \draw
    (C1.north) +(-0.6cm, -0.8cm) node[gn] (gn1) {}
    (C1.north) +(0.6cm, -0.8cm) node[gn] (gn2) {}
    (C1.north) +(0cm, -1.5cm) node[gn] (gn3) {}
    ;
    \draw[ee]
    (gn1) -- (gn2)
    (gn2) -- (gn3)
    (gn1) -- (gn3)
    ;
    \draw
    (C2.north) +(-0.45cm, -0.6cm) node[gn] (gn4) {}
    (gn4) +(0cm, -0.6cm) node[gn] (gn5) {}
    ;
    \draw[ee]
    (gn4) -- (gn5)
    ;

    \draw
    (C2.north) +(+0.45cm, -0.9cm) node[gn] (gn6) {}
    ;

    \draw
    (C3.north) +(-0.6cm, -0.8cm) node[gn] (gn7) {}
    (C3.north) +(0.6cm, -0.8cm) node[gn] (gn8) {}

    (C4.north) +(0cm, -0.8cm) node[gn] (gn9) {}
    ;
    \draw
    (T.east) +(1cm, -0.4cm) node[gn, label=above:$z_0$] (z) {}
    ;
    \draw[ee]
    (z) -- (gn7)
    (z) -- (gn8)
    (z) -- (gn9)
    ;

    \draw[xe]
    (gn3) -- (gn4)
    (gn3) -- (gn5)
    (gn2) -- (gn5)
    (gn1) edge[bend left] (gn2)
    (gn6) -- (gn5)
    (gn7) edge[bend left=25] (gn8)
    (gn6) -- (gn7)
    (gn8) edge[bend left=25] (gn9)
    ;

  \end{tikzpicture}
  \caption{An illustration of the auxiliary instance $H_\cP = H(S, X, F, \cP)$. Here, $C_1,\dotsc, C_4$ are
  all components of $S'=S-(X\cup F)$ with $C_3$ and $C_4$ being the
  only rigid components. Black circular vertices represent the
  vertices of the primal graph of $H_\cP$, i.e. the terminals in
  $V(X\cup F)$. Moreover, edges in
  light grey represent (possible) edges in $X\cup F$ and black edges
  represent constraints added to $H_\cP$, more specifically, black edges
  inside $C_1$ and $C_2$ represent the constraints $e_{xy}(S')$ and
  black edges incident with $z$ represent the constraints $x-z_0=0$ and $x+z_0=0$.
  The dotted vertical
  lines within the rectangle for $V(X\cup F)$ give the partition $\cP$ of $V(X\cup F)$, which is a refinement of
  the partition $\cP'$ given by the components $C_1\dotsc,C_4$.}
  \label{fig:hcp}
\end{figure}
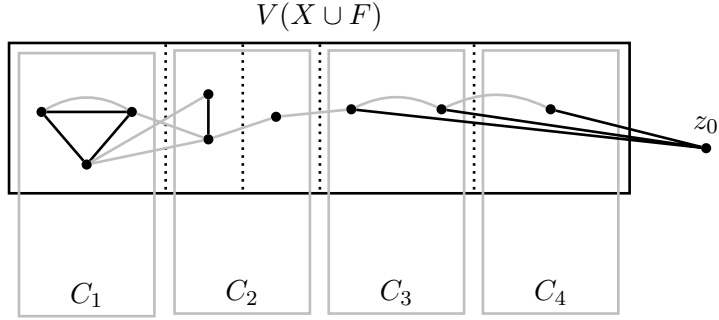

Let $\cP'$ be the partition of terminals into connected components of $S'$
i.e. $\cP'(x) = \cP'(y)$ if and only if $x$ and $y$ are in
the same connected component of $S'$.
For every partition $\cP$ that refines $\cP'$,
we describe the construction of 
an auxiliary instance $H_\cP = H(S, X, F, \cP)$ of $\lin{2}{\D}$
that is used in the algorithm (see Figure~\ref{fig:hcp} for an illustration).
$H_\cP$
contains all variables in $V(X\cup F)$ plus an additional zero
variable $z_0$. Moreover, $H_\cP$ contains all equations in $X\cup
F$ plus the following additional equations:
\begin{itemize}
\item For every terminal $x$ that is in a rigid component of $S'$,
  the equations $x - z_0 = 0$ and $x + z_0 = 0$.
\item For every pair of terminals $x, y$ such that 
  $\cP(x) = \cP(y)$  and $x,y$ appear in a flexible component of $S'$,
  the equation $e_{xy}(S')$.
\end{itemize}
This completes the construction of $H_\cP$.
We distinguish between different kinds of terminals:
terminals appearing in rigid components of $H_\cP$ are called \emph{determined}, 
while those appearing in flexible components are called \emph{undetermined}.
Note that all terminals appearing in the connected component of zero variable $z_0$
are determined since equations connecting $z_0$ and $z'_0$ form a non-identity cycle.
We call them \emph{zero-determined} terminals.
Observe further that not all determined terminals have to be zero-determined
as $H_\cP$ may contain rigid components apart from the one including $z_0$.

If $Z$ is a solution to $(S, w_S, X, k)$ and 
$\cP_Z$ is the partition of terminals into 
connected components of $S' - Z$,
then, intuitively, $H_{\cP_Z}$ serves as the ``projection'' of $S - Z$ 
onto the terminals i.e. it encapsulates
all constraints in $S - Z$ between the pairs of terminals.
This intuition is formalized below.

\begin{lemma} \label{lem:HP-facts}
  Let $(S, w_S,X, k)$ be an instance of $\DML(\D)$,
  $Z$ be a solution, and $\varphi_Z$ be a satisfying assignment to $S - Z$.
  Let $F \subseteq S - (X \cup Z)$ be a cleaning set with respect to $\varphi_Z$, and
  ${\cP_Z}$ be the partition of $V(X \cup F)$ into 
  connected components of $S' - Z$, where $S': = S - (X \cup F)$.
  Then the following statements hold:
  \begin{enumerate}[1.]
  \item \label{lem:HP-facts:consistent} 
    $H_{\cP_Z}$ is consistent.
  \item \label{lem:HP-facts:determined}
    If a terminal $x \in V(X \cup F)$ is determined, then $\varphi(x)=\varphi_Z(x)$ for every satisfying
    assignment $\varphi$ of $H_{\cP_Z}$.
  \end{enumerate}
\end{lemma}
\begin{proof}
  {\em Statement~\ref{lem:HP-facts:consistent}.} 
  We show that the assignment $\varphi$ obtained from $\varphi_Z$
  after setting $\varphi(z_0)=0$ for the zero variable $z_0$ satisfies
  $H_{\cP_Z}$. To this end, let $e$ be an equation of $H_{\cP_Z}$. If $e
  \in X\cup F$, then $e \in S-Z$ and $\varphi$ satisfies $e$. If $e$
  contains $z_0$, then $e$ is equal to $x-z_0=0$ or $x+z_0=0$, where $x$
  is contained in a rigid component of $S'$. Because $F$ is a cleaning set
  with respect to $\varphi_Z$, it holds that $\varphi(x)=\varphi_Z(x)=0$ and
  therefore $e$ is satisfied by $\varphi$. Otherwise, $e$ is equal to
  $e_{xy}(S')$ for some terminals $x$ and $y$ with $\cP_Z(x)=\cP_Z(y)$ that appear together in
  some flexible component $K$ of $S'$. 
  Because $\cP_Z(x)=\cP_Z(y)$ and
  $K$ is flexible, it holds that $e_{xy}(S'-Z)$ is equivalent to
  $e_{xy}(S')$ and $e$, and therefore $\varphi$ satisfies $e$.
  
  {\em Statement~\ref{lem:HP-facts:determined}.} 
  If $x$ is a zero-determined terminal, then $\varphi(x) = 0$ for every
  satisfying assignment $\varphi$ to $H_{\cP_Z}$.
  Moreover, $\varphi_Z(x) = 0$ since $x$ is in a rigid component of $S'$
  and $F$ is a cleaning set.
  On the other hand, if $x$ is not zero-determined, 
  then by construction of $H_{\cP_Z}$, 
  $x$ is contained in an equivalent non-identity cycle in $S - Z$,
  so $\varphi$ and $\varphi_Z$ agree on all terminals in these cycles.
  Therefore, in both cases we have $\varphi(x)=\varphi_Z(x)$.
\end{proof}

Lemma~\ref{lem:HP-facts}.\ref{lem:HP-facts:consistent} 
suggests that the algorithm
for $\DML(\D)$ can start by guessing 
the partition $\cP$ of the terminals
and checking whether $H_\cP$ is consistent.
If yes, then a $\cP$-cut $Y$ in $S'$ of size $k$ can be computed in fpt time
(or we can correctly report that no such cut exists).
However, $S - Y$ is not necessarily consistent.
The reason is that some paths of equations in $S - Y$ 
may be inconsistent.
Thus, the cut needs to fulfil an additional set of requirements
to ensure that it is a solution.
The key insight for computing these requirements is that all paths
avoiding $X$ are homogeneous (hence they imply consistent equations
satisfied by setting all variables to zero), so 
it is sufficient to take care of the paths
containing a variable from $V(X)$.
Then there are two kinds of inconsistent paths:
those confined to a component connecting a terminal and a non-terminal
and those connecting two non-terminals in different components
using at least one equation from $X$.
We show that these requirements can be handled using \PPC{}. For this
we will construct the set $\cF_\cP = \cF(S, X, F, \cP)$ of pair cut
requests one needs to fulfil as follows.
Let $\varphi_H$ be a satisfying assignment to $H_\cP$. Then, the set
$\cF_\cP = \cF(S, X, F, \cP)$ of pair cut request contains the
following pairs. For every determined terminal $x$ that is in a flexible component $K$ of $S'$, 
consider every non-terminal $v$ in $K$ and compute $e_{xv}(S')$.
Plug in $\varphi_H(x)$ for $x$ into the equation $e_{xv}(S')$. 
If there is no value for $v$
that satisfies the equation,
then add $(\{x,v\}, \{x,v\})$ to $\cF_\cP$. 

Now, for every flexible component $K$ of $H_\cP$,
consider every pair of terminals $x,y \in K$
such that $\cP(x) \neq \cP(y)$.
Note that $x$ and $y$ are undetermined.
Let $K'_1$ and $K'_2$ be the (not necessarily distinct) components of $S'$
that contain $x$ and $y$, respectively.
Note that $S'[K'_1]$ and $S'[K'_2]$ are flexible
(otherwise, by construction of $H_\cP$,
variables $x$ and $y$ would form non-identity cycles with $z_0$).
For every pair of non-terminals $u \in K'_1$ and $v \in K'_2$,
compute $e_{ux}(S'[K'_1])$, $e_{xy}(H_\cP[K])$, $e_{yv}(S'[K'_2])$,
and let $e_{uv}$ be the equation implied by composing them
(i.e. treating them as parts of a path, and computing the implied equation).
If $e_{uv}$ has no solution, 
then add $(\{u,x\}, \{y,v\})$ to $\cF_\cP$.
This concludes the definition of $\cF_\cP$.

The algorithm for $\DML(\D)$ can now be summarized as follows. Let $I=(S,
w_S,X, k)$ be an instance of $\DML(\D)$.
\begin{enumerate}
\item Construct the rooted graph $(G_I,\B_I)$ for $I$ as described in
  Definition~\ref{def:GI}. Assume that $s$ is the root of
  $(G_I,\B_I)$.
\item\label{alg:compcH} Let $\cG:= \cG(G_I, \B_I, k, s)$
  be the family of connected balanced subgraphs in $(G_I, \B_I)$
  rooted in $s$ with cost at most $3k+1$. Compute a dominating family
  $\cH$ for $\cG$ using Observation~\ref{obs:edom-rooted-enum}.
\item\label{alg:setF}
  For every $H \in \cH$, let $F_H$ be the set of deleted edges excluding those incident to
  $s$. Guess the intersection $F_Z$ with a
  solution, i.e. for every $F_Z\subseteq F_H$ with $w_S(F_Z)\leq k$, 
  do the following. Let $I'=(S',w_{S},X, k')$ be the
  instance obtained from $I$ by removing all edges in $F_Z$ from $S$
  and
  decreasing $k$ by the weight of $F_Z$. Let
  $F=F_H\setminus F_Z$, $T=V(X\cup F)$, $S''=S'-(X\cup F)$, and 
  $\cP'$ be the partition of $T$ in $S''$. Then, for every partition
  $\cP$ that refines $\cP'$, proceed as follows:
  \begin{enumerate}
  \item Construct the auxiliary instance $H_\cP = H(S', X, F, \cP)$
    of $\lin{2}{\D}$ as described above.
  \item Use Lemma~\ref{lem:edom-algo} to decide whether $H_\cP$ is
    consistent and if so to compute a satisfying assignment
    $\varphi_H$ for $H_\cP$. If $H_\cP$ is inconsistent, then disregard the
    current partition $\cP$ and continue with the next partition. 
  \item Use $\varphi_H$ to construct the set $\cF_\cP = \cF(S', X, F,
    \cP)$ of pair cut requests as described above. Let $I_\cP$ be
    the instance $(S'', w_S, T, \cP, \cF_\cP, k)$ of \PPC.
  \item\label{alg:EDout} Use Theorem~\ref{thm:pair-mc-is-fpt} to solve $I_\cP$. If
    $I_\cP$ has a solution $Y$, then use
    Lemma~\ref{lem:edom-algo} to check whether $S'-Y$ is consistent. If
    so, output $Y\cup F_Z$ as the
    solution for $\DML(\D)$, otherwise disregard the current partition
    $\cP$ and continue with the next partition.
  \end{enumerate}
\item If no solution was output at Step~\ref{alg:EDout}, then reject.
\end{enumerate}

\subsection{Correctness Proof and Complexity Analysis}
\label{ssec:edom-correctness-time}

We will now prove that the algorithm presented in Section~\ref{ssec:edom-alg-def} is correct and
we will analyze its time complexity.
The correctness proof is based on an auxiliary result (Lemma~\ref{lem:solution})
that show the connection between the cleaned
$\DML(\D)$ instance and the \PPC{} instances that are computed in step 3 of the algorithm.
The proof of Lemma~\ref{lem:solution} is simplified
with the aid of the following lemma.

\begin{lemma}\label{lem:z-is-solution}
  Let $(S, w_S, X, k)$ be an instance of $\DML(\D)$ with solution $Z$
  and let $\varphi_Z$ be a satisfying assignment of $S - Z$.
  Let $F \subseteq S - (X \cup Z)$ be a cleaning set with respect to $\varphi_Z$, and
  ${\cP_Z}$ be the partition of $V(X \cup F)$ into 
  connected components of $S' - Z$, where $S'=S-(X\cup F)$.
  Then $Z$ is a ${\cP_Z}$-cut in $S'$ that fulfills $\cF_{\cP_Z}$.
\end{lemma}
\begin{proof}
  Clearly, $Z$ is a ${\cP_Z}$-cut. Suppose now for a contradiction
  that $Z$ does not fulfil $\cF_{\cP_Z}$. First consider the case
  that $Z$ does not fulfil a cut request $(\{x,v\}, \{x,v\})$ in $\cF_{\cP_Z}$,
  where $x$ is a determined terminal. Because of
  Lemma~\ref{lem:HP-facts}.\ref{lem:HP-facts:consistent}, we know
  that $H_{\cP_Z}$ is consistent. 
  Let $\varphi_H$ be a satisfying assignment to $H_{\cP_Z}$ and let 
  $K$ contain the connected component
  of $S'$ that contains $x$ and $v$.
  By Lemma~\ref{lem:HP-facts}.\ref{lem:HP-facts:determined},
  $\varphi_H(x) = \varphi_Z(x)$.
  Since $Z$ does not separate $x$ and $v$ in $S'$,
  at least one path implying the equation $e_{xv}(S'[K])$ persists in
  $S - Z$. However, due to the construction of $\cF_{\cP_Z}$ this
  implies that $\varphi_Z$ does not satisfy $e_{xv}(S'[K])$ and this
  contradicts our assumption that $S-Z$ is consistent.

  Now consider the only remaining case that $Z$ does not fulfil
  a cut request $(\{u,x\}, \{y,v\})$ in $\cF_{\cP_Z}$,
  where $x$ and $y$ are undetermined terminals.
  Let $K'_1$ and $K'_2$ be the connected components of $S'$
  such that $\{u,x\} \subseteq K'_1$ and $\{y,v\} \subseteq K'_2$.
  Further, let $K$ be the connected component of 
  $H_{\cP_Z}$ that contains $x$ and $y$.
  Since $Z$ does not disconnect $u,x$ or $y,v$ in $S'$, 
  a path implying $e_{ux}(S'[K'_1])$ and 
  a path implying $e_{yv}(S'[K'_2])$ persist 
  in $S - Z$.
  Moreover, by the construction of $H_{\cP_Z}$,
  a path implying $e_{xy}(H_{\cP_Z}[K])$ 
  exists in $S - Z$.
  Finally, the construction of $\cF_{\cP_Z}$ ensures that the
  composition of these equations
  does not have a solution in $\D$.
  We conclude that $S - Z$ is inconsistent and this leads to a contradiction.
\end{proof}

\begin{lemma}\label{lem:solution}
  Let $I=(S, w_S, X, k)$ be an instance of $\DML(\D)$ with solution $Z$
  and let $\varphi_Z$ be a satisfying assignment of $S - Z$.
  Let $F \subseteq S - (X \cup Z)$ be a cleaning set with respect to $\varphi_Z$, and let
  ${\cP_Z}$ be the partition of $V(X \cup F)$ into 
  connected components of $S'-Z$, where $S'=S-(X\cup F)$.
  Then every minimum ${\cP_Z}$-cut $Y$ in $S'$ that fulfills $\cF_{\cP_Z}$
  is a solution to $I$.
\end{lemma}
\begin{proof}
  We know that $H_{\cP_Z}$ is consistent by Lemma~\ref{lem:HP-facts}.\ref{lem:HP-facts:consistent} and we let
  $\varphi_H$ denote a satisfying assignment.
  We construct an assignment $\varphi_Y$ based on $\varphi_H$ 
  and prove that $\varphi_Y$ satisfies 
  ${S'' := S'-Y}$,
  considering one connected component of $S''$ at a time.
  Then we show that $\varphi_Y$ also satisfies $X \cup F$, and conclude that it satisfies $S - Y$.

  First note that every connected component of $S''$ 
  is a subset of a component of $S'$.
  If a variable $v$ appears in a rigid component of 
  $S'$, then let $\varphi_Y(v) = 0$.
  If $v$ appears in a component that does not 
  contain any terminal, then let $\varphi_Y(v) = 0$.
  Note that all equations in $S - X$ are homogeneous so $\varphi_Y$ satisfies all equations inside the components of $S''$ considered so far.
  
  Now consider a flexible component $K$ of $S''$ 
  that contains a determined terminal $x$.
  Set $\varphi_Y(x) = \varphi_H(x)$.
  Since $Y$ fulfills $\cF_{\cP_Z}$,
  for every $v \in K$ the equation
  $e_{xv}(S'')$ has a solution where 
  $x \mapsto \varphi_H(x)$. Therefore, we can extend $\varphi_Y$, by assigning every variable $v
  \in K$ with $v \neq x$ to the unique value satisfying $e_{xv}(S'')$
  if $x$ is set to $\varphi_Y(x)$.
  It follows that $\varphi_Y$
  satisfies $\STS(K,x)$, which due to
  Lemma~\ref{lem:edom-flexible} implies that $\varphi_Y$ satisfies $S''[K]$.

  All remaining variables appear in flexible components of $S''$
  that only contain undetermined terminals.
  Let $U$ be the set of all vertices appearing in these components.
  \begin{claim}
    $(S-Y)[U]$ is flexible and consistent.
  \end{claim}
  \begin{claimproof}
    Towards showing that $(S-Y)[U]$ is flexible, first note that
    $S''[U]$ is flexible. Moreover, $H_{\cP_Z}[U \cap V(X \cup F)]$ is also flexible,
    since all terminals in $U$ are undetermined.
    Thus, $(S-Y)[U]$ does not contain any non-identity cycle avoiding $X \cup F$.
    Furthermore, if there were a non-identity cycle in $(S-Y)[U]$ intersecting $X \cup F$,
    then by construction there would also be such a cycle in $H_{\cP_Z}[U \cap V(X \cup F)]$,
    which would be a contradiction.
    Hence, $(S-Y)[U]$ cannot contain a non-identity cycle and it is indeed flexible. 

    We now show that $(S-Y)[U]$ is consistent. Because $(S-Y)[U]$
    is flexible, we obtain from
    Lemma~\ref{lem:edom-flexible} that it suffices to show
    that $(S-Y)[U]$ contains no inconsistent path. Suppose for a contradiction that $(S-Y)[U]$ contains an
    inconsistent path $P$ between say $u$ and
    $v$. We know that $S-X$ is consistent so we can assume that $P$
    intersects $X$. Let $x \in V(X \cup F)$ and $y \in V(X\cup F)$ be the closest terminals to 
    $u$ and $v$ on $P$, respectively. Then, the equation $e_{ux}(P)$ is equivalent
    to $e_{ux}(S')$ and similarly the equation $e_{yv}(P)$ is
    equivalent to $e_{yv}(S')$. Moreover, an equation
    equivalent to the equation $e_{xy}(P)$ is implied by the
    $\{x,y\}$-paths in $H_{\cP_Z}$ due to the construction of
    $H_{\cP_Z}$. Therefore, if $P$ is inconsistent, then so is the
    equation obtained by combining $e_{ux}(S')$, $e_{xy}(P)$, and
    $e_{yv}(S')$, which implies that $(\{x,u\}, \{y,v\})$ is a pair
    cut request in $\cF_{\cP_Z}$. But this contradicts our assumption
    that $P$ is in $(S-Y)[U]$ because $Y$ fulfills $\cF_{\cP_Z}$ and
    therefore intersects $P$.
  \end{claimproof}
  
  Using the claim above, we can now extend $\varphi_Y$ to $U$ using
  any satisfying assignment $\varphi_U$ of $(S-Y)[U]$ by setting
  $\varphi_Y(u) = \varphi_U(u)$ for all $u \in U$. We show that
  $\varphi_Y$ obtained in this manner satisfies not only $X \cup F$
  but also $H_{\cP_Z}$.
  \begin{claim}
    The assignment $\varphi_Y$ satisfies $H_{\cP_Z}$.
  \end{claim}
  \begin{claimproof}
    Let $K$ be a connected component of $H_{\cP_Z}$.
    If $z_0 \in K$, then $\varphi_H(v) = 0$ for all $v \in K$.
    By the construction of $H_{\cP_Z}$, $K \setminus \{z_0\}$
    is a subset of a rigid component of $S'$ so 
    $\varphi_Y(v) = 0$ for all $v \in K$.
    If $z_0 \notin K$ and $K$ is rigid, then $K$ only contains
    determined terminals and it follows that $\varphi_Y$ 
    agrees with $\varphi_H$ on $K$ by construction.
    Finally, if $K$ is flexible, then consider arbitrary $x, y \in K$.
    By the construction of $H_{\cP_Z}$, there is a path in $(S - Y)[U]$
    that implies $e_{xy}(H_{\cP_Z}[K])$.
    Hence, $\varphi_Y$ satisfies $e_{xy}(H_{\cP_Z}[K])$ for all $x,y \in K$.
    We have thus exhausted all cases and the claim holds.
  \end{claimproof}
  
  We have shown that $\varphi_Y$ satisfies both
  $S'' = S'-Y$ and $X \cup F \subseteq H_{\cP_Z}$. Therefore, $S - Y$
  is consistent and it only remains to show that $\abs{Y}\leq k$. Lemma~\ref{lem:z-is-solution} implies that $Z$ is a ${\cP_Z}$-cut in $S'$ that fulfills
  $\cF_{\cP_Z}$. Moreover, $Z$ is a solution of $I$ so $|Z|\leq k$. 
  It follows that if $Y$ is a minimum such ${\cP_Z}$-cut,
  then $\abs{Y}\leq k$ and $Y$ is a solution of $I$.
\end{proof}

We are now ready to prove correctness
and to provide the time complexity analysis of the algorithm.
For the analysis of
the run-time, we will use $Q(k)$ to denote the run-time dependency on
the parameter $k$ for the algorithm for \PPC{}, i.e. $\bigoh^*(Q(k))$ is the
run-time for the algorithm for \PPC{} given in
Theorem~\ref{thm:pair-mc-is-fpt}. This makes it clear that the main
bottleneck for our algorithm is the underlying algorithm for \PPC{}.
\begin{theorem}
  $\minlin{2}{\D}$ is in \FPT and can be solved in time
  $\bigoh^*(k^{\bigoh(k)}Q(k))$. 
\end{theorem}
\begin{proof}
We start by analyzing the algorithm for $\DML(\D)$ presented in Section~\ref{ssec:edom-alg-def}.
Let $I=(S, w_s, X, k)$ be an arbitrary instance of $\DML(\D)$. We 
show that the algorithm accepts if and only if $I$ is a yes-instance. The forward direction is
simple because if the algorithm returns a solution $Y$, then $\abs{Y}\leq
k$ and $S-Y$ is consistent because of Step~\ref{alg:EDout} of the
algorithm. 

Towards showing the reverse direction, suppose that $I$ is a
yes-instance having a solution $Z$. Let $\varphi_Z$ be a satisfying
assignment of $S - Z$ and define $V_0$ and $V_\emptyset$ accordingly.
Let $H_\emptyset$ denote the zero-free subgraph of $G_I$ as given in Definition~\ref{def:ZFS}.
By Lemma~\ref{lem:balanced}, $H_\emptyset$ is 
balanced, connected, and $c_{G_I}(H_\emptyset) \leq 3k + 1$.
Because the family $\cH$ that is computed in Step~\ref{alg:compcH} of the algorithm is a
dominating family for $\cG$, there is an
(important) balanced subgraph $H\in \cH$ that dominates
$H_\emptyset$. Moreover, because $H\in \cH$, $H$ is considered by the
algorithm in Step~\ref{alg:setF}.

Let $F_H$ be the corresponding set of deleted edges in $G_I - \{s\}$,
let $F_Z=F_H\cap Z$, and let $F=F_H\setminus F_Z$. Let $I'=(S',
w_{S},X, k')$ be the instance obtained from $I$ by removing all edges in $F_Z$ from $S$ and
decreasing $k$ by the weight of $F_Z$.
Note that $I$ has a
solution if and only if  $I'$ has a solution. Moreover, note that $F$ is considered by the algorithm because the algorithm
considers all subsets of $F_Z$ of $F_H$ of weight at most $k$ in
Step~\ref{alg:setF}.
Let $T=V(X\cup F)$, $S''=S'-(X\cup F)$ and let $\cP'$ be the partition
of $T$ in $S''$. Let $\cP_Z$ be the partition of $T$ in
$S''-Z$. Then, because the algorithm considers all refinements of
$\cP'$, it also considers the partition $\cP_Z$.
Finally, note that
$F$ is a cleaning set with respect to $\varphi_Z$ in $S'$. This is because $V_\emptyset
\subseteq V(H_\emptyset) \subseteq V(H)$
and all components in $S''[V(G_I) \setminus \{s\}]$ are flexible.
Hence, all variables in the rigid components of $S''$
are assigned zero values by $\varphi_Z$.
Therefore, $Z\setminus F_Z$, $\varphi_Z$, $\cP_Z$, and $F$ satisfy all conditions
of Lemma~\ref{lem:HP-facts}.\ref{lem:HP-facts:consistent} for the
instance $I'$, which
implies that $H_{\cP_Z}$ is consistent. Moreover, $Z\setminus F_Z$, $\varphi_Z$,
$\cP_Z$, and $F$ also satisfy all conditions of
Lemma~\ref{lem:solution} on the instance $I'$ and therefore every ${\cP_Z}$-cut $Y$ in $S''$
that fulfills $\cF_{\cP_Z}$ is a solution for $I'$. Therefore, the set
$Y\cup F_Z$ returned by the algorithm in Step~\ref{alg:EDout}
is a solution for $I$.

We continue by analyzing the run-time of the algorithm. The algorithm
starts by computing a dominating family $\cH$ of
$\cG:=\cG(G_I,\B_I,k,s)$ of size at most $4^{3k+1}$ in time $\bigoh^*(4^{3k+1})$ using
Observation~\ref{obs:edom-rooted-enum}. Let $H \in \cH$ and
let $F_H$ be the set of
deleted edges for $H$ excluding those incident with $s$. Then,
for every $F_H$, the algorithm considers at most $2^{|F_H|}\leq
2^{3k+1}$ (because $c_{G_I}(H)\leq 3k+1$) subsets $F_Z$ and computes the updated instance $I'=(S',w_S,X,k')$ in
Step~\ref{alg:setF} in polynomial-time. Let $F=F_H\setminus F_Z$, $S''=S'-(X\cup F)$,
$T=V(X\cup F)$, and let $\cP'$ be the partition of $T=V(X \cup F)$ in
$S''$. The algorithm then enumerates all refinements $\cP$ of
$\cP'$. Because the number of such refinements $\cP$ is at most
$|T|^{|T|}\leq (4k)^{4k}$, this can be achieved in time
$\bigoh((4k)^{4k})$. For each $\cP$, the algorithm then constructs 
$H_\cP=H(S', X, F, \cP)$ in polynomial-time and decides whether $H_\cP$ is consistent
in polynomial time using Lemma~\ref{lem:edom-algo}. If $H_\cP$ is not
consistent, the algorithm stops, otherwise it constructs the set of
pair-cut requests $\cF_\cP=\cF(S', X, F, \cP)$ and the instance
$I_\cP$ of \PPC{} in polynomial-time. Finally, the algorithm solves $I_\cP=(S'', w_S, T, \cP, \cF, k)$
using Theorem~\ref{thm:pair-mc-is-fpt} in fpt-time with respect to
$k'\leq k$, i.e. in
time $\bigoh^*(Q(k))$. Therefore, the total time
required by the algorithm is at most
\[\bigoh^*(4^{3k+1}2^{3k+1}(4k)^{4k}Q(k))=\bigoh^*(k^{\bigoh(k)}Q(k))\]
which
shows that $\DML(\D)$ is fpt with respect to $k$.
By Lemma~\ref{lem:dml}, there is another factor of $2^k$ in the running time
of the algorithm for $\minlin{2}{\D}$, which is dominated by $k^{\bigoh(k)}$,
so asymptotically we obtain the same running time for $\minlin{2}{\D}$. 
\end{proof}

\section{Faster Algorithm for Fields}
\label{sec:fields-algorithm}

Let $\F$ be an effective field.
In this section we present improved 
fpt algorithms for $\minlin{2}{\F}$---an $\bigoh^*(k^{\bigoh(k)})$ time algorithm for arbitrary fields and
an $\bigoh^*((2p)^{k})$ time algorithm for finite $p$-element fields. 
The improvements use the nicer structural properties of fields, 
mainly the fact that every nonzero element has a multiplicative inverse.
Section~\ref{ssec:lin2f} demonstrates how $\minlin{2}{\F}$ differs
from the more general $\mintwolin$ over Euclidean domains and
several useful observations are derived from this.
In Section~\ref{ssec:fields-algo-def}, we present the algorithm for fields and continue in Section~\ref{ssec:correctness-arbfield} by
proving its correctness and analysing its running time.
Finally, we present
a faster algorithm for $\mintwolin$ over finite fields
in Section~\ref{ssec:finfields-algorithm}.

\subsection{$\twolin$ over Fields}
\label{ssec:lin2f}

Since the quotient of any two nonzero elements is an element of the field $\F$, 
instances of $\lin{2}{\F}$ enjoy rather pleasant properties 
that do not necessarily hold in arbitrary Euclidean domains.
First, note that any single equation $ax + by = c$ over $\F$
is consistent unless $a = b = 0$ and $c \neq 0$.
By preprocessing, we may assume that
such equations do not occur in our instances.
Hence, we may assume that all paths in our instances are consistent.
This implies the following via 
Lemmas~\ref{lem:acyclic=star} and~\ref{lem:edom-acyclic-consistent}.

\begin{corollary} \label{cor:flexible-propagate}
  Every flexible instance $S$ of $\lin{2}{\F}$ is consistent.
  Moreover, for any variable $z \in V(S)$ and any element $d$ in $\F$,
  there is an assignment that satisfies $S$ and sets $z$ to $d$. 
\end{corollary}

Flexible instances have another useful property.
We call a variable substitution $\Phi$ \emph{equalising}
if every equation in $\Phi(S)$ is equality, i.e.
it has the form $x=y$.

\begin{lemma} \label{lem:equalise}
  Every flexible instance of $\lin{2}{\F}$ admits 
  an equalising variable substitution.
\end{lemma}
\begin{proof}
Let $S$ be a connected flexible instance of $\lin{2}{\F}$.
The instance $S$ is consistent by Corollary~\ref{cor:flexible-propagate} so 
Lemma~\ref{lem:homogenise} allows us to 
assume that $S$ is homogeneous.
We may additionally assume (by division of field elements) that
every equation is of 
the form $x = ay$ for some $a \in \F$.
Pick an arbitrary variable $z \in V(S)$ and 
construct a spanning tree $T \subseteq S$ rooted in $z$.
Define a variable substitution $\Phi$ by
$x \mapsto a_x x'$, where $x = a_x z$ is the equation $e_{xz}(S)$.
Note that this map is reversible since division is available in $\F$.
Clearly, $\Phi(S)$ is homogeneous.
Moreover, equation $e_{x',z'}(\Phi(S))$ is $a_x x' = a_x z'$
which simplifies $x' = z'$.
We conclude that every equation in $\Phi(S)$ is equality.
\end{proof}

Yet another consequence of division in $\F$ 
is the following lemma that allows us to remove a factor of $2^{\bigoh(k)}$ from 
the time complexity of iterative compression.

\begin{lemma} \label{lem:field-dml}
  If $\DML(\F)$ is solvable in $\bigoh^*(f(k))$ time, 
  then $\minlin{2}{\F}$ is solvable in $\bigoh^*(f(k))$ time.
\end{lemma}
\begin{proof}
Given an instance $(S, w_S, k)$ of $\minlin{2}{\D}$,
apply \emph{(equation) subdivision} to it:
for every equation $e$ of the form $ax + by = c$ in the instance,
introduce a new variable $z_e$ and replace the original equation by
a \emph{subdivided pair} of equations 
$P_e = \{ x = z_e, az_e + by = c \}$.
Both equations in the pair
are assigned the same weight as the original one.

Clearly, any minimal solution only needs to contain one equation
from each subdivided pair.
Hence, the resulting instance has a solution of weight $k$
if and only if the original instance has one.
Moreover, when applying iterative compression to $(S, w_S, k)$
and having a suboptimal but minimal solution $X$ at hand,
we may safely assume that the optimal solution $Z$ to the instance
is disjoint from $X$ (e.g. if $X$ and $Z$ 
need to separate the same pair of original variables, they 
may choose different equations from the subdivided pair).
Hence, there is no need to branch on the intersection of $X$ and $Z$
and the instance can be solved directly by passing it to the $\DML(\F)$ algorithm.
\end{proof}

\subsection{Algorithm for $\mintwolin$ over Fields}
\label{ssec:fields-algo-def}

Let $I = (S, w_s, X, k)$ be an instance of $\DML(\F)$.
By Lemma~\ref{lem:dml}, it suffices to construct an fpt algorithm for the latter problem.
The opening of the algorithm is equation subdivision which allows for speeding up iterative compression by Lemma~\ref{lem:field-dml}.
In fact, we apply subdivision twice to replace every equation with three new ones, i.e. two fresh variables $z_1$ and $z_2$ are introduced and
$ax + by = c$ is replaced by $\{x = z_1, z_1 = z_2, az_2 + by = c \}$. 
This allows us to avoid several branching steps---more details are given after the algorithm description.
Then we construct the rooted graph $(G_I,\B_I)$ and compute 
a dominating family of important balanced subgraphs to obtain a cleaning set $F$.
In contrast to the algorithm for Euclidean domains, 
the following steps are simplified by the additional structure of fields.
In the iterative compression step, it
is ensured that the solution
is disjoint from $X \cup F$ simply by using subdivision as in Lemma~\ref{lem:field-dml}.
The cutting step is simplified even more dramatically:
it turns out that guessing the correct partition of the terminals $\cP$
and computing a minimum $\cP$-cut is sufficient since there are no
inconsistent paths in the instances of $\lin{2}{\F}$.
We claim that the following algorithm solves the
instance $I=(S, w_S, k, X)$ of $\DML(\F)$ in 
$O^*(2^{O(k \log k)})$ time. 

\begin{enumerate}

\item Apply equation subdivision (like in Lemma~\ref{lem:field-dml}) 
twice to $(S, w_S, k)$ so that every equation is divided into three equations.

\item Construct the rooted graph $(G_I,\B_I)$  for $I$ as described in Definition~\ref{def:GI}. Assume that $s$ is the root.

\item Let $\cG:= \cG(G_I, \B_I, k, s)$
  be the family of connected balanced subgraphs in $(G_I, \B_I)$
  rooted in $s$ with cost at most $3k+1$. Compute a dominating family
  $\cH$ for $\cG$ using Observation~\ref{obs:edom-rooted-enum}.
\item \label{alg:Fout}
  For every $H \in \cH$, let $F_H$ be the set of deleted edges excluding those incident to
  $s$.
  For each partition $\cP$ of $V(X \cup F_H)$,
  check if there is a $\cP$-cut $Y$ in $S - (X \cup F_H)$ of size at most $k$
  using Lemma~\ref{lem:p-cut-fpt}.
  If $Y$ exists and $S-Y$ is consistent, then output $Y$. Otherwise
  continue with the next partition.
\item If no solution was output in the previous step, then reject $I$.
\end{enumerate}

The double subdivision in step 1 allows us to assume that
an optimal solution $Z$, the set $X$, and 
the current cleaning set $F_H$ are pairwise disjoint.
We can thus avoid branching on their intersections
(analogously to how the iterative compression algorithm for fields
presented in Lemma~\ref{lem:field-dml} 
avoids the branching step in 
the general compression algorithm from Lemma~\ref{lem:dml}).

\subsection{Correctness Proof and Complexity Analysis}
\label{ssec:correctness-arbfield}

We start with a lemma that will help us prove correctness of the algorithm.
This lemma can be viewed as an analogue
of Lemma~\ref{lem:solution} but its proof
is noticeably different.

\begin{lemma} \label{lem:support}
  Let $I=(S, w_S, X, k)$ be an instance of $\DML(\F)$ with solution $Z$
  and let $\varphi_Z$ be a satisfying assignment of $S - Z$.
  Let $F \subseteq S - (X \cup Z)$ be a cleaning set with respect to $\varphi_Z$, and let
  $\cP$ be the partition of $V(X \cup F)$ into 
  connected components of $S'-Z$, where $S'=S-(X\cup F)$.
  Then every minimum $\cP$-cut in $S'$ is a minimum solution for $(S, w_S, X, k)$. 
\end{lemma}
\begin{proof}
Let $K$ be a component of $S'$ which does not contain any
non-identity cycles. Then $K$ is flexible, and by
Lemma~\ref{lem:equalise}, we can perform a substitution $\Phi(S)$ on
$S$ such that $K$ becomes equalised (i.e. all equations of
$\Phi(S)[K]$ except those in $X \cup F$ are equalities).  Perform
this substitution for all flexible components $K$ of $S'$, and 
let $\varphi_F$ be the updated satisfying assignment to $S-Z$.
Observe that $\varphi^{-1}(0)=\varphi_F^{-1}(0)$ since $S-X$ is
homogeneous. As before, we refer to the vertices of $V(X \cup F)$ as terminals.
Consider a set $B \in \cP$.
Lemma~\ref{lem:zero-nzero-comps} implies that if one variable in $B$
is assigned the zero value, then all variables in $B$ are assigned the
zero value by $\varphi_F$.
On the other hand, if no variable in $B$ is assigned zero value,
then, by variable substitution, 
all paths connecting variables in $B$ imply equalities 
between all variables in $B$.
Hence, $\varphi_F$ is constant on every $B \in \cP$ and every
connected component of $S' - Z$.

Now, let $Y$ be a minimum $\cP$-cut. 
We show that $Y$ is a solution by constructing
an assignment $\varphi_Y$ that satisfies $S - Y$.
Let $\varphi_Y(v) = \varphi_F(v)$ for any terminal
$v \in V(X \cup F)$, propagate values so that every connected
component of $S'-Y$ takes the same value on every vertex, and set
$\varphi_Y(v)=0$ for any vertex $v$ in a connected component of
$S'-Y$ without terminals. We note that $\varphi_Y$ is well-defined.
Indeed, if $u$ and $v$ are terminals such that
$\varphi_F(u) \neq \varphi_F(v)$, then
$u$ and $v$ are in different parts of $\cP$. Since $Y$ is a
$\cP$-cut, no component of $S'-Y$ contains both $u$ and $v$.

Consider an arbitrary equation $e \in S-Y$.
If $e \in X \cup F$, then $\varphi_Y$ matches
$\varphi_F$ on $e$. Since $e \notin Z$ by assumption, this implies
that $\varphi_Y$ satisfies $e$. Next, assume that $e$ is in a
flexible connected component $K$ of $S'$. We know that
$e \notin X \cup F$ and $e$ is equality so by
construction both variables of $e$ take the same value in
$\varphi_Y$. Finally, assume that $e$ appears in a rigid component $K$ of $S'$.
By assumption, $\varphi_F$ assigns zero to $K$.
Assume first that there exists a path $P$ in $S'-Y$ connecting $e$ to a
terminal $v$. Then necessarily $P$ is contained in $K$ and $\varphi_Y(v)=\varphi(v)=0$. If no such path exists, then the
variables in $e$ take the value zero by default. In both cases, 
the variables in $e$ are assigned zero and $e$ is satisfied by $\varphi_Y$.
This exhausts the cases and shows that $Y$ is a solution.

Since $Y$ is a minimum-weight
$\cP$-cut and $Z$ is a $\cP$-cut by definition, 
we conclude that $Y$ is an optimal solution. 
\end{proof}

Now we are ready to present the correctness proof and the analysis of
the running time of the algorithm.

\begin{theorem} \label{thm:min2linQ-fpt}
$\minlin{2}{\F}$ can be solved in in $\bigoh^*(2^{\bigoh(k \log k)})$ time.
\end{theorem}
\begin{proof}
By Lemma~\ref{lem:field-dml}, it suffices to analyze the algorithm 
for $\DML(\D)$ that was presented at the end of Section~\ref{ssec:fields-algo-def}.
Let $I=(S, w_s, X, k)$ be an arbitrary instance of this problem.
We  show that the algorithm accepts if and only if $I$ is a yes-instance. 
For the forward direction, note that if the algorithm finds a solution $Y$, 
then $\abs{Y} \leq k$ and $S-Y$ is consistent because of Step~\ref{alg:Fout} of the
algorithm.

Towards showing the reverse direction, 
suppose that $I$ is a yes-instance and
$Z$ is an optimal solution. 
Let $\varphi$ be an assignment satisfying $S - Z$, and
define $V_0$ and $V_\emptyset$ as in
Section~\ref{ssec:edom-cleaning} i.e.
$V_0 = \{ v \in V(S) \mid \varphi(v) = 0 \}$ and
$V_\emptyset = \{ v \in V(S) \mid \varphi(v) \neq 0 \}$.
Let $H_\emptyset$ denote the zero-free subgraph of $G_I$ (see Definition~\ref{def:ZFS}).
By Lemma~\ref{lem:balanced}, the
subgraph $H_\emptyset$ is balanced and connected,
and $c_G(H_\emptyset) \leq 3k + 1$.
Hence, there is an important balanced subgraph $H \in \cH$ 
considered by the algorithm in line~\ref{alg:Fout} that dominates $H_\emptyset$.
Let $F_H$ be the corresponding set of deleted edges in $G_I - \{s\}$ 
and let $\cP_Z$ be the partition of the terminals $T = V(X \cup F_H)$
into connected components of $S - (X \cup Z)$.
The algorithm exhaustively considers all possible partitions $\cP$ of $T$
and tries to compute a minimum $\cP$-cut in $S' := S - (X \cup F_H)$.
We wish to apply Lemma~\ref{lem:support} to prove that 
such a cut exists so we verify that the preconditions of the lemma are met.
By subdividing equations into three parts in the first step of the algorithm, 
we can assume without loss of generality that $X$, $F_H$ and $Z$ are
pairwise disjoint.
Further, we note that $V_\emptyset \subseteq V(H_\emptyset) \subseteq V(H)$
and all components in $S'[V(H) \setminus \{s\}]$ are flexible.
Hence, all variables in the rigid components of $S'$
are assigned zero values by $\varphi$,
the set $F_H$ is indeed a cleaning set with respect to $\varphi_Z$,
and the lemma applies. 
We conclude that the algorithm accepts the instance $I$.

We continue by analysing the time complexity of the algorithm.
Using Observation~\ref{obs:edom-rooted-enum}, the algorithm computes a
dominating family $\cH$ of $\cG$ of size at most $4^{3k+1}$
in time $\bigoh^*(4^{\bigoh(k)})$.
Let $H \in \cH$ and let $F_H$  be corresponding set of deleted edges
excluding those incident to vertex $s$. Note that $c_{G_I}(H)\leq 3k+1$.
For each $H$, every partition $\cP$ of $V(X \cup F_H)$ is computed in line 4.
Recall that $\abs{X} = k+1$ and $\abs{F_H} \leq 3k + 1$ by~Lemma~\ref{lem:balanced} 
so $\abs{V(X \cup F_H)} \leq 4k$ and the enumeration of partitions requires 
$\bigoh^*((4k)^{\bigoh(4k)})$ time. \
Computing the $\cP$-cut requires at most $\bigoh^*(2^{4k})$ time
by Lemma~\ref{lem:p-cut-fpt} and the total running time is
\[\bigoh^*(4^{\bigoh(k)})+ \bigoh^*(4^{\bigoh(k)}) \cdot \bigoh^*((4k)^{\bigoh(4k)}) \cdot \bigoh^*(2^{4k}) \in \bigoh^*(2^{\bigoh(k \cdot \log k)}).\]
\end{proof}

\subsection{Even Faster Algorithm for Finite Fields}
\label{ssec:finfields-algorithm}

Let $\F_p$ be a finite $p$-element field with $p \geq 3$. 
For $\minlin{2}{\F_2}$ a $\bigoh^*(1.977^k)$ time algorithm
can be obtained using the approach of~\cite{pilipczuk2019edge}.
Every finite field obviously has an effective representation so
we assume without loss of generality that $\F_p$ is effective.
Wedderburn's Little Theorem (see, for instance, \cite{Herstein:amm61}) 
implies that if $\D$ is a finite Euclidean domain, then $\D$ is a field.
Hence, the results in this section cover $\mintwolin$ for every finite Euclidean domain.
As mentioned in the introduction, $\minlin{2}{\F_p}$
is a special case of ULC with a finite alphabet, 
so it can be solved in $\bigoh^*(p^{2k})$ time
by the currently best algorithm for ULC~\cite{iwata2016half}.
In this section we present a faster algorithm for $\minlin{2}{\F_p}$
that runs in $\bigoh^*((2p)^k)$ time. 
By equation subdivision and Lemma~\ref{lem:field-dml}, 
the problem can be reduced to polynomially many instances of $\DML(\F_p)$.
Let $(S, w_s, X, k)$ be an instance of the latter problem.
The key to improved running time of our algorithm
is the fact that $X$ has at most $p^k$
satisfying assignments, and an optimal assignment to $S$
must extend one of these assignments.
Suppose $\alpha : V(X) \rightarrow \F_p$ is an assignment that satisfies $X$.
Then the problem can be solved by checking whether $S - X$
admits an assignment that extends $\alpha$ and leaves unsatisfied 
equations of total weight at most $k$.
A reduction to \RBGCE allows us to answer this question in $\bigoh^*(2^k)$ time.
The reader should note that this approach avoids using the method of important balanced subgraphs.

We continue with some definitions. Given an instance $S$ of $\lin{2}{\D}$, a subset of equations $X$
such that $S - X$ is consistent, and an assignment $\alpha$ satisfying $X$,
we define $S_\alpha$ as follows:
start with all equations of $S - X$, introduce two new variables $s$ and $t$, and 
add equations $x = s \cdot \alpha(x)$ of weight $k+1$ for all $x \in V(X)$ where $\alpha(x) \neq 0$,
and $x=t$ of weight $k+1$ for all $x \in V(X)$ where $\alpha(x)=0$. 
Finally, add two more variables $t'$, $t''$ and equations $t'= \gamma t$, $t''=t'$, $t=t''$ each of weight $k+1$,
where $\gamma$ is any element in $\F_p \setminus \{0, 1\}$.
We refer to $S_\alpha$ as the \emph{restriction of $S$ to $\alpha$}.
Note that $S_\alpha$ is homogeneous by construction.
Furthermore, setting $s$ to $1$ and $t$ to $0$ implies that the variables in $V(X)$
are assigned the values in accordance with $\alpha$.
Let $G_\alpha$ be the primal graph of $S_\alpha$ and 
define $\B_\alpha$ to be the family of identity cycles in $S_\alpha$.
Since all equations in $S_\alpha$ are homogeneous, we can view it
as a group-labelled graph with the group being
$\F_p^*$ i.e. the multiplicative group of the field.
Hence, we immediately obtain the following.

\begin{lemma}[\cite{Zaslavsky:jctb89}] \label{lem:minlin-fin-to-rbgce}
  $(G_\alpha, \B_\alpha)$ is a biased graph.
\end{lemma}

Clearly, there is a polynomial time algorithm that checks whether a cycle is identity
since we can multiply all labels along the cycle and check whether the result equals identity.
Now we are ready to prove the theorem.

\begin{theorem} \label{thm:min2lin-finite-fpt}
  $\minlin{2}{\F_p}$, where $\F_p$ is a finite $p$-element field with $p \geq 3$, 
  is in \FPT and solvable in $\bigoh^*((2p)^k)$ time.
\end{theorem}
\begin{proof}
By equation subdivision and Lemma~\ref{lem:field-dml}, 
we can focus on $\DML(\F_p)$.
Let $(S, w_S, k, X)$ be an instance of this problem.
Pick one variable from each equation in $X$ and 
place them into a set $U$.
Note that $\abs{U} \leq \abs{X} \leq k + 1$.
Enumerate assignments $\alpha : U \rightarrow \F_p$.
For each $\alpha$, propagate the values from the variables in
$U$ to $V(X) \setminus U$ according to the equations of $X$.
If no conflict arises, i.e. if $\alpha$ satisfies $X$, 
then construct the restriction $S_\alpha$ of $S$ to $\alpha$.
Recall that $G_\alpha$ is the primal graph of $S_\alpha$.
In the following we identify the edges of $G_\alpha$ 
and the equations of $S_\alpha$.

\begin{claim} \label{claim:direction-one}
  Suppose there exists $Z \subseteq S_\alpha$ such that
  $\sum_{e \in Z} w_S(e) \leq k$, 
  and an assignment $\varphi$ that satisfies
  $S_\alpha - Z$ and sets $\varphi(s) = 1$ and $\varphi(t) = 0$.
  Then $(G_\alpha, \B_\alpha, s, k)$ is a yes-instance of \RBGCE.
\end{claim}
\begin{claimproof}
    We claim that $Z$ is a solution for $(G_\alpha, \B_\alpha, s, k)$. Suppose for a contradiction
    that this is not the case, i.e. there is a non-identity cycle $C \notin \B_{\alpha}$ that is reachable from $s$ in $G_\alpha-Z$.
    There are two cases. If $s \notin V(C)$, then since $S_\alpha$ is homogeneous,
    $C$ is only satisfied by the all-zero 
    assignment. But $\varphi(s)=1$ and the nonzero value
    propagates to $C$, which contradicts $S_\alpha-Z$ being satisfied by $\varphi$.
    Otherwise, let $P$ be the path resulting from deleting $s$ from $C$. Let $x_1$ and $x_2$ be the endpoints of $P$.
    Let $\beta$ be the result of multiplying the edge labels of $P$ from $x_1$ to $x_2$. Since $C$ is non-identity, $\alpha(x_1) \cdot \beta \neq \alpha(x_2)$.
    Then, $P$ is a path in $S_\alpha-Z$ incompatible with setting $\varphi(x_1)=\alpha(x_1)$ and $\varphi(x_2)=\alpha(x_2)$, which is again a contradiction.
\end{claimproof}

\begin{claim} \label{claim:direction-two}
  Suppose there exists a subset $Z$ of edges of $G_\alpha$ such that 
  the connected component of the vertex $s$ in $G_\alpha - Z$ is balanced.
  Then $S_\alpha - Z$ admits a satisfying assignment $\varphi$
  such that $\varphi(s) = 1$ and $\varphi(t) = 0$.
\end{claim}
\begin{claimproof}
  Let $H$ be the connected component of $s$ in $G_\alpha-Z$. Since edges connecting $t,t',t''$ 
  form a high-weight non-identity cycle, $t \notin V(H)$.
  Since every cycle of $H$ is identity, $H$ viewed as a subset of $S_\alpha$ is flexible. 
  Then in particular there is a satisfying
  assignment $\varphi$ to $H$ setting $\varphi(s)=1$ by Corollary~\ref{cor:flexible-propagate}. 
  For every variable $v$ not in $H$
  we can safely set $\varphi(v)=0$ since $S_\alpha$ is homogeneous.
  Then $\varphi$ satisfies $S_\alpha-Z$ and sets $\varphi(s) = 1$ and $\varphi(t) = 0$.
\end{claimproof}

Together, Claims~\ref{claim:direction-one}~and~\ref{claim:direction-two}
imply that the assignment $\alpha$ can be extended to $S$ so that it
leaves equations of total weight at most $k$ unsatisfied if and only if
$(G_\alpha, \B_\alpha, s, k)$ is a yes-instance of \RBGCE.
Note that the algorithm considers all satisfying assignments to $X$
so by exhaustion $(S, w_S, k, X)$ is a yes-instance
if and only if the algorithm finds a suitable assignment $\alpha$.
There are $p^{k+1}$ candidates for $\alpha$.
Computing the instance $(G_\alpha, \B_\alpha, s, k)$
requires polynomial time and \RBGCE 
can be solved in $\bigoh^*(2^k)$ time by 
Proposition~\ref{pro:rbgce}
so the total running time is $\bigoh^*((2p)^k)$.
\end{proof}

\section{Hardness Results}
\label{sec:hardness}

Let $\D=(D;+,\cdot)$ be a commutative ring.
The reduction from \textsc{Multicut} presented in the introduction
shows that $\minlin{r}{\D}$ is NP-hard (for $r \geq 2$), and it rules out the possibility that our fixed-parameter tractable algorithms can be improved to
polynomial-time algorithms.
In Section~\ref{sec:w-hardness}, we show $\Weft[1]$-hardness for $r \geq 3$ whenever
$(D;+)$ is an abelian group with at least
two elements. This result consequently
covers all (commutative and non-commutative) rings except the trivial
zero ring.
We continue in Section~\ref{sec:non-integral} by studying $\minlin{2}{\D}$ for commutative rings $\D$ that contain 
a zero divisor (i.e. an element $\alpha \neq 0$ such that there exists an
element
$\beta \neq 0$ and $\alpha \cdot \beta = 0$). We show that  $\minlin{2}{\D}$ is \W{1}-hard
for many such structures.
We note that hardness results for certain special cases have appeared
earlier in the literature---for instance, Crowston et al.~\cite{crowston2013parameterized} prove \W{1}-hardness for $\minlin{3}{\F_2}$.

For proving the hardness results, we use \emph{parameterized reductions} (or fpt-reductions).
Consider two parameterized problems $L_1, L_2 \subseteq \Sigma^* \times \naturals$.
A mapping $P: \Sigma^* \times \naturals \rightarrow \Sigma^* \times \naturals$
is a parameterized reduction from $L_1$ to $L_2$ if

\begin{enumerate}[(1)]
  \item $(x, k) \in  L_1$ if and only if $P((x, k)) \in L_2$, 
  \item the mapping can be computed in $f(k) \cdot n^{O(1)}$ time for some computable function $f$, and 
  \item there is a computable function $g : \naturals \rightarrow \naturals$  
  such that for all $(x,k) \in \Sigma^* \times \naturals$, if $(x', k') = P((x, k))$, then $k' \leq g(k)$.
\end{enumerate}

The class $\Weft[1]$ contains all problems that are fpt-reducible to \textsc{Independent Set} 
parameterized by the solution size, i.e. the number of vertices in the independent set.
Showing $\Weft[1]$-hardness (by an fpt-reduction) for a problem rules out the existence of an fpt
algorithm under the standard assumption that $\FPT \neq \Weft[1]$.

\subsection{Three Variables per Equation}
\label{sec:w-hardness}

Let
$\G=(D;+)$ denote an arbitrary abelian group.
An expression $x_1+\dots+x_r=c$ is an {\em equation over} $\G$ if
$c \in D$ and $x_1,\dots,x_r$ are either variables or inverted variables with domain $D$.
We say that it is an {\em $r$-variable equation} if it contains at most $r$
distinct variables. We consider the following group-based variant of
the $\minlin{r}{\D}$ problem.

\pbDefP{$\minlin{r}{\G}$}
{A system $S$ of equations over $\G$, a weight function 
$w : S \rightarrow \naturals^+$, and an integer $k$.}
{$k$.}
{Is there a set $Z \subseteq S$ such that $S - Z$ is consistent and $w(Z) \leq k$?}

The crux of the proof is essentially the same as the \W{1}-hardness proof
for \textsc{Odd Set} presented in Theorem 13.31 of~\cite[Section 13.6.3]{cygan2015book},
however many details are different. 
The reduction is based on the following \W{1}-hard 
problem~\cite[Lemma 1]{Fellows:etal:tcs2009}.

\pbDefP{Multicoloured Clique}
{A graph $G = (V,E)$ with vertices partitioned into $k$ colour classes $V_1, \dots, V_k$.}
{$k$.}
{Does $G$ contain a clique with exactly 
one vertex from each $V_i$, $1 \leq i \leq k$?}

\begin{theorem} \label{thm:min3lin-is-w1hard}
Let $\G=(D;+)$ denote a group with at least two elements. Then,
$\minlin{r}{\G}$ is \W{1}-hard for any $r \geq 3$
even if all equations have weight 1.
\end{theorem}
\begin{proof}

The reduction is presented in two steps:
given an arbitrary instance $(G,k,(V_1,\dots,V_k))$ of {\sc Multicoloured Clique}, we first
compute an instance $(S,w,k')$ of $\minlin{s}{\G}$ where $s = |V(G)|+|E(G)|$, and
then we transform this instance into an instance of  $\minlin{3}{\G}$ with unit weights.
We let $0$ denote the identity element in $\G$ and let $1$ be any non-identity element.

\smallskip

\noindent
{\em Step 1.}
Consider the arbitrarily chosen instance $(G,k,(V_1,\dots,V_k))$ of {\sc Multicoloured Clique}.
We will now reduce it to an instance of $\minlin{s}{\G}$.
We let $E_{ij}$ denote the set of edges in $E(G)$ with one endpoint in $V_i$ and another in $V_j$, and
we let $E_{ijv}$ be the subset of $E_{ij}$ containing all edges incident to a vertex $v$.
We define an instance $(S,w,k')$ of $\minlin{s}{\G}$ as follows.
Introduce variables $x_v$ for all $v \in V(G)$ and $y_e$ for all $e \in E(G)$.
Set the parameter $k' = k + \binom{k}{2}$.
Let $S$ contain the following equations:
\begin{enumerate}[(1)]
  \item \label{case:31} $x_v = 0$ for all $v \in V(G)$.
  \item \label{case:32} $y_e = 0$ for all $e \in E(G)$.
  \item \label{case:33} $\sum_{v \in V_i} x_v = 1$ for all $1 \leq i \leq k$.
  \item \label{case:34} $\sum_{e \in E_{ij}} y_e = 1$ for all $1 \leq i < j \leq k$.
  \item \label{case:35} $\sum_{u \in V_j \setminus \{v\}} x_u + \sum_{e \in E_{ijv}} y_e = 1$
  for all $v \in V(G)$.
\end{enumerate}
The equations in~\eqref{case:33}--\eqref{case:35} are assigned weight $k+1$, while all others
are given unit weight.
Thus, only equations in~\eqref{case:31}~and~\eqref{case:32} may appear in a solution to $(S,w,k')$.
Observe that the equations in~\eqref{case:31}--\eqref{case:34} imply that
exactly one variable in $\{x_v \; | \; v \in V_i\}$ for each $1 \leq i \leq k$ and
exactly one variable in $\{y_e \; | \; e \in E_{ij}\}$ for each pair $1 \leq i < j \leq k$
may be set to $1$ since the budget $k + \binom{k}{2}$ is tight.

Now consider the equations in~\eqref{case:35}.
Intuitively, for any variable $v \in V_i$, 
the corresponding equation implies that either $x_v$ is set to $0$
or at least one $y_e$ for an edge $e$ incident to $v$ is set to $1$.
Formally, let $\varphi$ be an assignment that does not satisfy 
$k'$ constraints in $S$.
If $\varphi(x_v) = 1$, then $\varphi(x_u) = 0$ for all $u \in V_i \setminus \{v\}$.
Hence, $\sum_{u \in V_j \setminus \{v\}} \varphi(x_u) = 0$
and $\varphi(y_{\{v,w\}}) = 1$ for some $w \in V_j$.
Moreover, $\varphi(y_e) = 0$ for all edges $e \in E_{ij} \setminus \{v,w\}$.
In the equation for $w$ in~\eqref{case:35} we have
$\sum_{e \in E_{ijw}} y_e = 1$.
Hence, $\varphi(u) = 0$ for all $u \in {V_j \setminus \{w\}}$, and $\varphi(w) = 1$.
On the other hand, if $\varphi(x_v) = 0$, then
there is exactly one $u \in V_i \setminus \{v\}$ such that $\varphi(x_u) = 1$ so
$\varphi(y_e) = 0$ for all edges $e \in E(G)$ incident to $v$.
We conclude that the reduction is correct
and it can clearly be
carried out in polynomial time.

\smallskip

\noindent
{\em Step 2.}
We continue by transforming the instance $(S,w,k')$ into an instance
of $\minlin{3}{\G}$ with unit weights.
Consider an equation $\sum_{i=1}^{r} v_i = 1$ in $S$.
We first show how to make it
undeletable without assigning it the weight $k+1$.
To this end, introduce variables $v_i^{(j)}$ 
for all $1 \leq i \leq r$ and $1 \leq j \leq k+2$.
Create a system of equations $S'$ by adding equations
$\sum_{i=1}^{r} v_i^{(j)} = 1$ for all $j$
and equations $v_i^{(j)} - v_i^{(j')} = 0$ for all $i$ and all $j < j'$.
We claim that any assignment that does not set $v_i^{(1)},\dots,v_i^{(k+2)}$
to the same value
does not satisfy at least $k+1$ constraints.
Suppose an assignment sets $\ell$ copies of the variable
to one value and $k + 2 - \ell$ remaining copies to another.
Then at least $(k + 2 - \ell) \ell$ equations are not satisfied by the assignment.
For $1 \leq \ell \leq k$, this quantity is minimized by $\ell = 1$ and it equals $k + 1$.
Thus, any assignment that does not satisfy at most 
$k$ constraints also satisfies $S'$.

Finally, we show how to reduce 
the number of variables in each equation 
to at most three.
Again, consider an equation of the form $\sum_{i=1}^{r} v_i = 1$.
Introduce auxiliary variables $a_i$ for 
$i \in \range{1}{r}$ and replace
the equation with the following system:
\begin{equation*}
  \begin{cases}
    v_1 + (-a_i) = 0, & \\
    a_i + v_{i+1} + (-a_{i+1}) = 0 & \text{for } i \in \range{1}{r-1} \\
    a_{r} + v_{r} = 1 &
  \end{cases}
\end{equation*}
where each equations is given weight 1.
Observe that the sum of all equations above telescopes and
the auxiliary variables cancel out,
leaving exactly the equation $\sum_{i=1}^{r} v_i = 1$.
Hence, an assignment that satisfies all equations
in the system also satisfies the original equation.
Moreover, any assignment $\varphi$ 
that does not satisfy the original equation 
can be extended to the auxiliary variables to satisfy all but one equation by setting
$\varphi(v_1) = \varphi(a_1)$ and
$\varphi(v_{i+1}) = \varphi(a_i) + \varphi(v_{i})$
for all $i \in \range{1}{r-1}$.
Hence, replacing every long equation
in this way reduces the initial instance to an instance of
$\minlin{3}{\G}$ with unit weights.
\end{proof} 

\subsection{Rings with Zero Divisors}
\label{sec:non-integral}

Recall that a Euclidean domain cannot contain a zero divisor.
Next, we give examples of commutative rings $\K$ with zero divisors such that
\minlin{2}{\K} is \W{1}-hard. Our starting point is the
following problem, which has previously been used as a source of
\W{1}-hardness for \textsc{MinCSP} problems~\cite{KimKPW21flow,marx2009constant}.

\pbDefP{Paired Min Cut}{A graph $G$, vertices $s, t \in V(G)$,
and an integer $k$, where the $st$-max flow in $G$ is $2k$;
a set of disjoint edge pairs $C \subseteq \binom{E(G)}{2}$}
{$k$}
{Is there an $st$-mincut $X \subseteq E$ 
which is the union of $k$ pairs from $C$?} 

We will consider a restricted variant of {\sc Paired Min Cut} in the following.
We say that an instance of \textsc{Paired Min Cut} is \emph{split} if the
following statements hold.
\begin{enumerate}
  \item There are two induced subgraphs $G_1=G[U_1]$ and $G_2=G[U_2]$ of $G$
  such that $U_1 \cup U_2 = V(G)$, $U_1 \cap U_2 = \{s,t\}$
  and $G - \{s,t\}$ is the disjoint union of $G_1 - \{s,t\}$ and $G_2 - \{s,t\}$
  \item For every pair $\{e_1,e_2\} \in C$, 
  one edge lies in $G_1$ and the other lies in $G_2$
\end{enumerate}

\begin{lemma} \label{lm:graph-cut-split}
  \textsc{Paired Min Cut} is \W{1}-hard, even for split instances.
\end{lemma}
\begin{proof}
  It is well known that \textsc{Paired Min Cut} is \W{1}-hard in its standard
  form~\cite{KimKPW21flow,marx2009constant}. We show that we can also
  impose the split property. Thus, let $I=(G, s, t, k, C)$ be an arbitrary
  instance of \textsc{Paired Min Cut}. We construct an instance
  $I'=(G',s,t,k',C')$ of \textsc{Paired Min Cut} where $I'$ is split and $k'=4k$.

  Create two graphs $G_1$ and $G_2$ on disjoint vertex sets, each a
  copy of $G$, and let $G'$ be their union. For every edge $e=\{u,v\}$ in $G'$,
  introduce a new vertex $x_e$ and the two edges $e'=\{u,x_e\}$ and $e''=\{x_e,v\}$.
  For an edge or vertex $z$ of~$G$ and $i \in \{1,2\}$, let $z_i$ denote the copy of $z$ in $G_i$.
  For every pair $p=\{e,f\}$ in $C$, place the four pairs
  \[
 \{e_1,e_2\}, \{e_1', f_2\}, \{f_1,e_2'\}, \{f_1', f_2'\}
  \]
  in $C'$ (thereby keeping the pairs in $C'$ disjoint). 
  Finally, identify $s_1$ with $s_2$ as $s$
  and $t_1$ with $t_2$ as $t$. This finishes the description of our
  output $I'=(G', s,t,k',C')$. Note that $G'$ is split, and that
  the $st$-max flow in $G'$ is $8k=2k'$.

We show that $I$ is a yes-instance if and only if $I'$ is a yes-instance.
  First, let $X \subseteq E(G)$ be a solution to $I$.
  Let $X'=\{e_1,e_1',e_2,e_2' \mid e \in X\}$.
  Then~$X'$ is the union of precisely four pairs for every pair in
  $X$, and it is clear that~$X'$ is an $st$-cut. 

  On the other hand, assume that $I'$ has a solution $X'=X_1' \cup X_2'$
  (where $X_i' \subseteq E(G_i)$, $i \in \{1, 2\}$).  We claim that $X_1'$
  and $X_2'$ represent the same edge set~$X$ in $G$.
  By assumption, $X'$ partitions into edge pairs, and since
  the $st$-max flow in $G$ is $2k'$, $X'$ must be an $st$-min cut.
  In particular, by the structure of the pairs, for every $e \in E(G)$, $X'$ contains $e$ if and only
  if it contains $e'$, and therefore also the other endpoint
  of the pair $p' \in C'$ containing the respective edge. 
  Hence for every edge $e$ represented in $X'$
  there must be a pair $\{e,f\} \in C$
  such that all four pairs $\{e_1,f_1\} \times \{e_2,f_2\}$
  are represented in $X'$. Hence
  \[
    X=\{e \in E(G) \mid \{e_1,e_1',e_2,e_2'\} \subseteq X'\}.
  \]
  defines a set of $2k$ edges in $G$, which partitions into pairs from $C$.
  Furthermore, $X$ is an $st$-cut, since $X_1'$ is an $st$-cut in $G_1$
  and $G_1$ was created as a copy of~$G$. 
\end{proof}

We now show a general \W{1}-hardness result for \minlin{2}{\K}.  
 
\begin{figure}
 \centering
  \begin{tikzpicture}
    \coordinate (u1) at (0,3);
    \coordinate (v1) at (4,3);
    \coordinate (u2) at (0,0);
    \coordinate (v2) at (4,0);
    \coordinate (xp) at (1,3/2);
    \coordinate (yp) at (3,3/2);

    \filldraw[black] (u1) circle (2pt) node[anchor=east]{$u_1$};
    \filldraw[black] (v1) circle (2pt) node[anchor=west]{$v_1$};
    \filldraw[black] (u2) circle (2pt) node[anchor=east]{$u_2$};
    \filldraw[black] (v2) circle (2pt) node[anchor=west]{$v_2$};
    \filldraw[black] (xp) circle (2pt) node[anchor=east]{$x_p$};
    \filldraw[black] (yp) circle (2pt) node[anchor=west]{$y_p$};

    \draw[dashed] (u1) -- (v1);
    \draw[dashed] (u2) -- (v2);
    \draw[thick] (u1) -- (xp) node[midway, left,  color=gray] {$\alpha_1$};
    \draw[thick] (u2) -- (xp) node[midway, left,  color=gray] {$\alpha_2$};
    \draw[thick] (xp) -- (yp) node[midway, above, color=gray] {$=$};
    \draw[thick] (yp) -- (v1) node[midway, right, color=gray] {$\alpha_1$};
    \draw[thick] (yp) -- (v2) node[midway, right, color=gray] {$\alpha_2$};
  \end{tikzpicture}
 \caption{System of equations obtained from a pair of edges $p = \{e_1, e_2\}$ 
 where $e_i = \{u_i, v_i\}$ in the fpt reduction from \textsc{Paired Min Cut}.
 Edges $e_1$ and $e_2$ are illustrated by dashed lines, while the equations
 are illustrated by solid lines with labels describing 
 equations between connected variables.}
 \label{fig:gadget}
\end{figure}
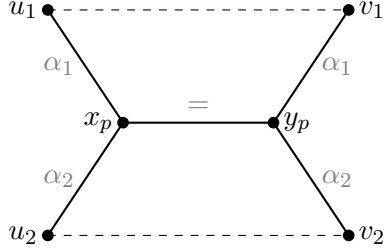
 
\begin{theorem}
\label{thm:non-integral-hard}
  Let $\K=(K;+, \cdot)$ be a commutative ring with additive neutral element 0.
  If there are elements $\alpha_1, \alpha_2 \in K$
  such that $\alpha_1^2 \neq 0$, $\alpha_2^2 \neq 0$,
  and $\alpha_1 \cdot \alpha_2 = 0$,
  then \minlin{2}{\K} is \W{1}-hard. 
 \end{theorem}
\begin{proof}
  We reduce from an arbitrary split instance of \textsc{Paired Min Cut}. 
  Let $I=(G, s, t, k, C)$ be the input instance
  and let $G=G_1 \cup G_2=G[U_1] \cup G[U_2]$ form the split.
  Divide the source $s$ into two vertices $s_1$, $s_2$
  where $N_G(\{s_i\}) = N_G(\{s\}) \cap U_i$
  for $i\in\{1, 2\}$, but keep
  the sink $t$ intact. We compute an instance of \minlin{2}{\K} as follows.
  Introduce one variable for every vertex in the resulting graph
  and initially turn every edge $\{u,v\}$ into an equation $u=v$ of
  weight $k+1$.
  Force $s_1=\alpha_1$, $s_2=\alpha_2$ and $t=0$ by equations
  of weight $k+1$ each. Finally, for every pair $p \in C$,
  say $p=\{e_1,e_2\}$ where $e_i \subseteq U_i$ for $i\in\{1,2\}$, 
  do the following. Remove the equations corresponding to the
  edges in $p$. Create two new variables $x_p$, $y_p$.
  Let $e_1=\{u_1,v_1\}$ and $e_2=\{u_2,v_2\}$.
  Create equations
  \begin{align*}
      \alpha_1 \cdot u_1 &= \alpha_1 \cdot x_p, \\
      \alpha_2 \cdot u_2 &= \alpha_2 \cdot x_p, \\
      \alpha_1 \cdot y_p &= \alpha_1 \cdot v_1, \\
      \alpha_2 \cdot y_p &= \alpha_2 \cdot v_2  
  \end{align*}
  of weight $k+1$ each,
  and an equation $x_p = y_p$ of weight 1. 
  See Figure~\ref{fig:gadget} for an illustration.
  Perform this for every
  pair in $C$. Let $S$ be the resulting set of equations.
  We claim that $S$ has a solution of cost at most $k$ if and only if
  $I$ is a yes-instance.

  On the one hand, let $X \subseteq E(G)$
  be the union of $k$ pairs of edges,
  and let $X'=\{x_p=y_p \mid p \in C \; {\rm and} \; p \subseteq X\}$.
  We claim that $S - X'$
  is satisfiable. For a vertex $v \in U_i$, set $v=\alpha_i$ if $v$ is reachable
  from $s$ in $G-X$, and $v=0$ otherwise. Consider a pair
  of vertices $x_p, y_p$ for $p \in C$, and suppose that the
  assignment above cannot be consistently extended to $x_p$
  and $y_p$. 
  Then this implies that there is an edge
  $e_i=\{u,v\} \in p$ such that $u, v \in U_i$ and $\alpha_i \cdot u \neq \alpha_i \cdot v$.
  Since the value assigned to $u$ and $v$ is 
  either $0$ or $\alpha_i$,
  we have $u \neq v$.
  This implies that $e_i$ crosses the cut in $G-X$, contradicting our
  assumption that $e_i \notin X$. Hence, $S - X'$ is
  satisfiable. 
    
  On the other hand, suppose that there is a solution where equations
  of cost at most $k$ are not satisfied, and let $X'$ be the set of these
  equations. Then clearly every equation in $X'$ is
  of the form $x_p=y_p$ for some pair $p \in C$. Let $X \subseteq C$
  be the union of edges participating in these pairs. 
  We claim that $X$ is an $st$-cut. Assume to the contrary that $G-X$
  contains a path $P$ from $s$ to $t$. Then that path corresponds 
  to a chain of equations in $S$, from $s_i$ to $t$ ($i \in \{1,2\}$),
  where every edge $\{u,v\}$ of the path corresponds to either
  an equation $u=v$ or a chain of equations $\alpha_i u = \alpha_i x$, 
  $x=y$, $\alpha_i y = \alpha_i v$, where every equation in the chain is satisfied.
  Since $\alpha_i^2 \neq 0$, we have $\alpha_i \cdot s_i \neq \alpha_i \cdot 0$ so 
  every variable on the path is assigned
  a non-zero value.
  This contradicts that $S - X'$ is
  satisfiable since $t = 0$.
\end{proof}
 
 We illustrate Theorem~\ref{thm:non-integral-hard} with an
 example.
 The {\em direct product} of two rings 
 $\K_1=(K_1;+_1,\cdot_1)$ and $\K_2=(K_2;+_2,\cdot_2)$ 
 is denoted $\K_1 \times \K_2 = (K; +, \cdot)$.
 Its domain $R$ consists of the ordered pairs 
 $\{(d_1,d_2) \mid d_1 \in K_1, d_2 \in K_2\}$ 
 and the operations are defined coordinate-wise:
 $(d_1,d_2)+(d'_1,d'_2)=(d_1+_1 d'_1,d_2+_2 d'_2)$ and
 $(d_1,d_2) \cdot (d'_1,d'_2)=(d_1 \cdot_1 d'_1,d_2 \cdot_2 d'_2)$.
 We claim that whenever $\K = \K_1 \times \K_2$
 and $\K_1, \K_2$ are commutative rings that are not zero rings, then
 \minlin{2}{\K} is \W{1}-hard. 
 To see this, let $0_1 \in K_1,0_2 \in K_2$ denote the additive identities and $1_1 \in K_1,1_2 \in K_2$ denote the
 multiplicative identities.
 By setting $\alpha_1=(0_1,1_2)$ and $\alpha_2=(1_1,0_2)$,
 Theorem~\ref{thm:non-integral-hard} is applicable and
 \minlin{2}{\K} is \W{1}-hard. 
 This argument can easily be extended to products of several commutative rings.
 The ring $\integers/m \integers$ (i.e. the ring based on standard arithmetic modulo $m$) is isomorphic to a direct product of non-trivial commutative rings whenever
 $m$ is not a prime power. For example, $\integers / 6\integers \cong \integers / 2\integers \times \integers / 3\integers$. Hence, $\minlin{2}{\integers/6 \integers}$ and more generally $\minlin{2}{\integers/m \integers}$ where $m$ is not a prime power is \W{1}-hard.

\section{Conclusions and Discussion}
\label{sec:conclusion-and-discussion}

We have proved that $\minlin{2}{\D}$
is fixed-parameter tractable (with parameter $k$ being the number of unsatisfied
equations) when $\D$ is an efficient Euclidean domain.
We additionally proved that $\minlin{r}{\D}$
is \W{1}-hard when $r \geq 3$ and this result
holds for all rings.
Furthermore, we demonstrated that there exist commutative rings
$\K$ (that are not Euclidean domains) such that $\minlin{2}{\K}$
is \W{1}-hard.

The borderline between fixed-parameter tractable and \W{1}-hard
$\mintwolin$ problems is not clear, 
and this is true even for finite commutative rings. 
Wedderburn's Little Theorem (see, for instance, \cite{Herstein:amm61} for a proof) states that if $\K$ is a finite ring, then either 
(1) $\K$ is a field (and $\minlin{2}{\K}$ is in \FPT) or 
(2) $\K$ contains zero divisors. 
We know that there are $\K$ with zero divisors such that $\minlin{2}{\K}$ is \W{1}-hard, but it is an open question whether the problem is 
always \W{1}-hard when $\K$ contains zero divisors, even in the finite case.
A concrete question is the following: what is the parameterized 
complexity of $\minlin{2}{\integers/4 \integers}$ or more generally
$\minlin{2}{\integers/p^n \integers}$ where $p$ is a prime and $n \geq 2$?
Resolving these cases would give us a complete understanding of $\minlin{2}{\integers/m \integers}$
for every $m$.
However, there are still many
open cases left, even for small commutative rings.
A noticeable example is the four-element commutative ring  $\F_2[x]/(x^2+x)$
whose elements can be viewed as arrays
\[\left( \begin{array}{cc} x & 0 \\ y & x \end{array} \right)\]
with $x,y \in {\mathbb F}_2$.

We suspect that our fixed-parameter tractable algorithm for \mintwolin\
over Euclidean domains can be improved with respect to running time.
The slowest part in it is solving the {\sc Pair Partition Cut} problem.
We solve this problem via a reduction to a finite-domain {\sc MinCSP} problem
that is solved by flow augmentation, but there may be alternative ways of doing this.
However, as the problem is a strict generalization of \textsc{(Edge) Multicut},
a running time of, say, $\bigoh^*(2^{\bigoh(k \log k)})$ or better would be a significant challenge.
There is also room for improvements in 
the $\mintwolin$ algorithms for fields.
Consider our $\bigoh^*(2^{\bigoh(k \log k)})$ time algorithm
for arbitrary fields.
After iterative compression and cleaning,
the problem reduces to the following:
\pbDefP{$\lin{2}{\F}$ Compatibility }
{Two instances $S_1,S_2$ of $\lin{2}{\F}$ and an integer $k$ such that
$V(S_1) \subseteq V(S_2)$,
$\abs{S_1} \leq 3k$, and 
$S_2$ only contains equalities.}
{$k$.}
{Is there a set $Z \subseteq S_2$ such that $\abs{Z} \leq k$ and
$(S_1 \cup S_2) - Z$ is consistent?}
This problem is the bottleneck for our $\minlin{2}{\F}$ algorithm, since
it is the only part that requires more than single-exponential time.
Can it be solved in single-exponential time in $k$?
For the finite field $\F_p$ with $p$ elements
we show that $\minlin{2}{\F_p}$ can be solved in $\bigoh^*((2p)^k)$ time.
Is there an $\bigoh^*(c^k)$ algorithm for $\minlin{2}{\F_p}$, where
$c$ is a universal constant that does not depend on $p$?
Or is there at least a constant $d < 2$
such that $\minlin{2}{\F_p}$ is solvable in $\bigoh^*((dp)^k)$ time?

Another more general question is about the utility of the method of important balanced subgraphs. 
Important separators
are a key component of many classical fpt algorithms for graph separation problems, and important balanced
subgraphs appear to be a significant, and unexpected, generalization of them. It would be interesting
to see more applications of the method. 
We have used it to avoid random sampling of important separators, speeding up our $\mintwolin$ algorithm for fields from $\bigoh^*(2^{\bigoh(k^3)})$ to $\bigoh^*(2^{\bigoh(k \log k)})$, and simplifying the algorithm for Euclidean domains. What other problems can be solved using this method? Can we use it to obtain simpler algorithms or improve upper bounds for other parameterized deletion problems?
Other questions include generalizations or improvements on the method of important balanced subgraphs itself. We have provided the result only for edge deletion problems; is there an equivalent statement for vertex deletion?
Furthermore, the polynomial factor in the running time of the algorithm producing a dominating family is significant, since it comes
from solving an LP given only oracle access to the constraints. However, the optima computed by the LP
are extremal half-integral solutions with an inherent structure that can probably be exploited. A combinatorial method
for computing such optima could substantially improve the polynomial factor. 
Such a result was developed in the algorithm for 0/1/all CSPs by Iwata et al.~\cite{iwata201801all}, where a previous method based on half-integral LP-relaxations was replaced by a linear-time combinatorial solver.
Can a similar method be developed for the \textsc{Rooted Biased Graph Cleaning} problem, perhaps for special cases such as biased graphs coming from group-labelled graphs or
the biased graphs used in the \minlin{2}{\D} algorithm? 

\section*{Acknowledgements}

The second and the fourth authors were supported by
the Wallenberg AI, Autonomous Systems and Software Program (WASP) funded
by the Knut and Alice Wallenberg Foundation. In addition, the second
author was partially supported by the Swedish Research Council (VR)
under grants 2017-04112 and 2021-04371. The first and third authors acknowledge support from the
Engineering and Physical Sciences Research Council (EPSRC, project EP/V00252X/1).


\end{document}